\documentclass[11pt]{article}
\usepackage{amsmath,amsthm,amsfonts,amssymb, graphicx,mathabx,epsfig}
\usepackage{bbm}
\usepackage{authblk}
\usepackage{bm}
\usepackage{graphicx,color}
\usepackage{footmisc}
\usepackage{epstopdf}
\usepackage{dsfont}
\usepackage{pgfplots}
\usepackage{subfig}
\usetikzlibrary{fit}
\usepackage{slashed}
\usepackage{hyperref}

\numberwithin{equation}{section}

\newtheorem{thm}{Theorem}[section]

\newtheorem{prop}[thm]{Proposition}

\newtheorem{rem}[thm]{Remark}
\newtheorem{assu}[thm]{Assumption}

\def\Tr{{\rm{Tr}}}
\def\tr{{\rm{Tr}}}
\def\xx{{\bf x}}
\def\yy{ {\bf y}}
\def\zz{{\bf z}}
\let\e=\varepsilon
\def \be{\begin{equation}}
\def \ee{\end{equation}}
\def\d{\delta}
\def \bea{\begin{eqnarray}}
\def \eea{\end{eqnarray}}
\def\m{\mu}
\def\b{\beta}
\def\nn{\nonumber}
\def\G{\Gamma}
\def\l{\lambda}

\begin{document}

\title{Multi-Channel Luttinger Liquids at the Edge of Quantum Hall Systems}

\author[1]{Vieri Mastropietro}
\affil[1]{University of Milano, Department of Mathematics ``F. Enriquez'', Via C. Saldini 50, 20133 Milano, Italy}
\author[2]{Marcello Porta\footnote{e-mail address: mporta@sissa.it}}
\affil[2]{SISSA, Via Bonomea 265, 34136 Trieste, Italy}

\maketitle

\begin{abstract} 
We consider the edge transport properties of a generic class of interacting quantum Hall systems on a cylinder, in the infinite volume and zero temperature limit. We prove that the large-scale behavior of the edge correlation functions is effectively described by the multi-channel Luttinger model. In particular, we prove that the edge conductance is universal, and equal to the sum of the chiralities of the non-interacting edge modes. The proof is based on rigorous renormalization group methods, that allow to fully take into account the effect of backscattering at the edge. Universality arises as a consequence of the integrability of the emergent multi-channel Luttinger liquid combined with lattice Ward identities for the microscopic $2d$ theory.
\end{abstract}
\maketitle

\tableofcontents     

\section{Introduction}

\noindent{\bf Quantum Hall Effect.} The integer quantum Hall effect (IQHE) is a paradigmatic example of topological quantum phenomenon in condensed matter physics: the transverse, or Hall, conductivity of two-dimensional insulating systems exposed to magnetic fields is equal to an integer multiple of $e^{2}/h$, where $e$ is the electric charge and $h$ is the Planck's constant. Despite the complexity of the system at the microscopic scale, the IQHE is measured with astonishing precision. For translation invariant, noninteracting models, the explanation of the IQHE has been given in \cite{TKNN}, see also \cite{AS2}, identifying the Hall conductivity as the Chern number of a suitable vector bundle. If translation invariance is broken, quantization can be understood via a noncommutative-geometric framework \cite{B1}, or in a more functional-analytic setting \cite{AS2index, AG}. In all these approaches, the Hall conductivity is identified as a topological invariant; universality against perturbations is understood via the stability of an index.

The quantum Hall effect admits dual descriptions, in terms of either bulk or edge transport properties \cite{Halp}. The infinite volume, or bulk, Hall conductivity turns out to be equal to the edge conductance, describing the charge transport along the boundary of the system. This remarkable fact is called the {\it bulk-edge correspondence.} In the absence of interactions, the edge conductance is equal to the signed number of gapless edge modes, the sign taking into account the chirality (or direction of propagation) of the edge states.

For noninteracting quantum Hall systems, the bulk-edge correspondence has been proved in \cite{Hat} for translation invariant systems, in \cite{SKR, EG} for disordered systems with a spectral gap, and in \cite{EGS} for disordered systems with a mobility gap. The positive temperature extension has been recently discussed in \cite{CMT}. A general field theoretic approach to the bulk-edge duality, based on anomaly cancellations, has been introduced in \cite{FK, FST}.

Concerning many-body systems, the quest for the rigorous understanding of the quantization of the Hall conductivity has been opened for a while. Recently, two different approaches have been proposed; see \cite{A} for a review of the problem, and of its solution. In \cite{HM}, see also \cite{BBdRF}, the quantization of the interacting Hall conductivity has been proved via geometric methods, following earlier insights by \cite{ASint}. A key assumption is the spectral gap for the interacting fermionic model, proved later in \cite{HM1, dRS} for weak interactions. Afterwards, an index theorem for gapped many-body systems has been introduced in \cite{BBdRF2}. A different approach has been introduced in \cite{GMP}. There, quantization follows from universality, by showing that the interacting Hall conductivity is equal to the non-interacting one. The proof holds under the assumption that the non-interacting Hamiltonian is gapped, and that the many-body interaction is weak enough. It is based on cluster expansion methods, to construct the interacting Gibbs state, and on Ward identities, to prove universality. The same cluster expansion techniques have then been used to prove the stability of the spectral gap \cite{dRS}. More recently, in the context of the interacting Haldane model, the approach of \cite{GMP} has been combined with renormalization group (RG) methods, to obtain results that hold uniformly in the size of the spectral gap \cite{GJMP, GMPcritical}.

Much less is known for interacting edge transport. One expects the large-scale properties of the edge currents to be effectively described by the multi-channel Luttinger model \cite{W}, a relativistic theory for interacting fermions in $1+1$ dimensions; see \cite{KF, C} for reviews. It generalizes the Luttinger model, describing interacting fermions with two opposite chiralities, to an arbitrary number of chiral modes. The Luttinger model and its variants can be solved by bosonization, as first pointed out in the pioneering work \cite{ML}, see \cite{MM} for a review. 

A striking feature of the exact solution of the Luttinger model is the appearence of anomalous exponents in the scaling of correlations. These are determined by the so-called Luttinger parameter $K\equiv K(\lambda)$, a nontrivial function of the coupling constant $\lambda$, such that $K(0) \equiv 1$. Besides fixing the anomalous exponents, the Luttinger parameter also allows to determine the transport properties of the system. In fact, the quantum of conductance due to a single channel is $K(\lambda) e^{2} / h$. More generally, in the multi-channel Luttinger model every chiral channel comes with its own Luttinger parameter, and their values strongly depend on whether backscattering is present. In the special case in which all chiral fermions propagate in the same direction (that is, no backscattering is present), the conductance of a single chiral mode is universal, and equal to $e^{2} / h$. The Landauer conductance, describing the current appearing in the system after putting it in contact with reservoirs at different chemical potentials, turns out to be equal to the sum of the absolute values of the conductances of all channels. See \cite{Gi} for a computation of the Landauer conductance of the Luttinger model based on Kubo formula, or \cite{LLMM} for a recent derivation from quantum dynamics.

The Landauer conductance is related to the two-terminal conductance of the Hall bar, see e.g. \cite{KF, KFP}. In the absence of impurities, one expects this transport coefficient to be universal only in the absence of backscattering between edge modes. Notice however that counterpropagating edge modes appear in a wide class of physically relevant systems, see e.g. \cite{Fukui} for the case of the Hofstadter model. Instead, the edge conductance of quantum Hall systems, defined as response of the edge modes after a variation of the chemical potential at the edge, is related to the sum of the signed conductances of the edge channels, where the sign takes into account the direction of propagation. This quantity is expected to be stable against backscattering and disorder effects, and quantized.

The multi-channel Luttinger model is not a fundamental description of the system: it is an emergent theory, whose validity needs to be proved. It depends on a number of free parameters, which need to be fixed in order to provide a quantitatively correct description of the edge physics. It is a key ingredient of a phenomenological theory of the fractional quantum Hall effect \cite{W, KF, KFP}, where counterpropagating edge modes play an important role. There, the Luttinger parameter is fixed as the inverse of an odd integer, as a consistency condition dictated by the fermionic nature of the microscopic theory. In the absence of backscattering, this choice of the Luttinger parameter yields fractional values of the quantum of conductance of the form $e^{2} / qh$ with $q$ odd. In principle, other rational values of the two-terminal conductance might be realized taking into account counterpropagating modes \cite{KF, KFP}. However, as mentioned above, the many-body interactions between counterpropagating edge currents might give rise to non-universal values, in contrast to experimental observations. In order to understand the quantization of the two-terminal conductance in the Hall bar, disorder effects have been proposed to play an important role \cite{KF, KFP, S}.

\medskip

\noindent{\bf Achievements of the present work.} The present paper has a two-fold main goal. We start from an interacting, microscopic $2d$ Hamiltonian and we rigorously justify the multi-channel Luttinger liquid description for the edge currents in a wide setting, which allows to take into account counterpropagating edge modes. As stressed above, counterpropagating edge modes appear typically in topological insulators \cite{Fukui} and play an important role in the phenomenological theory of the quantum Hall effect. This effective approach needs to be justified from first principle, and this is what is what we do in the present paper, in the context of the IQHE. We then use this result to compute the edge conductance of the two-dimensional lattice model, defined as the response of the edge currents after introducing a local variation of the chemical potential in proximity of the corresponding boundary. It is worth pointing out that this quantity is, in general, different from the conductance of the multi-channel Luttinger model as defined via Landauer formula, which is not expected to take quantized values in presence of backscattering. We prove the universality of the edge conductance, for weakly interacting quantum Hall systems on the cylinder. Universality is guaranteed by the gauge symmetry of the two-dimensional lattice model, which is lost in the Luttinger liquid approximation for the edge currents. Combined with earlier work on the universality of the bulk Hall conductivity \cite{GMP}, and with results about non-interacting models \cite{Hat}, our result allows to lift the bulk-edge correspondence to the realm of weakly interacting fermionic systems. Our proof holds at the level of equality of transport coefficients. In the non-interacting case, it is known that both transport coefficients have the interpretation of topological indices. Recently, the interpretation of the bulk Hall conductivity for interacting fermions as a topological index has been established in \cite{BBdRF2}. It would be of interest to prove a similar statement for the edge conductance, and to prove the bulk-edge duality as an identity between topological indices. 

In principle, the methods of the present paper could be used to compute all edge correlation functions and other edge transport coefficients, such as the edge Drude weight and the edge susceptibility, which we do not expect to be universal. In addition, we believe that the strategy could also be used to rigorously establish the non-universality of the two-terminal conductance for the Hall bar, as predicted in \cite{KF, KFP} on the basis of the Landauer formula for the conductance of the multi-component Luttinger model.

\medskip

\noindent{\bf Rigorous renormalization.} The main technical tool used in our work is a rigorous formulation of the Wilsonian renormalization group. In its application to condensed matter systems, this method has been pioneered in \cite{BG, FeTr}, see \cite{BGrev, Sabook, Mabook} for reviews. These techniques have been used to study interacting fermionic systems with extended Fermi surface, in dimension greater than one, at finite temperature or with asymmetric Fermi surface, in \cite{DR1, DR2, BGM, FKT}. For one-dimensional systems, these techniques can be used to construct the ground state correlations of interacting models. The first results have been obtained in \cite{BGPS}, for spinless fermions, and in \cite{BoM}, for spinful fermions. Both works crucially rely on the exact solvability of the Luttinger model \cite{ML}, which describes the scaling limit of the lattice models of \cite{BGPS, BoM}, and which ultimately allows to control the infrared fixed point of the renormalization group flow. The use of the exact solution represented an important limitation of the approach. Such restriction has been overcome in \cite{BMdensity, BMchiral, MaQED}, via a completely different strategy. Instead of relying on exact solutions, the approach combines rigorous RG with Ward identities (WIs); together with the Schwinger-Dyson equation for the four-point function, these allow to control non-perturbatively the flow of the effective marginal couplings, and to construct the RG fixed point.  A major technical difficulty is the rigorous control of the anomalies in the Ward identities of the Luttinger liquid, which unavoidably appear in the rigorous formulation of the QFT. This ultimately allows to solve the Schwinger-Dyson equation, and to determine the non-trivial fixed point of the RG without relying on exact solutions. Furthermore, the combination of emergent WIs and of Hamiltonian (or lattice) WIs allowed to prove a number of long-standing conjectures. We mention here the proof of the Haldane relations for transport coefficients of $1d$ quantum systems \cite{BMimp} and the Kadanoff relations for classical $2d$ spin systems \cite{BFM}. Further extensions include the results \cite{F, GMT} for interacting dimers, and \cite{AMP} for the edge transport coefficients of a class of single-mode $2d$ topological insulators. 

Recently, related methods have been used to study transport in two- and three-dimensional quantum systems, with gapped spectrum or with pointlike Fermi surface. The work \cite{GM} considered the two-dimensional Hubbard model on the hexagonal lattice at half-filling, as a model for interacting graphene, and used RG methods to prove the analyticity of the ground state correlations of the system. In \cite{GMPcond}, the combination of RG methods and of lattice Ward identities allowed to prove the universality of the longitudinal conductivity for interacting graphene from a microscopic model, in agreement with the experimental observations of \cite{Nair}; universality holds thanks to the subtle combination of high and low energy modes, and it is missed in the effective relativistic approximation of graphene, widely used in the theoretical physics literature. Concerning the quantum Hall effect, cluster expansion techniques, lattice Ward identities and Schwinger-Dyson equations have been used to prove the universality of the Hall conductivity for weakly interacting fermions on two-dimensional lattices \cite{GMP}. Furthermore, the combination of the strategy of \cite{GMP} with RG techniques allowed to study quantum Hall transitions in the Haldane-Hubbard model \cite{GJMP, GMPcritical}, a prototypical example of interacting topological insulator. In particular, the method allows to fully characterize the interacting phase diagram, and to prove the universality of the longitudinal conductivity on the transition curves. Finally, concerning three-dimensional systems with pointlike Fermi surface, the combination of RG with lattice and emergent Ward identities has been used to establish the non-renormalization of the lattice analogue of the chiral anomaly for Weyl semimetals \cite{GMPweyl}. In all these works, an important simplification is provided by the fact that the interaction is irrelevant in the RG sense, which allows to construct the RG fixed point of the theories in a considerably simpler way. 

In our present setting, the many-body interaction between the edge modes is marginal, as in $1d$ quantum systems. The major difference with respect to previous works is that the low-energy properties of the model are not described by the Luttinger model: it is a genuinely new case. One has to consider an arbitrarily large number of marginal couplings (their number depends on the number of edge modes), and it is a priori not clear at all that the RG flow can be controlled. Furthermore, even if such control is achieved, the proof of universality of edge transport coefficients requires detecting an enormous cancellation between an arbitrarily large number of marginal parameters, impossible to achieve in perturbation theory, and that cannot be inferred by previous results. These problems are solved in the present paper, thus establishing a major advance with respect to the previous literature on rigorous RG. In the next paragraph, we give an outline of the main new ideas introduced in the present work.

\medskip

\noindent{\bf Outline of the proof.} This work generalizes \cite{AMP, MPhelical}. In \cite{AMP} we considered the special case of one chiral edge mode, up to spin degeneracy, while in \cite{MPhelical} we considered the special case of two counterpropagating edge modes, with opposite spins. The models of \cite{AMP, MPhelical} fall into the universality classes of the chiral Luttinger model and of the helical Luttinger model, respectively, and can be studied thanks to the methods of \cite{FM} and of \cite{BMchiral}. In the present paper we allow for an arbitrary number of edge states, without any constraint on their velocities. The main novelty of the present work with respect to \cite{AMP, MPhelical} is the control of the backscattering of the edge modes. In particular, we prove that the edge conductance is unaffected by backscattering. To achieve this, we extend the proofs of the vanishing of the beta function of \cite{FM, BMchiral} to a much larger class theories describing relativistic chiral fermions. The extension is non-trivial, and requires a critical revisitation of the method of \cite{BMchiral}.

Let us give a brief overview of the strategy of the proof. We start by proving that the real-time edge transport coefficients can be analytically continued to imaginary-time, via a rigorous version of the Wick rotation. The argument streamlines and generalizes the one in \cite{AMP}. Imaginary-time correlation functions can then be studied via rigorous RG methods. Our analysis allows to quantitatively compare the edge correlation functions with the correlations of a suitably regularized version of the multi-channel Luttinger model, where the regularization is introduced by a momentum cut-off to be removed at the end. This QFT plays the role of reference model in our work, and it underlies the scaling limit of the edge correlation functions. With respect to the original lattice theory, the reference model enjoys an extra chiral gauge symmetry, associated to independent phase transformations for the chiral fermions. The reference model is also studied via RG; a main technical achievement of the present work is the proof of the vanishing of the beta function, by a nontrivial extension of \cite{BMchiral}. This ultimately allows to control the RG flow of the reference model, which in turn can be used to control the RG flow of the original lattice model. A consequence of this result is the construction of the Gibbs state of the lattice model, in terms of a convergent renormalized series. In particular, we are able to express the edge conductivity of the $2d$ lattice model in terms of the density-density correlation functions of the multi-channel Luttinger model, up to finite renormalizations, which are a priori impossible to compute explicitly.

To access the precise value of the edge conductance, we make use once more of the chiral gauge symmetry of the reference model, to derive closed equations for the density-density correlation functions. These equations take the form of Ward identities, which are anomalous due to the presence of the momentum regularizations. By the analog of the Adler-Bardeen non-renormalization property for the chiral anomaly, we are able to solve the equations, and to find explicit expressions for the density-density correlation functions. However, this analysis is not sufficient to prove the universality of the edge conductance. In fact, the edge conductance is still dependent on various non-universal parameters, entering the definition of the reference model, which have to be fine-tuned in order to capture the correct large-scale behavior of the edge correlations. A key observation is to combine the Ward identity for the lattice vertex function together with the anomalous Ward identity for the vertex function of the reference model, to derive {\it constraints} on the various renormalized parameters. These relations ultimately imply a remarkable cancellation in the expression of the edge conductance, from which universality follows.

\medskip

\noindent{\bf Structure of the paper.} The paper is structured as follows. In Section \ref{sec:model} we introduce the class of many-body lattice systems we consider; in particular, Assumption \ref{assu} specifies a generic structure of the noninteracting edge states. Informally, Assumption \ref{assu} restricts to models exhibiting configurations of edge modes which give rise to a bosonizable many-body interaction, from the emergent field theory viewpoint. Next, in Section \ref{sec:curr} we define the edge conductance, the Schwinger functions, and we prove some important identities among correlations that follow from the lattice conservation laws (lattice Ward identities). In Section \ref{sec:main} we state our main result, Theorem \ref{thm:main}, about the universality of the edge conductance for many-body fermionic systems on a cylinder. The rest of the paper is devoted to the proof of the result. In Section \ref{sec:wick} we prove the Wick rotation. In Section \ref{sec:func} we introduce the Grassmann integral formulation of the Gibbs state, which can then be studied via cluster expansion and multiscale methods. In particular, in Section \ref{sec:red1d} we recall the rigorous mapping of the Grassmann field theory in terms of a $1d$ theory, and in Section \ref{sec:sketchRG} we sketch its RG analysis. Sections \ref{sec:red1d}, \ref{sec:sketchRG} summarize the analysis of \cite{AMP}. The reference model is introduced in Section \ref{sec:defref}, and its RG analysis is discussed in Section \ref{sec:RG}. In Section \ref{sec:flow} we prove the vanishing of the beta function for the reference model, which allows to construct its infrared RG fixed point.  In Section \ref{sec:compare} we compare the original lattice model with the reference model. In particular, in Proposition \ref{prp:rel} we recall a result of \cite{AMP}, that allows to express the edge correlation functions of the lattice model in terms of those of the reference model, up to finite multiplicative and additive renormalizations. Finally, in Section \ref{sec:proof} we collect all the ingredients, and we give the proof of our main result, Theorem \ref{thm:main}.

\section{The model}\label{sec:model}

\subsection{The Hamiltonian and the Gibbs state}
Let $L \in \mathbb{N}$ and let us consider the square lattice:
\begin{equation}
\Lambda_{L} = \big\{ x \in \mathbb{Z}^{2} \mid 0\leq x_{i} \leq L-1\;,\quad i =1,2 \big\}\;.
\end{equation}
We equip the lattice with periodic boundary conditions in the direction of $x_{1}$, and Dirichlet boundary conditions in the direction of $x_{2}$. The wave function of a single particle is $f\in \ell^{2}(\Lambda_{L}; \mathbb{C}^{M})$, where the parameter $M$ takes into account the possible presence of internal degrees of freedom, such as the spin or the sublattice label. The wave function $f(x) = (f_{1}(x), \ldots, f_{M}(x))$ satisfies the following boundary conditions:
\begin{equation}\label{eq:cyl}
f_{\rho}(x) = f_{\rho}(x + n L e_{1})\;,\qquad f_{\rho}((x_{1}, 0)) = f_{\rho}((x_{1}, L-1)) = 0\;,\qquad \rho = 1,\ldots, M\;.
\end{equation}
Thus, we shall think of $\Lambda_{L}$ as lying on a cylinder, and  we shall say that the function $f$ of (\ref{eq:cyl}) satisfies cylindric boundary conditions. It is convenient to introduce the following notion of distance on $\Lambda_{L}$, compatible with the cylindric boundary conditions:
\begin{equation}
\| x - y \|_{L} = \inf_{n \in \mathbb{Z}} \| x - y - n L e_{1} \|\;,\qquad \forall x,y\in \Lambda_{L}\;,
\end{equation}
with $\|\cdot\|$ the usual Euclidean norm on $\mathbb{Z}^{2}$. Let:
\begin{equation}
\mathbb{S}^{1}_{L} := \Big\{ k = \frac{2\pi}{L} n\mid n = 0, 1, \ldots, L-1 \Big\}\;.
\end{equation}
We introduce the partial Fourier transform of $f$ as:
\begin{equation}\label{eq:FT}
f(x_{1}, x_{2}) = \frac{1}{L} \sum_{k\in \mathbb{S}^{1}_{L}} e^{-ikx_{1}} \hat f(k, x_{2}) \iff \hat f(k, x_{2}) = \sum_{x_{1} = 0}^{L-1} e^{ikx_{1}} f(x_{1}, x_{2})\;.
\end{equation}

We consider a class of interacting fermionic models on $\Lambda_{L}$, that we define in second quantization as follows. The Fock space Hamiltonian is:
\begin{eqnarray}\label{eq:Ham}
\mathcal{H}_{L} &=& \sum_{x,y \in \Lambda_{L}} \sum_{\rho, \rho' = 1}^{M} a^{+}_{x,\rho} H_{\rho\rho'}(x;y) a^{-}_{y,\rho'} + \lambda \sum_{x,y \in \Lambda_{L}} \sum_{\rho, \rho' = 1}^{M}\Big( n_{x, \rho} - \frac{1}{2} \Big) w_{\rho \rho'}(x,y) \Big( n_{y,\rho'} - \frac{1}{2} \Big)\nonumber\\
&\equiv& \mathcal{H}_{L}^{0} + \lambda \mathcal{V}_{L}\;,
\end{eqnarray}
where $a^{\pm}_{x,\rho}$ are the usual fermionic creation/annihilation operators, acting on the fermionic Fock space, $\mathcal{F} = \bigoplus_{n\geq 0} \frak{h}^{\wedge n}$, with $\frak{h} = \ell^{2}(\Lambda_{L}; \mathbb{C}^{M})$. The fermionic creation/annihilation operators satisfy the canonical anticommutation relations:
\begin{equation}
\{ a^{-}_{x,\rho}, a^{-}_{y,\rho'} \} = \{ a^{+}_{x,\rho}, a^{+}_{y,\rho'} \} = 0\;,\qquad \{ a^{-}_{x,\rho}, a^{+}_{y,\rho'} \} = \delta_{x,y}  \delta_{\rho, \rho'}\;.
\end{equation}
Furthermore, the fermionic operators are compatible with the periodic boundary conditions:
\begin{equation}
a^{\pm}_{x,\rho} = a^{\pm}_{x + n e_{1} L, \rho}\;,\qquad \forall x\in \Lambda_{L}\;.
\end{equation}
We shall introduce the partial Fourier transform of the fermionic operators, as follows:
\begin{eqnarray}
a^{\pm}_{x,\rho} = \frac{1}{L} \sum_{k \in \mathbb{S}^{1}_{L}} e^{\pm ikx_{1}} \hat a^{\pm}_{k, x_{2}, \rho} &\iff& \hat a^{\pm}_{k, x_{2},\rho} = \sum_{x_{1} = 0}^{L-1} e^{\mp ik x_{1}} a^{\pm}_{x,\rho}\;.
\end{eqnarray}
The first term in (\ref{eq:Ham}), $\mathcal{H}_{L}^{0}$, is the second-quantization of the single particle Hamiltonian $H$ on $\frak{h}$, $H = H^{*}$. The second term in (\ref{eq:Ham}), $\mathcal{V}_{L}$, is the second-quantization of the many-body interaction, defined by a real-valued two-body potential $w$. The symbol $n_{x,\rho}$ denotes the density operator, defined as $n_{x,\rho} = a^{+}_{x,\rho} a^{-}_{x,\rho}$. The parameter $\lambda \in \mathbb{R}$ is the coupling constant of the model. The Hamiltonian $H$ and the two-body potential $w$ are compatible with the periodic boundary condition in the $e_{1}$ direction, and we can suppose that they act trivially on the boundary, in view of the Dirichlet boundary condition (\ref{eq:cyl}). That is:
\begin{eqnarray}
&&H_{\rho\rho'}(x;y) = w_{\rho\rho'}(x;y) = 0\qquad \text{if $x_{2} = 0, L-1$ or $y_{2} = 0, L-1$}\\
&&H_{\rho\rho'}(x,y) = H_{\rho\rho'}(x + nLe_{1},y + m L e_{1})\;,\qquad w_{\rho\rho'}(x,y) = w_{\rho\rho'}(x + nLe_{1},y + m L e_{1})\;.\nonumber
\end{eqnarray}
Also, we will suppose that both $H_{\rho\rho'}(x;y)$ and $w_{\rho\rho'}(x;y)$ are finite-ranged:
\begin{equation}
H_{\rho\rho'}(x;y) = w_{\rho\rho'}(x;y) = 0\qquad \text{if $\| x - y \|_{L} > D$}
\end{equation}
for some finite $D>0$, and translation-invariant:
\begin{equation}
H_{\rho\rho'}(x;y) \equiv H_{\rho\rho'}(x_{1} - y_{1}; x_{2}, y_{2})\;,\qquad w_{\rho\rho'}(x;y) \equiv w_{\rho\rho'}(x_{1} - y_{1}; x_{2}, y_{2})\;.
\end{equation}
Without loss of generality, we suppose that the Hamiltonian only involves nearest and next-to-nearest neighbour hopping terms. That is, 
\begin{equation}
H_{\rho\rho'}(x;y) = 0\;\qquad \text{if $\|x - y\|_{L} > \sqrt{2}$.}
\end{equation}
Longer range hoppings can be taken into account enlarging the number of internal degrees of freedom $M$. Next, we introduce the partial Fourier transforms:
\begin{equation}
\hat H(k; x_{2}, y_{2}) = \sum_{z_{1} = 0}^{L-1} e^{ikz_{1}} H(z_{1}; x_{2}, y_{2})\;,\qquad \hat w(k_{1}; x_{2}, y_{2}) = \sum_{z_{1} = 0}^{L-1} e^{ikz_{1}} w(z_{1}; x_{2}, y_{2})\;.
\end{equation}
The first identity defines the partial Bloch decomposition of $H$, 
\begin{equation}
H = \bigoplus_{k\in \mathbb{S}^{1}_{L}} \hat H(k)\;,
\end{equation}
where $\hat H(k)$ is an effective one-dimensional Schr\"odinger operator on $\ell^{2}([0,L-1] \cap \mathbb{Z}; \mathbb{C}^{M})$. The self-adjointness of $H$ implies the self-adjointness of $\hat H(k)$, and the finite-range of $H$ implies the analyticity of $k\mapsto \hat H(k)$.

The grand-canonical Gibbs state of the model is defined in the usual way, for inverse temperature $\beta > 0$:
\begin{equation}\label{eq:gibbs}
\langle \mathcal{O} \rangle_{\beta, \mu, L} := \frac{1}{\mathcal{Z}_{\beta, \mu, L}} \Tr_{\mathcal{F}}\, \mathcal{O} e^{-\beta (\mathcal{H}_{L} - \mu \mathcal{N}_{L})}\;,
\end{equation}
where $\mu \in \mathbb{R}$ is the chemical potential of the system, and $\mathcal{N}_{L}$ is the number operator:
\begin{equation}
\mathcal{N}_{L} := \sum_{x\in \Lambda_{L}} \sum_{\rho = 1}^{M} n_{x,\rho}\;.
\end{equation}
Being the fermionic Fock space finite-dimensional, the Gibbs state (\ref{eq:gibbs}) is well-defined for finite $\beta$ and finite $L$ and any $\lambda \in \mathbb{R}$. The factors $-1/2$ in (\ref{eq:Ham}) amount to a $O(\lambda)$ shift of the chemical potential, that we introduce in order to simplify the Grassmann representation of the model later on.

We will now make some important assumptions of the single-particle Hamiltonian $H$. We view $H$ as the single-particle Schr\"odinger operator for a quantum Hall system. By the bulk-edge duality, these systems display an insulating behavior in the bulk, and a metallic behavior in proximity of the edges. For a system with no boundary, the insulating behavior follows by requiring that the spectrum of $H$ is gapped, and that the chemical potential lies in the spectral gap. In our setting, the model comes with Dirichlet boundary conditions in the $x_{2}$ direction; a nontrivial bulk topological phase implies the presence of edge modes, understood as solutions of the Schr\"odinger equation that are extended along the edges and integrable in the bulk:
\begin{equation}\label{eq:gedgeintro}
\psi(x_{1}, x_{2}) = e^{ik_{1} x_{1}} \hat \xi(k_{1}; x_{2})\;,
\end{equation}
where $\hat \xi(k_{1})$ is a $\ell^{2}$-normalized solution of:
\begin{equation}\label{eq:Seq}
\hat H(k_{1}) \hat \xi(k_{1}) = \varepsilon(k_{1}) \hat \xi(k_{1})\;,
\end{equation}
compatible with the Dirichlet boundary conditions, $\hat \xi(k_{1}; x_{2}) = 0$ for $x_{2} = 0, L-1$. Edge modes correspond to solutions of (\ref{eq:Seq}) associated to isolated eigenvalues $\varepsilon(k_{1})$. This implies, in particular, that $\hat\xi(k_{1})$ decays exponentially fast away from the boundary of the system:
\begin{equation}
| \hat\xi(k_{1}; x_{2}) | \leq Ce^{-c|x_{2}|}\;,\qquad | \hat\xi(k_{1}; x_{2}) | \leq Ce^{-c|x_{2} - L|}
\end{equation}
where the decay rate $c>0$ depends on the distance between $\varepsilon(k_{1})$ and the rest of the spectrum of $\hat H(k_{1})$. Notice that, in general, there might be multiple edge modes branches, corresponding to dispersion relations $k_{1}\mapsto \varepsilon_{\omega}(k_{1})$, for $\omega = 1,2\ldots$. As $L\to \infty$, these dispersion relations intersect the Fermi level $\mu$ at Fermi points $k^{\omega}_{F}$, $\varepsilon_{\omega}(k^{\omega}_{F}) = \mu$. If non-intersecting, the functions $\varepsilon_{\omega}(k_{1})$ are smooth for $k_{1}$ in a neighbourhood of the corresponding Fermi point $k_{F}^{\omega}$, as a consequence of the smoothness of $\hat H(k_{1})$. Of course, if the system is defined on a finite lattice the notion of smoothness has to be properly interpreted, due to the discreteness of the momenta $k_{1}$. The next assumption specifies the class of Hamiltonians $H$ we will consider; in particular, it specifies the properties of the edge mode branches we will consider in our analysis.
\medskip

\begin{assu}\label{assu} There exists $\delta > 0$, $\tilde \delta > \delta$ such that for all $L$ large enough the following is true. 
\begin{itemize}
\item[a)] Let $I_{\tilde \delta} = (\mu - \tilde \delta, \mu + \tilde \delta)$. We suppose that all the spectrum in this interval of energies consists of edge mode branches. That is, consider a pair $(E, \hat\xi_{E})$ solving the Schr\"odinger equation for energies in $I_{\tilde \delta}$:
\begin{equation}\label{eq:Sedge}
\hat H(k_{1}) \hat\xi_{E} = E \hat\xi_{E}\;,\qquad E \in I_{\tilde \delta}\;.
\end{equation}
Then, $E$ belongs to an edge mode branch, $E = \varepsilon_{\omega}(k_{1})$. Let $\partial_{k_{1}}$ the discrete derivative, $\partial_{k_{1}} f(k_{1}) = \frac{L}{2\pi}( f(k_{1} + 2\pi/L) - f(k_{1}))$. The edge mode branches satisfy:
\begin{eqnarray}\label{eq:deredge}
&&\max_{k_{1}: |\varepsilon_{\omega}(k_{1}) - \mu| \leq \delta}| \partial_{k_{1}}^{n} \varepsilon_{\omega}(k_{1}) | \leq C_{n}\;, \qquad \forall n\in \mathbb{N}\;,\nonumber\\
&&\min_{k_{1}: |\varepsilon_{\omega}(k_{1}) - \mu| \leq \delta} | \partial_{k_{1}} \varepsilon_{\omega}(k_{1}) | > 0\;.
\end{eqnarray}

\item[b)] Let $\hat{\xi}^{\omega}(k_{1})$ be a solution of (\ref{eq:Sedge}), with eigenvalue $\varepsilon_{\omega}(k_{1})$. It satisfies the bounds:
\begin{equation}\label{eq:bdedge}
\| \partial^{n}_{k_{1}} \hat{\xi}^{\omega}(k_{1}; x_{2}) \| \leq C_{n} e^{-c |x_{2}|}\qquad \text{or} \qquad \| \partial^{n}_{k_{1}} \hat{\xi}^{\omega}(k_{1}; x_{2}) \| \leq C_{n} e^{-c |L- x_{2}|}\;.
\end{equation}
We shall denote by $n_{\text{e}}$ the total number of edge modes intersecting the Fermi level, localized at the lower edge (that is satisfying the first bound in (\ref{eq:bdedge})).

\item[c)] The functions $\varepsilon_{\omega}(k_{1})$ converge, as $L\to \infty$, to $C^{1}$ functions of $k_{1} \in I_{\delta}$. We shall set $v_{\omega} := \partial_{k_{1}} \varepsilon_{\omega}(k^{\omega}_{F})$. The second inequality in (\ref{eq:deredge}) implies that $v_{\omega} \neq 0$.

\item[d)] For any fixed $k_{1}\in\mathbb{S}^{1}_{L}$ and $E \in I_{\tilde \delta}$, the eigenvalue equation (\ref{eq:Sedge}) has at most one solution. Also, we suppose that the momentum conservation associated to the two-fermion scattering at the Fermi level is satisfied by at most two different edge modes. That is, the Fermi momenta $k_{F}^{\omega}$ associated to the edge states localized on a given edge satisfy, modulo $2\pi$,
\begin{eqnarray}\label{eq:asse}
k_{F}^{\omega_{1}} &=& k_{F}^{\omega_{2}}\nonumber\\
k_{F}^{\omega_{1}} - k_{F}^{\omega_{2}} &=& k_{F}^{\omega_{3}} - k_{F}^{\omega_{4}}
\end{eqnarray}
only if: for the first line, $\omega_{1} = \omega_{2}$; for the second line $\omega_{1} = \omega_{2}$ and $\omega_{3} = \omega_{4}$ or $\omega_{1} = \omega_{3}$ and $\omega_{2} = \omega_{4}$. For all the other choices of Fermi momenta, there exists $\gamma > 0$ independent of $L$ such that:
\begin{eqnarray}\label{eq:asse2}
| k_{F}^{\omega_{1}} - k_{F}^{\omega_{2}} | &\geq& \gamma \nonumber\\
| (k_{F}^{\omega_{1}} - k_{F}^{\omega_{2}}) - (k_{F}^{\omega_{3}} - k_{F}^{\omega_{4}})| &\geq& \gamma\;.
\end{eqnarray}
\end{itemize} 
\end{assu}

\begin{rem} 
\begin{itemize}
\item[a)] The assumptions a)-c) can be proved to hold for the class of finite-ranged Hamiltonians we are considering, and for generic choices of the chemical potential $\mu$. 
\item[b)] We expect assumption d) to be generically true, unless the Hamiltonian satisfies extra symmetries, usually not present for quantum Hall systems. For an explicit example, consider the direct sum of Haldane's zero-flux Hamiltonians $H_{\text{H}}$ \cite{H1} on the cylinder in a nontrivial topological phase, with an energy shift: $H = H_{\text{H}} \oplus (H_{\text{H}} + \epsilon \mathbbm{1}) \oplus (H_{\text{H}} + \eta \mathbbm{1})$. The resulting system has three edge modes, that satisfy the assumption for generic choices of $\epsilon$ and $\eta$. Notice also that the assumption is trivially true for quantum Hall systems with two edge modes, at non-exceptional values of the chemical potential (to avoid special degeneracies of the Fermi momenta). See {\it e.g.} \cite{Fukui} for numerical results about the edge spectrum of the Hofstadter model. Assumption $d)$ implies that the effective interactions between edge modes in the emergent QFT description are of density-density type; from a non-rigorous viewpoint, these could be studied via bosonization. Here we shall pursue a rigorous approach, based on Ward identities.

\item[c)] The assumption d) rules out the case of multiple edge modes with spin degeneracy. For the case of a single chiral edge mode with spin degeneracy, see \cite{AMP}. For the case of two counterpropagating edge modes with fixed spins, see \cite{MPhelical}.

\item[d)] From now on, we will drop the $\mu$-dependence of the Gibbs state. We will suppose that $\mu$ is compatible with all the assumptions above, and we will write:
\begin{equation}
\langle \cdot \rangle_{\beta,\mu,L} \equiv \langle \cdot \rangle_{\beta, L}\;.
\end{equation}
\end{itemize}
\end{rem}

\section{Lattice current, transport coefficients and conservation laws}\label{sec:curr}

\subsection{Lattice current}

To begin, let us define the current operator. Let $n_{x} = \sum_{\rho} a^{+}_{x,\rho}a^{-}_{x,\rho}$ be the total density operator, and consider its time-evolution:
\begin{equation}
n_{x}(t) = e^{i\mathcal{H}_{L}t} n_{x} e^{-i\mathcal{H}_{L}t}\;.
\end{equation}
The time-evolution of the density operator satisfies:
\begin{eqnarray}\label{eq:cont}
\partial_{t} n_{x}(t) &=& i[ \mathcal{H}_{L}, n_{x}(t) ] \nonumber\\
&=& i e^{i\mathcal{H}_{L}t}[ \mathcal{H}_{L}^{0}, n_{x} ] e^{-i\mathcal{H}_{L}t}\;,
\end{eqnarray}
where $\mathcal{H}_{L}^{0}$ is the second quantization of $H$, and where we used that $[\mathcal{V}_{L}, n_{x}] = 0$. A simple computation gives:
\begin{equation}
i [ \mathcal{H}_{L}^{0}, n_{x} ]  = \sum_{y\in \Lambda_{L}} i ( a^{+}_{y}, H(y;x) a^{-}_{x} ) - i ( a^{+}_{x}, H(x;y) a^{-}_{y} )
\end{equation}
where we used the notation $(a^{+}, A a^{-}) = \sum_{\rho,\rho'} a^{+}_{\rho} A_{\rho\rho'} a^{-}_{\rho'}$. We shall also write:
\begin{equation}\label{eq:bondJ}
i [ \mathcal{H}_{L}^{0}, n_{x} ] = -\sum_{y\in \Lambda_{L}} j_{(x,y)}\;,\qquad j_{(x,y)} = i ( a^{+}_{x}, H(x;y) a^{-}_{y} ) - i ( a^{+}_{y}, H(y;x) a^{-}_{x} )
\end{equation}
where $j_{(x,y)}$ is the bond current associated with the bond $(x,y)$. Notice that $j_{(x,y)} = - j_{(y,x)}$. For our purposes, it will be convenient to rewrite the continuity equation (\ref{eq:cont}) in a divergence form. Consider the lattice derivatives and the lattice divergence:
\begin{equation}
\text{d}_{i} f(x) = f(x) - f(x - e_{i})\;,\qquad \text{div}_{x} \vec f = \text{d}_{1} f_{1}(x) + \text{d}_{2} f_{2}(x)\;.
\end{equation}
Eq. (\ref{eq:cont}) can be rewritten as:
\begin{equation}
\partial_{t} n_{x}(t) + \sum_{i=1,2}\sum_{\alpha = \pm} j_{(x, x + \alpha e_{i})}(t) + \sum_{\alpha = \pm} \sum_{\beta = \pm} j_{(x, x + \alpha e_{1} + \beta e_{2})}(t) = 0\;,
\end{equation}
which we can reformulate as, using the discrete derivative and the lattice divergence:
\begin{equation}\label{eq:conslaw}
\partial_{t} n_{x}(t) + \text{div}_{x} \vec j_{x}(t) = 0\;,
\end{equation}
where $\vec j_{x} = (\vec j_{1,x}, \vec j_{2,x})$ is the {\it current operator.} In terms of the bond currents:
\begin{eqnarray}\label{eq:currents}
j_{1,x} &=& j_{(x,x+e_{1})} + \frac{1}{2}( j_{(x, x + e_{1} - e_{2})} + j_{(x, x+ e_{1} + e_{2})} ) + \frac{1}{2}(j_{(x-e_{2}, x+e_{1})} + j_{(x+e_{2}, x+e_{1})}) \nonumber\\
j_{2,x} &=& j_{(x, x+e_{2})} + \frac{1}{2}( j_{(x, x - e_{1} + e_{2})} + j_{(x, x + e_{1} + e_{2})} ) + \frac{1}{2}( j_{(x - e_{1}, x + e_{2})} + j_{(x + e_{1}, x + e_{2})} )\;.
\end{eqnarray}
Clearly, the definition of $\vec j_{x}$ is far from being unique: we could always replace $\vec j_{x}$ by $\vec j_{x} + \text{curl}_{x} \alpha_{x}$, where $\alpha_{x}$ is a quadratic operator in Fock space, and where $\text{curl}_{x} \alpha_{x} = (\text{d}_{2} \alpha_{x}, -\text{d}_{1} \alpha_{x})$, without affecting the conservation law (\ref{eq:conslaw}). The results we will obtain will not depend on this choice.
\subsection{Edge conductance}
We are interested in the transport properties of the system in proximity of the lower edge of the cylinder, located at $x_{2} = 0$. Due to the edge states for the noninteracting Hamiltonian, we expect to observe a quasi-one dimensional metallic behavior. In order to probe the transport properties of the edge states, we introduce a local variation of the chemical potential, for $t\leq 0$, and $\eta > 0$:
\begin{equation}
\mathcal{H}_{L} \to \mathcal{H}_{L} + \sum_{x\in \Lambda_{L}} e^{\eta t} \mu (x) n_{x} \equiv \mathcal{H}_{L}(t)\;,
\end{equation}
where $\mu(x)$ is supported for $x_{2}\leq a$, with $a$ the width of a strip adjacent the $x_{2} = 0$ edge, $1\ll a \ll L$. The parameter $\eta$ fixes the time-scale on which the perturbation acts, and will be eventually sent to zero. The state of the system evolves according to the von Neumann equation:
\begin{equation}
i\partial_{t} \rho(t) = [ \mathcal{H}_{L}(t), \rho(t) ]\;,\qquad \rho(-\infty) = \rho_{\beta, L}\;.
\end{equation}
We shall denote by $\langle \cdot \rangle_{t}$ the time-dependent state associated to $\rho(t)$:
\begin{equation}
\langle \cdot \rangle_{t} = \tr \cdot \rho(t)\;.
\end{equation}
We are interested in the linear response of the current operator in proximity of the lower edge, at the time $t=0$. Let $a'< a$. We shall suppose that $1\ll a' \ll a \ll L$. Duhamel formula gives, at first order in the perturbation $\mu(\cdot)$:
\begin{equation}\label{eq:condedge}
\langle \sum_{x_{2} \leq a'} j_{1,x} \rangle_{t} - \langle \sum_{x_{2} \leq a'} j_{1,x} \rangle_{\beta, L} = -i\int_{-\infty}^{0}dt\, e^{\eta t}  \sum_{y} \mu(y) \sum_{x_{2}\leq a'} \langle [ n_{y}(t), j_{1,x} ]\rangle_{\beta, L} + \text{h.o.t.}
\end{equation}
For simplicity, suppose that $\mu(x) \equiv \mu(x_{1})$ for $x_{2} \leq a$ and $\mu(x) = 0$ for $x_{2} > a$. We rewrite the first term in the right-hand side of (\ref{eq:condedge}) as, setting $\mathcal{O}^{\leq a}_{x_{1}} := \sum_{x_{2}\leq a} \mathcal{O}_{x}$:
\begin{equation}
-i\int_{-\infty}^{0}dt\, e^{\eta t} \sum_{y_{1}} \mu(y_{1}) \langle [ n^{\leq a}_{y_{1}}(t), j^{\leq a'}_{1,x_{1}} ]\rangle_{\beta, L}  \equiv \frac{1}{L} \sum_{p \in\mathbb{S}^{1}_{L}} e^{-ipx_{1}} \hat \mu(-p) G^{\underline{a}}(\eta, p)\;,
\end{equation}
where $G^{\underline{a}}(\eta, p)$ is given by (using translation invariance in the $x_{1}$-direction):
\begin{equation}\label{eq:Ga}
G_{\beta, L}^{\underline{a}}(\eta, p) := -i\int_{-\infty}^{0}dt\, e^{\eta t}\sum_{y_{1}} e^{-ipy_{1}}  \langle [ n^{\leq a}_{y_{1}}(t), j^{\leq a'}_{1,0} ]\rangle_{\beta, L}\;. 
\end{equation}
Equivalently, by translation invariance, we can rewrite (\ref{eq:Ga}) as:
\begin{eqnarray}\label{eq:Gaequi}
G_{\beta, L}^{\underline{a}}(\eta, p) &=& -i\int_{-\infty}^{0}dt\, e^{\eta t} \frac{1}{L}\sum_{y_{1}, x_{1}} e^{-ip(y_{1} + x_{1})} e^{ipx_{1}}  \langle [ n^{\leq a}_{y_{1} + x_{1}}(t), j^{\leq a'}_{1,x_{1}} ]\rangle_{\beta, L} \nonumber\\
&\equiv& -i\int_{-\infty}^{0}dt\, e^{\eta t}  \frac{1}{L} \langle [ \hat n^{\leq a}_{-p_{1}}(t), \hat j^{\leq a'}_{1,p_{1}} ] \rangle_{\beta, L}\;. 
\end{eqnarray}
Let: 
\begin{equation}\label{eq:edgeconductance}
G^{\underline{a}}(\eta, p) := \lim_{\beta\to \infty} \lim_{L\to \infty}G_{\beta, L}^{\underline{a}}(\eta, p)\;.
\end{equation}
Notice that the left-hand side is defined for $p\in \mathbb{R}$, while the right-hand side is defined for $p \in \frac{2\pi}{L}\mathbb{Z}$. In the $L\to \infty$ limit, it has to be understood that a sequence $p_{L} \to p$ is taken, with $p_{L}\in \frac{2\pi}{L}\mathbb{Z}$ and $p\in \mathbb{R}$ (the limit will not depend on the choice of the sequence, as a byproduct of our analysis).  We define the {\it edge conductance} as:
\begin{equation}\label{eq:Gdef}
G := \lim_{a'\to \infty} \lim_{a\to \infty} \lim_{p\to 0} \lim_{\eta \to 0^{+}} G^{\underline{a}}(\eta, p)\;.
\end{equation}
This quantity describes the response of the edge currents to slowly varying perturbations, in space and in time. A few remarks are in order.
\begin{rem}
\begin{itemize}
\item[a)] Proving the existence of the limits in (\ref{eq:Gdef}) is nontrivial, due to the fact that the spectral gap of the Hamiltonian is closed by the presence of the edge states: this fact implies a slow algebraic decay of correlations in the direction of the edge. Not only we will prove that the limit is finite; our main result, presented in Section \ref{sec:main}, will provide an explicit expression of the edge conductance.
\item[b)] The order of the limits in (\ref{eq:Gdef}) is important. The parameter $a$ defines the width of the strip where the perturbation is supported, while the parameter $a'$ defines the width of the strip where the edge current is supported. We think the external perturbation as living on a macroscopic/mesoscopic scale; instead, the edge current probes the edge modes, which live on a microscopic scale. This justifies the order of limits $a\to \infty$ and then $a'\to \infty$. Furthermore, our analysis will be able to quantify the error introduced by not sending to infinity the parameters $a$ and $a'$. Finally, concerning the relative order of the limits $\eta \to 0^{+}$ and $p\to 0$: the rationale is that we are interested in the response of the system to a static perturbation. As we will see, exchanging the limits would produce a vanishing result.
\end{itemize}
\end{rem}

\subsection{Schwinger functions}\label{sec:schw}

The main part of our technical analysis will focus on the study of the Euclidean correlation functions of the model, also called Schwinger functions. These correlations involve the imaginary-time evolution of the fermionic creation and annihilation operators; for these operators, space-time decay estimates can be proved. Notice that, in general, these bounds cannot be used to directly study real time correlations. Nevertheless, suitable time integrals of real-time correlation functions can be rewritten as integrals of imaginary time correlations, via analytic continuation (also called Wick rotation). As we shall see in Section \ref{sec:wick}, an example of expression that can be analytically continued to imaginary times is the edge conductance, Eq. (\ref{eq:Ga}). 

Let $x_{0}\in [0,\beta)$, and let $\xx = (x_{0}, x) \in [0,\beta) \times \Lambda_{L}$. Let us define the imaginary-time evolution of the fermionic operators as:
\begin{equation}
a^{\pm}_{\xx, \rho} = e^{x_{0} (\mathcal{H}_{L} - \mu \mathcal{N}_{L})} a^{\pm}_{x,\rho} e^{-x_{0} (\mathcal{H}_{L} - \mu \mathcal{N}_{L})}\;.
\end{equation}
Notice that $a^{+}_{\xx,\rho}$ is not the adjoint of $a^{-}_{\xx,\rho}$, if $x_{0}\neq 0$. More generally, given a local Fock space operator $\mathcal{O}_{X}$ we define its imaginary-time evolution as:
\begin{equation}
\mathcal{O}_{(x_{0}, X)} = e^{x_{0} (\mathcal{H}_{L} - \mu \mathcal{N}_{L})} \mathcal{O}_{X} e^{-x_{0} (\mathcal{H}_{L} - \mu \mathcal{N}_{L})}\;.
\end{equation}
Given a collection of imaginary-time evolved fermionic operators, we define the time-ordered product as:
\begin{equation}
\mathbf{T} a^{\varepsilon_{1}}_{\xx_{1}, \rho_{1}} \cdots a^{\varepsilon_{n}}_{\xx_{n}, \rho_{n}} = \text{sgn}(\pi) a^{\varepsilon_{\pi(1)}}_{\xx_{\pi(1)}, \rho_{\pi(1)}} \cdots a^{\varepsilon_{\pi(n)}}_{\xx_{\pi(n)}, \rho_{\pi(n)}}\;,
\end{equation}
where the permutation $\pi$ is such that $x_{0,\pi(1)} \geq \ldots \geq x_{0,\pi(n)}$. If some times are equal, the operator $\mathbf{T}$ acts as normal ordering. We define the $n$-point Schwinger function as:
\begin{equation}\label{eq:Sn}
S^{\beta, L}_{n; \rho_{1}, \ldots, \rho_{n}}(\xx_{1}, \varepsilon_{1}; \ldots; \xx_{n}, \varepsilon_{n}) := \langle \mathbf{T} a^{\varepsilon_{1}}_{\xx_{1}, \rho_{1}} \cdots a^{\varepsilon_{n}}_{\xx_{n}, \rho_{n}} \rangle_{\beta, L} \;,
\end{equation}
and we will set 
\begin{equation}
S_{n}(\cdot) := \lim_{\beta \to \infty} \lim_{L\to \infty} S^{\beta, L}_{n}(\cdot)\;,
\end{equation}
provided the limits exist. So far, the Schwinger functions are defined for imaginary times in $[0,\beta)$. It turns out that the object defined in (\ref{eq:Sn}) satisfies antiperiodic boundary conditions in the time variables. It is then natural to antiperiodically extend the $n$-point Schwinger function to all times in $\mathbb{R}$. From now on, it will be understood that such antiperiodic extension has been taken.

In the absence of interactions, $\lambda = 0$, the Gibbs state is quasi-free and all $n$-point Schwinger functions can be computed from the two-point Schwinger function, thanks to the Wick rule. Suppose that $x_{0} - y_{0} \neq n\beta$. Let $\mathbb{M}^{F}_{\beta}$ be the set of frequencies compatible with antiperiodicity with period $\beta$, also called the fermionic Matsubara frequencies:
\begin{equation}
\mathbb{M}^{F}_{\beta} = \Big\{ k_{0} = \frac{2\pi}{\beta}\Big(n + \frac{1}{2}\Big) \mid n\in \mathbb{Z} \Big\}\;.
\end{equation}
Let us introduce the notation:
\begin{equation}
\underline{k} = (k_{0}, k_{1}) \in \mathbb{M}^{F}_{\beta}\times \mathbb{S}^{1}_{L}\;.
\end{equation}
Then, the two-point function $S_{2; \rho,\rho'}^{\beta, L}(\xx; \yy) \equiv S_{2; \rho,\rho'}^{\beta, L}(\xx, -; \yy, +)$ is given by, in the absence of interactions, see {\it e.g.} \cite[Section 7]{FW} or \cite[Section 4.2.4]{Sabook}:
\begin{equation}\label{eq:2pt}
S_{2; \rho,\rho'}^{\beta, L}(\xx; \yy)\big|_{\lambda = 0} = \frac{1}{\beta L} \sum_{\underline{k} \in \mathbb{M}_{\beta}^{\text{F}} \times \mathbb{S}^{1}_{L}} e^{-i\underline{k} \cdot (\underline{x} - \underline{y})}\Big( \frac{1}{-ik_{0} + \hat H(k) - \mu}\Big)_{\rho,\rho'}(x_{2}; y_{2})\;.
\end{equation}
If instead $x_{0} - y_{0} = n\beta$, the two-point function is the antiperiodic extension of the two-point function at equal times, which is defined via normal ordering. That is:
\begin{equation}
S_{2}^{\beta, L}((y_{0} + n\beta, x); \yy) = (-1)^{n} \lim_{y_{0} - x_{0} \to 0^{-}} S^{\beta, L}_{2}(\xx; \yy)\;.
\end{equation}
Besides the interacting two-point function, we will be interested in the current-current correlation function and in the vertex function. These last two correlations are defined as follows. For $\mu = 0,1,2$, let us define the space-time current $j_{\mu, \xx}$, with $j_{1,\xx}$ and $j_{2,\xx}$ the imaginary-time evolution of (\ref{eq:currents}), and let us set $j_{0,\xx} := n_{\xx}$. We define the {\it connected current-current correlation function} as:
\begin{equation}\label{eq:S02def}
S_{0,2;\mu,\nu}^{\beta, L}(\xx;\yy) = \langle {\bf T} j_{\mu,\xx}\;; j_{\nu,\yy} \rangle_{\beta, L}\;.
\end{equation}
Trivially, $S_{0,2;\mu,\nu}^{\beta, L}(\xx;\yy) = S_{0,2;\nu,\mu}^{\beta, L}(\yy;\xx)$. Also, we define the {\it connected vertex function} as:
\begin{equation}\label{eq:S21def}
S_{1,2;\mu, \rho, \rho'}^{\beta,L}(\zz; \xx; \yy) = \langle {\bf T} j_{\mu,\zz}\;; a^{-}_{\xx,\rho}\;; a^{+}_{\yy;\rho'} \rangle_{\beta, L}\;.
\end{equation}
Being the current operators even in the number of fermionic operators, the expressions in (\ref{eq:S02def}), (\ref{eq:S21def}) are extended periodically in the time variables appearing at the argument of the currents. Let us introduce the set of frequencies compatible with periodicity with period $\beta$, also called the bosonic Matsubara frequencies: 
\begin{equation}
\mathbb{M}^{B}_{\beta} = \Big\{ p_{0} = \frac{2\pi}{\beta} n \, \Big|\, n\in \mathbb{Z} \Big\}\;.
\end{equation}
We set, for $\underline{p} = (p_{0}, p_{1}) = \mathbb{M}^{\text{B}}_{\beta} \times \mathbb{S}^{1}_{L}$, and $\underline{k} \in \mathbb{M}^{\text{F}}_{\beta} \times \mathbb{S}^{1}_{L}$:
\begin{eqnarray}\label{eq:Fcor}
\widehat{S}^{\beta, L}_{2;\rho,\rho'}(\underline{k}; x_{2}, y_{2}) &:=& \int_{0}^{\beta} dx_{0} \sum_{x_{1}=0}^{L-1}e^{i\underline{k}\cdot \underline{x}}\, S^{\beta, L}_{2;\rho,\rho}(\xx; (0,y_{2}))\;,\nonumber\\ 
\widehat{S}_{0,2;\mu,\nu}^{\beta, L}(\underline{p}; x_{2}, y_{2}) &:=& \int_{0}^{\beta} dx_{0} \sum_{x_{1} = 0}^{L-1} e^{i\underline{p}\cdot \underline{x}}\, S_{0,2;\mu,\nu}^{\beta, L}(\xx;(\underline{0}, y_{2}))\;, \\
\widehat{S}_{1,2;\mu, \rho, \rho'}^{\beta,L}(\underline{p}, \underline{k}; z_{2}, x_{2}, y_{2}) &:=& \int_{0}^{\beta} dz_{0} \int_{0}^{\beta} dx_{0} \sum_{z_{1}, x_{1} = 0}^{L-1} e^{i\underline{p}\cdot \underline{z} + i \underline{k}\cdot \underline{x}}\, S_{1,2;\mu, \rho, \rho'}^{\beta,L}(\zz; \xx; (0,y_{2}))\;.\nonumber
\end{eqnarray}
As we shall see in the next section, the conservation of the lattice current (\ref{eq:conslaw}) implies nontrivial identities for these objects.

\subsection{Lattice Ward identities}\label{sec:WI}
The lattice conservation law (\ref{eq:conslaw}) implies nonperturbative identities between Schwinger functions, called Ward identities, that will play a crucial role in our analysis. Here we shall derive identities for the current-current correlation and for the vertex function.

\paragraph{Current-current identity.}  Consider the Euclidean current-density correlation function, for $x_{0} \neq y_{0}$, $0\leq x_{0}, y_{0} < \beta$:
\begin{equation}
\langle {\bf T} n_{\xx}\,; j_{i,\yy}  \rangle_{\beta, L} = \theta(x_{0} - y_{0})  \langle n_{\xx}\,; j_{i,\yy} \rangle_{\beta, L} + \theta(y_{0} - x_{0}) \langle j_{i,\yy}\,; n_{\xx} \rangle_{\beta, L}\;,
\end{equation}
where $\theta(x_{0} -y_{0})$ is $1$ for $x_{0} - y _{0} > 0$ and zero otherwise. Taking the derivative with respect to $x_{0}$, and using the lattice continuity equation (\ref{eq:conslaw}), we get:
\begin{equation}\label{eq:firstWI}
i\partial_{x_{0}} \langle {\bf T} n_{\xx}\,; j_{i,\yy}  \rangle_{\beta, L} = -\text{div}_{x} \langle {\bf T} \vec j_{\xx}\,; j_{i,\yy} \rangle_{\beta, L} + i\delta(x_{0} - y_{0}) \langle [n_{\xx}, j_{i,\yy}]\rangle_{\beta, L}\;.
\end{equation}
This identity will have important consequences on the structure of the edge response function. Let us take the Fourier transform of (\ref{eq:firstWI}):
\begin{equation}
p_{0} \hat S^{\beta, L}_{0,2;0,i}(\underline{p}; x_{2}, y_{2}) = -(1 - e^{ip_{1}}) \hat S^{\beta, L}_{0,2;1,i}(\underline{p}; x_{2}, y_{2}) - \text{d}_{x_{2}} \hat S^{\beta, L}_{0,2;2,i}(\underline{p}; x_{2}, y_{2}) + i \sum_{x_{1}} e^{ip_{1}x_{1}} \langle [n_{x}, j_{i,(0,y_{2})}]\rangle_{\beta, L}
\end{equation}
where we used the periodicity in $x_{0}$, to get rid of the boundary terms arising after integration by parts in the left-hand side. In particular, for $p_{0} \neq 0$:
\begin{equation}\label{eq:S020}
\sum_{x_{2} = 0}^{L-1} \hat S^{\beta, L}_{0,2;0,i}((p_{0}, 0); x_{2}, y_{2}) = 0\;,
\end{equation}
where we used the Dirichlet boundary conditions at $x_{2} = 0$ and $x_{2} = L-1$, and the fact that $[ \mathcal{N}_{L}, j_{i,y}] = 0$.

\paragraph{Vertex identity.} Another useful consequence of the lattice continuity equation is the following identity, for the vertex function:
\begin{eqnarray}
&&i\partial_{z_{0}} \langle {\bf T} n_{\zz}\,; a^{-}_{\xx,\rho}\;; a^{+}_{\yy,\rho'}\rangle_{\beta, L} \\&&\qquad= -\text{div}_{z} \langle {\bf T} \vec j_{\zz}\,; a^{-}_{\xx,\rho}\,; a^{+}_{\yy,\rho'}\rangle_{\beta, L} + i\delta(\yy - \zz) \langle {\bf T} a^{-}_{\xx,\rho} a^{+}_{\yy,\rho'}  \rangle_{\beta, L} - i\delta(\xx - \zz) \langle {\bf T} a^{-}_{\xx,\rho} a^{+}_{\yy,\rho'}  \rangle_{\beta, L}\;, \nonumber
\end{eqnarray}
where $\delta(\xx) = \delta(x_{0}) \delta_{x,0}$. Taking the Fourier transform:
\begin{eqnarray}
p_{0} \hat S^{\beta, L}_{1,2;0,\rho,\rho'}(\underline{p},\underline{k}; z_{2}, x_{2}, y_{2}) &=& -(1 - e^{ip_{1}}) \hat S^{\beta, L}_{1,2;1,\rho,\rho'}(\underline{p},\underline{k}; z_{2}, x_{2}, y_{2}) - \text{d}_{z_{2}}  \hat S^{\beta, L}_{1,2;2,\rho,\rho'}(\underline{p},\underline{k};z_{2},  x_{2}, y_{2}) \nonumber\\
&&  + i \delta_{y_{2},  z_{2}}\hat S^{\beta, L}_{2;\rho,\rho'}(\underline{k}; x_{2}, y_{2}) - i \delta_{x_{2}, z_{2}}\hat S^{\beta, L}_{2;\rho,\rho'}(\underline{k} + \underline{p}; x_{2}, y_{2})\;.
\end{eqnarray}
Finally, summing over $z_{2}$, and using the Dirichlet boundary conditions:
\begin{eqnarray}\label{eq:vertlat0}
&&p_{0} \sum_{z_{2} = 0}^{L-1} \hat S^{\beta, L}_{1,2;0,\rho,\rho'}(\underline{p},\underline{k}; z_{2}, x_{2}, y_{2}) \\
&&= -(1 - e^{ip_{1}}) \sum_{z_{2} = 0}^{L-1} \hat S^{\beta, L}_{1,2;1,\rho,\rho'}(\underline{p},\underline{k}; z_{2}, x_{2}, y_{2}) + i \hat S^{\beta, L}_{2;\rho,\rho'}(\underline{k}; x_{2}, y_{2}) - i \hat S^{\beta, L}_{2;\rho,\rho'}(\underline{k} + \underline{p}; x_{2}, y_{2})\;.\nonumber
\end{eqnarray}
This identity will play a key role in the proof of the universality of the edge conductance.

\section{Main result and sketch of the proof}\label{sec:main}

\subsection{Main result}

Our main result concerns the edge conductance, as defined in (\ref{eq:Gdef}). Already in the absence of interactions, the computation of this quantity is nontrivial. The next theorem provides the exact expression of the edge conductance for weakly interacting lattice models, satisfying the Assumption \ref{assu}.
\begin{thm}[Main result.]\label{thm:main} For the class of lattice models satisfying the Assumption \ref{assu}, there exists $\lambda_{0} > 0$ such that for $|\lambda| < \lambda_{0}$ the following is true. The limits in (\ref{eq:Ga}), (\ref{eq:Gdef}) exist and are finite. In particular,
\begin{equation}\label{eq:Gres}
G = \sum_{\omega}^{*} \frac{\text{sgn}(v_{\omega})}{2\pi}\;,
\end{equation}
where the asterisk denotes summation over all the edge modes localized at the lower edge $x_{2}=0$.
\end{thm}
\begin{rem}
\begin{itemize}
\item[a)] The result proves the universality of the edge conductance against weak many-body interactions: the expression in (\ref{eq:Gres}) is completely specified by the chiralities of the noninteracting edge modes. Combined with the universality of bulk transport \cite{GMP}, our result lifts the bulk-edge correspondence for this class of lattice models to the realm of weakly interacting fermionic models.
\item[b)] The proof of the theorem actually gives much more information about the system that just the value of the edge conductance. It allows to prove the validity of the multi-component Luttinger liquid description of the edge currents, an effective QFT playing an important role for the description of both the integer and the fractional quantum Hall effect. In particular, our result covers cases in which backscattering is present, and hence the Luttinger parameters of the chiral edge modes are not $1$. In principle, the method could be used to compute all edge correlation functions, and various other transport coefficients such as the susceptibility and the Drude weight. In these cases we expect a non-universal behavior, in contrast to the edge conductance. 
\item[c)] Compared to our previous works, \cite{AMP, MPhelical}, Theorem \ref{thm:main} allows to consider an arbitrary number of edge modes. In fact, \cite{AMP} was restricted to the case of quantum Hall systems with single-channel chiral edge modes (up to spin degeneracy), while \cite{MPhelical} focused on the case to two counterpropagating edge modes, a setting relevant for time-reversal invariant topological insulators. The proofs of \cite{AMP, MPhelical} heavily relied on special integrability features of the chiral Luttinger model and of the helical Luttinger model. Technically, they used the \emph{vanishing of the beta function property} for these models to rigorously control the RG flow. With respect to \cite{AMP, MPhelical}, a main technical innovation of the present work is the generalization of the proof of the vanishing of the beta function property for the multi-component Luttinger model.
\end{itemize}
\end{rem}

\subsection{Sketch of the proof}

Let us give a short and informal overview of the proof of the main result. We start by proving that the real-time edge conductance (\ref{eq:edgeconductance}) can be represented in terms of Euclidean correlation functions. This mapping is useful, because Euclidean correlation functions can be efficiently studied via field-theoretic methods. In particular, our method provides a quantitative approximation of the edge conductance in terms of the density-density correlation functions of a regularized version of the multi-channel Luttinger model, which plays the role of reference model in our analysis. {\it Formally}, the partition function of this $1+1$ dimensional reference model is:
\begin{eqnarray}\label{eq:formal0}
\mathcal{Z} &=& \int D\Psi\, e^{-\mathcal{S}(\Psi)} \nonumber\\
\mathcal{S}(\Psi) &=& \sum_{\omega} \int_{\mathbb{R}^{2}} d\underline{x}\, Z_{\omega}\psi^{+}_{\underline{x},\omega} (\partial_{0} + iv_{\omega} \partial_{1}) \psi^{-}_{\underline{x},\omega}\nonumber\\&& + \sum_{\omega,\omega'} \lambda_{\omega,\omega'} Z_{\omega} Z_{\omega'} \int_{\mathbb{R}^{2}\times \mathbb{R}^{2}} d\underline{x}d\underline{y}\, \psi^{+}_{\underline{x},\omega} \psi^{-}_{\underline{y},\omega} \psi^{+}_{\underline{y},\omega'} \psi^{-}_{\underline{x},\omega'} v(\underline{x} - \underline{y})\;,
\end{eqnarray}
for a smooth and short-ranged potential $v(\cdot)$, such that $\hat v(\underline{p}) = 1$. The parameters $Z_{\omega}$, $v_{\omega}$ have to be fine-tuned, in order to make sure that this QFT correctly describes the scaling limit of the edge correlation functions. The label $\omega$ should be understood as indexing the different edge states. The correct definition of the model involves infrared and ultraviolet cutoffs, needed in order to make sense of the functional integral.

At the classical level, the above QFT is invariant under the global chiral gauge symmetry:
\begin{equation}
\psi^{\pm}_{\underline{x},\omega} \to e^{\pm i \alpha_{\omega}} \psi^{\pm}_{\underline{x},\omega}\;.
\end{equation}
This symmetry is absent in the original lattice model. Noether's theorem implies the conservation of the chiral current:
\begin{equation}
(\partial_{0} + i v_{\omega} \partial_{1}) n_{\underline{x},\omega} = 0\;,\qquad n_{\underline{x},\omega} = \psi^{+}_{\underline{x},\omega} \psi^{-}_{\underline{x},\omega}\;.
\end{equation}
In the quantization of the field theory, the analogous conservation law follows from the covariance of the theory under {\it local} chiral gauge transformations. This symmetry, however, turns out to be {\it broken} in the quantized field theory; already in the absence of many-body interactions, $\lambda_{\omega,\omega'} = 0$, we have:
\begin{equation}\label{eq:anom}
(-ip_{0} + v_{\omega} p_{1}) \langle \hat n_{\underline{p},\omega}\;; \hat n_{-\underline{p},\omega'} \rangle = \delta_{\omega,\omega'} \frac{1}{Z^{2}_{\omega}} \frac{(ip_{0} + v_{\omega} p_{1})}{4\pi |v_{\omega}|}\;.
\end{equation}
Eq. (\ref{eq:anom}) is an example of anomalous Ward identity. This identity can be used to compute the density-density correlations in the non-interacting theory. The breaking of the conservation law is due to the fact that in order to give a meaning to the expressions in (\ref{eq:formal0}) one has to introduce cutoffs, which make the model finite-dimensional, but which break the chiral gauge symmetry of the system. This affects the conservation of the chiral current; the extra terms appearing in the continuity equation do not vanish in the limit of cutoff removal, and produce the result (\ref{eq:anom}). We stress that the proof of this identity is non-trivial already in the absence of interactions, and it is reviewed in Section \ref{sec:Ward}, see discussion after (\ref{eq:buban}). 

Remarkably, a similar identity can be proven in the interacting theory. The analysis is much more involved; it relies on the vanishing of the beta function of the QFT, and the non-renormalization of the chiral anomaly. We eventually find, for $\lambda_{\omega,\omega'} \neq 0$ small enough:
\begin{equation}\label{eq:dd}
\langle \hat n_{\underline{p},\omega}\;; \hat n_{-\underline{p},\omega'} \rangle = T_{\omega,\omega'}(\underline{p}) \frac{1}{Z^{2}_{\omega'}} \frac{1}{4\pi |v_{\omega'}|} \frac{ip_{0} + v_{\omega'} p_{1}}{-ip_{0} + v_{\omega'} p_{1}}\;,
\end{equation}
where the matrix $T(\underline{p})$ is explicit:
\begin{equation}
\Big(\frac{1}{T(\underline{p})}\Big)_{\omega,\omega'} = \delta_{\omega,\omega'} + \frac{ip_{0} + v_{\omega} p_{1}}{-ip_{0} + v_{\omega} p_{1}} \frac{1}{4\pi |v_{\omega}|} \frac{1}{Z_{\omega}} \lambda_{\omega,\omega'} Z_{\omega'}\;.
\end{equation}
Another useful anomalous Ward identity relates the vertex function with the two-point function:
\begin{equation}\label{eq:anvert}
\langle \hat n_{\underline{p},\omega}\;; \hat \psi^{-}_{\underline{k}, \omega'}\;; \hat \psi^{+}_{\underline{k} + \underline{p},\omega'} \rangle = T_{\omega,\omega'}(\underline{p}) \frac{1}{Z_{\omega'} D_{\omega'}(\underline{p})} \big(\langle \hat\psi^{-}_{\underline{k},\omega'} \hat\psi^{+}_{\underline{k},\omega'}\rangle - \langle \hat\psi^{-}_{\underline{k} +\underline{p},\omega'} \hat\psi^{+}_{\underline{k} + \underline{p},\omega'}  \rangle \big)\;.
\end{equation}
Let us denote by $S^{\text{ref}}_{0,2;\omega,\omega'}(\underline{p})$ the correlation function in (\ref{eq:dd}). Our analysis allows to express the edge conductance in terms of the density-density correlation function of the reference model. Let $G^{a'}(\underline{p}) = \lim_{a\to \infty} G^{\underline{a}}(\underline{p})$, recall\footnote{We will prove that the $a\to \infty$ limit can be interchanged with the $\underline{p} \to \underline{0}$ limit.} the definition Eq. (\ref{eq:edgeconductance}).  By a rigorous renormalization group analysis of the edge correlation functions of the model we find, for $\underline{p} = (\eta, p_{1})$, $\|\underline{p}\|$ small:
\begin{equation}\label{eq:info}
G^{a'}(\underline{p}) = \big( \vec Z_{0}, S_{0,2}^{\text{ref}}(\underline{p}) \vec Z_{1} \big) + R^{a'}(\underline{p}) + O(e^{-ca'})\;,
\end{equation}
where $(\vec a, \vec b) = \sum_{\omega}^{*} a_{\omega} b_{\omega}$ and the asterisk restricts the sum to the edge modes localized at the lower edge. The quantities $Z_{0,\omega}$, $Z_{1,\omega}$ are the renormalized parameters that allow to relate the density and the current of the lattice model, for momenta $\underline{k}$ close to the space-time Fermi point $\underline{k}_{F}^{\omega} = (0, k_{F}^{\omega})$, to the densities of the reference model.

As manifest from (\ref{eq:dd}), the first term in the right-hand side of (\ref{eq:info}) is singular at $\underline{p} = \underline{0}$. Instead, the second term turns out to have better regularity properties; in particular, it is H\"older continuous at $\underline{p} = \underline{0}$. Now, lattice conservation laws imply that:
\begin{equation}\label{eq:Ga'}
G^{a'}((\eta, 0)) = 0\;.
\end{equation}
This identity can be used to compute the value at $\underline{p} = \underline{0}$ of the error term. We obtain:
\begin{equation}
R^{a'}(\underline{0}) = -\lim_{\eta \to 0^{+}} \lim_{p_{1} \to 0} \big( Z_{0}, S_{0,2}^{\text{ref}}(\underline{p}) Z_{1} \big) + O(e^{-ca'})\;.
\end{equation}
In particular, by the H\"older continuity of $R^{a'}(\underline{p})$:
\begin{equation}
G^{a'}(\underline{p}) = \Big( Z_{0}, \Big[S_{0,2}^{\text{ref}}(\underline{p})  -\lim_{\eta \to 0^{+}} \lim_{p_{1} \to 0}S_{0,2}^{\text{ref}}(\underline{p})\Big] Z_{1} \Big) + O(a'\|\underline{p}\|^{\theta}) + O(e^{-ca'})\;.
\end{equation}
Hence, we find the following remarkable expression for the edge conductance:
\begin{eqnarray}\label{eq:Ginfo}
G &=& \lim_{a'\to \infty} \lim_{p_{1}\to 0} \lim_{\eta \to 0^{+}} G^{a'}(\underline{p}) \nonumber\\
&\equiv& \big(Z_{0}, \mathcal{A} Z_{1} \big)
\end{eqnarray}
where the matrix $\mathcal{A}$ captures the discontinuity of the density-density correlation functions of the reference model:
\begin{equation}
\mathcal{A} = \lim_{p_{1} \to 0}\lim_{\eta \to 0^{+}}  S_{0,2}^{\text{ref}}(\underline{p})  - \lim_{\eta \to 0^{+}} \lim_{p_{1} \to 0}S_{0,2}^{\text{ref}}(\underline{p})\;.
\end{equation}
Eq. (\ref{eq:Ginfo}) is an exact expression for the edge conductance in terms of a number of non-universal parameters, given by a convergent series in the coupling, essentially impossible to compute explicitly. Hence, it is a priori unclear why the expression in (\ref{eq:Ginfo}) should be quantized. The key observation at this point is to combine the lattice Ward identities for the original model (\ref{eq:vertlat0}), together with the anomalous Ward identities for the reference model (\ref{eq:anvert}). More precisely, by our renormalization group analysis we can rigorously approximate the correlation functions of the lattice model with those of the reference model; we find that, for $\underline{k}'$, $\underline{p}$ small, up to subleading terms:
\begin{eqnarray}\label{eq:approxGLM}
S_{2; \rho, \rho'}(\underline{k}' + \underline{k}_{F}^{\omega}; x_{2}, y_{2}) &\simeq& S_{2; \omega}^{\text{ref}} (\underline{k}') \xi^{\omega}_{\rho}(k_{F}^{\omega}, x_{2}) \overline{\xi^{\omega}_{\rho'}(k_{F}^{\omega}, y_{2})} \nonumber\\
\sum_{z_{2}} S_{1,2;\mu,\omega,\rho,\rho'}(\underline{p}, \underline{k}'; z_{2}, x_{2}, y_{2}) &\simeq& \sum_{\omega'} Z_{\mu,\omega'} S^{\text{ref}}_{1,2;\omega',\omega}(\underline{p}, \underline{k}') \xi^{\omega}_{\rho}(k_{F}^{\omega}, x_{2}) \overline{\xi_{\rho'}^{\omega}(k_{F}^{\omega}, y_{2})}\;,
\end{eqnarray}
where
\begin{equation}
S_{2; \omega}^{\text{ref}} (\underline{k}')  = \langle \hat\psi^{-}_{\underline{k},\omega'} \hat\psi^{+}_{\underline{k},\omega'}\rangle\;,\qquad S^{\text{ref}}_{1,2;\omega',\omega}(\underline{p}, \underline{k}') = \langle \hat n_{\underline{p},\omega'}\;; \hat \psi^{-}_{\underline{k}, \omega}\;; \hat \psi^{+}_{\underline{k} + \underline{p},\omega} \rangle\;,
\end{equation}
are the two-point and the vertex function of the reference model. Using the $\ell^{2}$-normalization of the edge modes, it is not difficult to see that the lattice Ward identity (\ref{eq:vertlat0}) together with (\ref{eq:approxGLM}) implies, up to subleading terms for $\|\underline{p}\|$, $\|\underline{k}'\|$ small:
\begin{equation}\label{eq:newWI}
\sum_{\mu = 0,1} (-i)^{\delta_{\mu,0}}p_{\mu}  \sum_{\omega'} Z_{\mu,\omega'} S^{\text{ref}}_{1,2; \omega',\omega}(\underline{p}, \underline{k}') \simeq S_{2; \omega}^{\text{ref}} (\underline{k}') - S_{2; \omega}^{\text{ref}} (\underline{k}' + \underline{p})\;.
\end{equation}
The simultaneous validity of (\ref{eq:newWI}) together with (\ref{eq:anvert}) ultimately implies the following nontrivial relation among renormalized parameters:
\begin{equation}\label{eq:vertices}
\vec Z_{0} = (T^{T}(\underline{0}^{-}))^{-1} \vec Z\;,\qquad \vec Z_{1} = (T^{T}(\underline{0}^{+}))^{-1}  v \vec Z\;,\qquad 
\end{equation}
where $(v \vec Z)_{\omega} = v_{\omega} Z_{\omega}$ and:
\begin{equation}
T(\underline{0}^{-}) = \lim_{p_{1} \to 0} \lim_{\eta \to 0^{+}} T(\underline{p})\;,\qquad T(\underline{0}^{+}) = \lim_{\eta \to 0^{+}} \lim_{p_{1} \to 0} T(\underline{p})\;.
\end{equation}
Plugging the relations (\ref{eq:vertices}) into (\ref{eq:Ginfo}), and using the explicit expression of the discontinuity matrix $\mathcal{A}$, the universality of $G$ follows. This concludes the sketch of the proof of Theorem \ref{thm:main}.

As a final comment, we stress the importance of the gauge symmetry of the original $2d$ lattice model, and of the contribution of the bulk degrees of freedom, which are typically neglected in the existing effective field theory descriptions of the edge currents. In fact, it is only thanks to the lattice conservation law (\ref{eq:Ga'}) and to the bulk contribution $R^{a'}(\underline{p})$ that we are able to obtain the expression (\ref{eq:Ginfo}). Also, the key input for universality follows from the lattice Ward identity for the vertex function, which allows to prove the relations (\ref{eq:vertices}).

\section{Wick rotation}\label{sec:wick}
To begin, let us discuss how to rewrite the edge conductance (\ref{eq:Ga}) in terms of Euclidean correlations. The next result generalizes Lemma B.1 of \cite{AMP}.
\begin{prop}[Wick rotation.]\label{prp:wick} The following identity holds true, for all $T>0$:
\begin{equation}\label{eq:wick}
\int_{-T}^{0} dt\, e^{\eta t} \frac{1}{L} \langle [ \hat n^{\leq a}_{p_{1}}(t), \hat j^{\leq a'}_{1,-p_{1}} ] \rangle_{\beta, L} = i\int_{0}^{\beta} dt\, e^{-i\eta_{\beta} t} \frac{1}{L} \langle {\bf T}\, \hat n^{\leq a}_{p_{1}}(-it) \;; \hat j^{\leq a'}_{1,-p_{1}} \rangle_{\beta, L} + \mathcal{E}_{\beta, L}(T,\eta,p_{1})\;,
\end{equation} 
with $\eta_{\beta} \in \frac{2\pi}{\beta} \mathbb{Z}$ is such that $|\eta - \eta_{\beta}| = \min_{\eta'\in \frac{2\pi}{\beta} \mathbb{Z}} |\eta - \eta'|$ and  where:
\begin{equation}\label{eq:errwick}
|\mathcal{E}_{\beta, L}(T,\eta,p_{1})| \leq \frac{C a a'}{\eta^{3}\beta} + Ke^{-\eta T}\;.
\end{equation}
The constant $C$ is independent of $\beta, L, T, \eta, p_{1}$, while the constant $K$ is independent of $T$.
\end{prop}
\begin{rem} 
\begin{itemize}
\item[a)] The above proposition implies that, recall Eq. (\ref{eq:Gaequi}):
\begin{eqnarray}
G^{\underline{a}}_{\beta, L}(\eta, p) &=& \int_{0}^{\beta} dt\, e^{-i\eta_{\beta} t} \frac{1}{L} \langle {\bf T}\, \hat n^{\leq a}_{p_{1}}(-it) \;; \hat j^{\leq a'}_{1,-p_{1}} \rangle_{\beta, L} + O\Big( \frac{aa'}{\eta^{3} \beta} \Big) \nonumber\\
&\equiv& \sum_{x_{2} = 0}^{a} \sum_{y_{2} = 0}^{a'} \widehat{S}_{0,2;0,1}^{\beta, L}(\underline{p}; x_{2}, y_{2}) + O\Big( \frac{aa'}{\eta^{3} \beta} \Big)
\end{eqnarray}
with $\underline{p} = (\eta_{\beta}, p)$. The argument of the double sum in the right-hand side is the Fourier transform of the current-current Euclidean correlation function, (\ref{eq:Fcor}). 
\item[b)] An important consequence of the identity (\ref{eq:wick}) is that:
\begin{eqnarray}\label{eq:G2S}
\lim_{\beta\to \infty} \lim_{L \to \infty} G^{\underline{a}}_{\beta, L}(\eta, p) = \sum_{x_{2} = 0}^{a} \sum_{y_{2} = 0}^{a'} \lim_{\beta\to \infty} \lim_{L \to \infty}\widehat{S}_{0,2;0,1}^{\beta, L}(\underline{p}; x_{2}, y_{2})\;.
\end{eqnarray}
In particular, the existence of the limits in the left-hand side follows from the existence of the limits in the right-hand side. This last result is nontrivial, and will follow from the renormalization group analysis of the Euclidean correlation functions of the model. 
\item[c)] Compared to Lemma B.1 of \cite{AMP}, the statement (\ref{eq:wick}) does not require any assumption on the Euclidean correlations. Detailed control on these correlations is then of course needed in order to prove the existence of the limits in the right-hand side of (\ref{eq:G2S}).
\end{itemize}
\end{rem}
\begin{proof} The proof is based on a complex deformation argument. To begin, we write:
\begin{equation}\label{eq:E1}
\int_{-T}^{0} dt\, e^{\eta t} \frac{1}{L} \langle [ \hat n^{\leq a}_{p_{1}}(t), \hat j^{\leq a'}_{1,-p_{1}} ] \rangle_{\beta, L} = \int_{-T}^{0} dt\, e^{\eta_{\beta} t} \frac{1}{L} \langle [ \hat n^{\leq a}_{p_{1}}(t), \hat j^{\leq a'}_{1,-p_{1}} ] \rangle_{\beta, L} + \mathcal{E}^{(1)}_{\beta, L}(T,\eta,p_{1})\;,
\end{equation}
where $\mathcal{E}^{(1)}_{\beta, L}$ takes into account the error introduced replacing $e^{\eta t}$ with $e^{\eta_{\beta} t}$. We then rewrite the main term as:
\begin{equation}\label{eq:trivial}
\int_{-T}^{0} dt\, e^{\eta_{\beta} t} \frac{1}{L} \langle [ \hat n^{\leq a}_{p_{1}}(t), \hat j^{\leq a'}_{1,-p_{1}} ] \rangle_{\beta, L} \equiv \int_{-T}^{0} dt\, e^{\eta_{\beta} t} \frac{1}{L} \langle [ \hat n^{\leq a}_{p_{1}}(t)\, ; \hat j^{\leq a'}_{1,-p_{1}} ] \rangle_{\beta, L}\;,
\end{equation}
where $[A;B] = [A - \langle A \rangle; B - \langle B \rangle]$ (a shift by a constant does not change the commutator). Next, we rewrite the right-hand side of (\ref{eq:trivial}) as:\footnote{Recall that $\langle A\,; B\rangle = \langle (A - \langle A\rangle) (B - \langle B \rangle) \rangle$.}
\begin{eqnarray}
&&\int_{-T}^{0} dt\, e^{\eta_{\beta} t} \frac{1}{L} \langle [ \hat n^{\leq a}_{p_{1}}(t)\, ; \hat j^{\leq a'}_{1,-p_{1}} ] \rangle_{\beta, L} \\
&&\qquad\qquad\qquad= \int_{-T}^{0} dt\, \Big[ e^{\eta_{\beta}t} \langle \hat n^{\leq a}_{p_{1}}(t)\, ; \hat j^{\leq a'}_{1,-p_{1}} \rangle_{\beta,L} - e^{\eta_{\beta}t} \langle \hat j^{\leq a'}_{1,-p_{1}}\,; \hat n^{\leq a}_{p_{1}}(t) \rangle_{\beta,L} \Big] \nonumber\\
&&\qquad\qquad\qquad = \int_{-T}^{0} dt\, \Big[ e^{\eta_{\beta}t} \langle \hat n^{\leq a}_{p_{1}}(t)\, ; \hat j^{\leq a'}_{1,-p_{1}} \rangle_{\beta,L} - e^{\eta_{\beta}(t - i\beta)} \langle \hat n^{\leq a}_{p_{1}}(t - i\beta)\,; \hat j^{\leq a'}_{1,-p_{1}} \rangle_{\beta,L} \Big]\;.\nonumber
\end{eqnarray}
The last step follows from the KMS identity, and from the trivial (but crucial) fact $e^{i\eta_{\beta} \beta} = 1$. For $\beta, L$ finite the function
\begin{equation}
f(z) := e^{\eta_{\beta} z} \langle \hat n^{\leq a}_{p_{1}}(z)\,; \hat j^{\leq a'}_{1,-p_{1}} \rangle_{\beta,L}
\end{equation}
is entire. Therefore, by Cauchy theorem the integral of $f(z)$ along the closed complex path 
\begin{equation}
0 \to -i\beta \to -T -i\beta \to -T\to 0
\end{equation}
is zero. We use this observation to rewrite:
\begin{equation}\label{eq:E2}
\int_{-T}^{0} dt\, e^{\eta_{\beta} t} \frac{1}{L} \langle [ \hat n^{\leq a}_{p_{1}}(t)\, ; \hat j^{\leq a'}_{1,-p_{1}} ] \rangle_{\beta, L} = i\int_{0}^{\beta} dt\, e^{-i\eta_{\beta}t} \frac{1}{L} \langle \hat n^{\leq a}_{p_{1}}(-it)\,; \hat j^{\leq a'}_{1,-p_{1}} \rangle_{\beta, L} + \mathcal{E}^{(2)}_{\beta, L}(T,\eta,p_{1})\;,
\end{equation}
where $\mathcal{E}^{(2)}_{\beta, L}$ takes into account the contribution of the path integration on $-T -i\beta \to -T$. Eqs. (\ref{eq:E1}), (\ref{eq:E2}) imply the claim (\ref{eq:wick}), with $\mathcal{E}_{\beta, L}(T,\eta,p_{1}) = \mathcal{E}^{(1)}_{\beta, L}(T,\eta,p_{1}) + \mathcal{E}^{(2)}_{\beta, L}(T,\eta,p_{1})$. Let us bound these error terms. We have:
\begin{equation}
\mathcal{E}^{(1)}_{\beta, L}(T,\eta,p_{1}) = \int_{-T}^{0} dt\, (e^{\eta t} - e^{\eta_{\beta} t}) \frac{1}{L} \langle [ \hat n^{\leq a}_{p_{1}}(t), \hat j^{\leq a'}_{1,-p_{1}} ] \rangle_{\beta, L}\;.
\end{equation}
Consider the average of the commutator. We can estimate this quantity using Lieb-Robison bounds; the argument is standard, and we repeat it for completeness. The commutator is of the form:
\begin{equation}
\Big[ \sum_{x: x_{2}\leq a} A_{x}(t)\,, \sum_{y: y_{2}\leq a'} B_{y} \Big]\;,
\end{equation}
for $A_{x}$, $B_{y}$ bounded and local (meaning that they depend on a finite number of lattice sites around $x$ and $y$). The Lieb-Robinson bound tells us that, for some constants $C, c, v$:
\begin{equation}
\| [ A_{x}(t), B_{y} ] \| \leq C e^{vt - c\|x-y\|_{L}}\;.
\end{equation}
Therefore, we rewrite:
\begin{eqnarray}\label{eq:LR}
\Big[ \sum_{x: x_{2}\leq a} A_{x}(t)\,, \sum_{y: y_{2}\leq a'} B_{y} \Big] &=& \sum_{x: x_{2} \leq a} \sum_{y: y_{2} \leq a'} \chi(|x_{1} - y_{1}|_{L} \leq Kt) \big[ A_{x}(t), B_{y} \big] \nonumber\\
&&+ \sum_{x: x_{2} \leq a} \sum_{y: y_{2} \leq a'} \chi(|x_{1} - y_{1}|_{L} > Kt) \big[ A_{x}(t), B_{y} \big]\;.
\end{eqnarray}
Taking $K$ large enough, the second term is bounded as:
\begin{equation}
\begin{split}
\sum_{x: x_{2} \leq a} \sum_{y: y_{2} \leq a'} \chi(|x_{1} - y_{1}|_{L} > Kt) &\big\| \big[ A_{x}(t), B_{y} \big] \big\| \\ &\leq \sum_{x: x_{2} \leq a} \sum_{y: y_{2} \leq a'} \chi(|x_{1} - y_{1}|_{L} > Kt) Ce^{- (c/2)|x_{1} - y_{1}|_{L} } \nonumber\\
&\leq CLaa' e^{-(c/2)Kt}\;.
\end{split}
\end{equation}
Consider now the first term in (\ref{eq:LR}). We use the following trivial bound, implied by the boundedness of the fermionic operators:
\begin{equation}
\Big\| \sum_{x: x_{2} \leq a} \sum_{y: y_{2} \leq a'} \chi(|x_{1} - y_{1}|_{L} \leq Kt) \big[ A_{x}(t), B_{y} \big] \Big\| \leq C L aa' |t| K\;.
\end{equation}
We are now ready to estimate $\mathcal{E}^{(1)}_{\beta, L}$. Using that, for fixed $\eta$ and for $\beta$ large enough:
\begin{equation}
| e^{\eta t} - e^{\eta_{\beta} t} |\leq \frac{C|t|}{\beta} e^{\eta t}\;,
\end{equation}
we have, possibly for a different constant $C>0$:
\begin{eqnarray}
| \mathcal{E}^{(1)}_{\beta, L}(T,\eta,p_{1}) | &\leq& \frac{C}{\beta}\int_{-T}^{0}dt\, |t| e^{\eta t} \Big[ aa' e^{-(c/2)Kt} + aa'|t|K \Big] \nonumber\\
&\leq& \frac{\widetilde C a a'}{\eta^{3}\beta}\;.
\end{eqnarray}
This estimate corresponds to the first contribution to the right-hand side of (\ref{eq:errwick}). To conclude, let us consider the error term $\mathcal{E}^{(2)}_{\beta, L}(T,\eta,p_{1})$, corresponding to the complex path $-T -i\beta \to -T$:
\begin{equation}
| \mathcal{E}^{(2)}_{\beta, L}(T,\eta,p_{1}) | \leq e^{-\eta T} \int_{0}^{\beta} dt\, \frac{1}{L} \big| \langle \hat n^{\leq a}_{p_{1}}(-T - it)\,; \hat j^{\leq a'}_{1,-p_{1}} \rangle_{\beta, L} \big|\;.
\end{equation}
All we have to do is to find an estimate for the argument of the integral which is sub-exponential in $T$. To do this, we apply Cauchy-Schwarz inequality for traces, to get:
\begin{equation}
\big| \langle \hat n^{\leq a}_{p_{1}}(-T - it)\,; \hat j^{\leq a'}_{1,-p_{1}} \rangle_{\beta, L} \big|^{2} \leq \big| \langle \hat n^{\leq a}_{p_{1}}(-T - it)\,; {\hat n^{\leq a}_{p_{1}}(-T - it)}^{*} \rangle_{\beta, L} \big| \big| \langle \hat j^{\leq a' *}_{1,-p_{1}}\,; \hat j^{\leq a'}_{1,-p_{1}} \rangle_{\beta, L} \big|\;.
\end{equation}
The second factor is $T$ independent, and it is trivially bounded since $\beta$ and $L$ are finite. For the first factor, we use that:
\begin{equation}
\langle \hat n^{\leq a}_{p_{1}}(-T - it)\,; {\hat n^{\leq a}_{p_{1}}(-T - it)}^{*} \rangle_{\beta, L} = \langle \hat n^{\leq a}_{p_{1}}(- it)\,; {\hat n^{\leq a}_{-p_{1}}(it)} \rangle_{\beta, L}\;,
\end{equation}
which is independent of $T$, and also trivially bounded. Therefore:
\begin{equation}
| \mathcal{E}^{(2)}_{\beta, L}(T,\eta,p_{1}) | \leq K e^{-\eta T}\;,
\end{equation}
for some finite constant $K$ that might depend on all parameters but not on $T$. This estimate corresponds to the second contribution in (\ref{eq:errwick}). End of proof.
\end{proof}

\section{Functional integral formulation of the model}\label{sec:func}

\subsection{Grassmann field theory}\label{sec:grass}

The low temperature and large volume limit of the Gibbs state will be investigated using multiscale analysis and renormalization group. To do so, it is convenient to map the statistical mechanics problem into the study of a Grassmann field theory. The mapping is exact, in the sense that all correlation functions of the system have a Grassmann counterpart. We will follow closely the discussion of Section 5 of \cite{AMP}.

Let $\chi(s)$ be a smooth, even, compactly supported function, such that $\chi(s) = 0$ for $|s|>2$ and $\chi(s) = 1$ for $|s|<1$. Let $N\in \mathbb{N}$, and 
\begin{equation}
\mathbb{M}^{\text{F}*}_{\beta} = \{ k_{0}\in \mathbb{M}^{\text{F}}_{\beta} \mid \chi(2^{-N}k_{0}) > 0 \}\;,\qquad \mathbb{D}^{*}_{\beta,L} = \mathbb{M}^{\text{F}*}_{\beta}\times \mathbb{S}^{1}_{L}\;.\; 
\end{equation}
The set $\mathbb{M}^{\text{F}*}_{\beta}$ is finite. The parameter $N$ plays the role of ultraviolet cutoff for the Masubara frequencies of the model. We consider the finite Grassmann algebra generated by the the Grassmann variables $\{\hat\Psi^{\pm}_{\underline{k}, q}\}$ with $\underline{k}\in \mathbb{D}^{*}_{\beta,L}$, $q=1,\ldots, ML$. The label $q$ is an index for the eigenstates of $\hat H(k_{1})$, which is an $LM\times LM$ matrix in a finite volume. The eigenvalue equation reads:
\begin{equation}
\hat H(k_{1}) \varphi^{q}(k_{1}) = e_{q}(k_{1}) \varphi^{q}(k_{1})\;,\qquad e_{q}(k_{1}) \in \mathbb{R}\;,\qquad \varphi^{q}(k_{1}) \in \mathbb{C}^{LM}\;.
\end{equation}
We shall denote by $\varphi^{q}_{\rho}(k_{1}; x_{2})$ the components of the eigenvector $\varphi^{q}(k_{1})$. We shall always suppose that $\varphi^{q}(k_{1})$ is normalized, $\| \varphi^{q}(k_{1}) \| = 1$.

The Grassmann Gaussian integration $\int P_{ N}(d\Psi)$ is a linear functional acting on the Grassmann algebra as follows. Its action on a given monomial $\prod_{j=1}^{n} \hat\Psi^{\e_{j}}_{\underline{k}_{j}, q_{j}}$ is zero unless $| \{j: \e_{j} =+ \} | = |\{ j: \e_{j} = - \}|$, in which case:
\be \int P_{N}(d\Psi) \hat \Psi^-_{\underline{k}_{1},q_{1}}\hat \Psi^+_{\underline{p}_1,q'_{1}}\cdots
\hat \Psi^-_{\underline{k}_n,q_{n}}\hat \Psi^+_{\underline{p}_n,q'_{n}}=\det[C(\underline{k}_j,q_{j};\underline{p}_k,q'_{k})]_{j,k=1,\ldots,n},
\ee
where $C(\underline{k},q;\underline{p},q')=\b L \d_{\underline{k},\underline{p}}\delta_{q,q'} \hat g^{(\leq N)}(\underline{k},q)$ and
\bea\label{eq:prop}
&&\hat g^{(\leq N)}(\underline{k},q) := \frac{\chi_{N}(k_{0})}{-ik_{0} + e_{q}(k_{1}) - \m}\;,\qquad \chi_{N}(k_{0}) \equiv \chi_{0}(2^{-N}k_{0})\;.
\eea
Eq. (\ref{eq:prop}) defines the {\it free propagator} of the Grassmann field. Thanks to the fact that $k_{0} \in \frac{2\pi}{\beta} ( \mathbb{Z} + \frac{1}{2})$, the propagator (\ref{eq:prop}) is bounded uniformly in $\underline{k}, q$. Next, we define the configuration space Grassmann fields as:
\be
\Psi^{+}_{\xx,\rho} := \frac{1}{\beta L}\sum_{\underline{k}\in \mathbb{D}^{*}_{\beta,L}} \sum_{q=1}^{ML} e^{i\underline{k}\cdot \underline{x}}\, \overline{\varphi^{q}_{\rho}(k_{1}; x_{2})}\hat\Psi^{+}_{\underline{k}, q}\;,\qquad \Psi^{-}_{\xx,\rho} := \frac{1}{\beta L}\sum_{\underline{k}\in \mathbb{D}^{*}_{\beta,L}} \sum_{q=1}^{ML} e^{-i\underline{k}\cdot \underline{x}}\, \varphi^{q}_{\rho}(k_{1}; x_{2})\hat\Psi^{-}_{\underline{k}, q}\;.
\ee
We then have:
\be \int P_{N}(d\Psi)\Psi_{\xx,\rho}^-\Psi^+_{\yy,\rho'}= g_{\rho,\rho'}^{(\leq N)}(\xx, \yy),\label{eq:g14.b}\ee
where
\be
g^{(\leq N)}_{\rho,\rho'}(\xx, \yy) = \frac{1}{\beta L}\sum_{\underline{k} \in \mathbb{D}^{*}_{\beta,L}}\sum_{q=1}^{ML} e^{-i\underline{k}\cdot (\underline{x} - \underline{y})} \varphi^{q}_{\rho}(k_{1}; x_{2}) \overline{\varphi^{q}_{\rho'}(k_{1}; y_{2})} \hat g^{(\leq N)}(\underline{k}, q).
\ee
As $N\to \infty$ and for $x_{0} \neq y_{0}$, the propagator converges pointwise to the two-point Schwinger function of the noninteracting lattice model, Eq. (\ref{eq:2pt}). Notice that the propagator is periodic in the $x_{1}$ direction with period $L$, antiperiodic in the $x_{0}$ direction with antiperiod $\beta$, and it satisfies the Dirichlet boundary conditions. 
%
%

If needed, $\int P_{N}(d\Psi)$ can be written explicitly in terms of the usual Berezin integral $\int d\Psi$,
which is the linear functional on the Grassmann algebra acting non trivially on a monomial 
only if the monomial is of maximal degree, in which case 
$$\int d\Psi \prod_{\underline{k}\in\mathbb{D}_{\b,L}^*}\prod_{q=1}^{ML}
\hat \Psi^-_{\underline{k},q}\hat \Psi^+_{\underline{k},q}=1.$$ The explicit expression of 
$\int P_{N}(d\Psi)$ in terms of $\int d\Psi$ is
\bea\label{eq:P} &&\int P_{N}(d\Psi)\big(\cdot\big) = \frac1{\mathcal{N}_{\b,L,N}}\int d\Psi \exp\Big\{-\frac1{\b L}
\sum_{\underline{k}\in\mathbb{D}_{\b,L}^*}\sum_{q=1}^{ML}
\hat\Psi^{+}_{\underline{k},q}\,\big[{\hat g}_{\beta,L,N}(\underline{k}, q)\big]^{-1}\hat\Psi^{-}_{\underline{k},q}\Big\}\big(\cdot\big),
\nonumber\\
&& \quad \text{with normalization}\qquad \mathcal{N}_{\b,L,N}=\prod_{\underline{k}\in\mathbb{D}_{\b,L}^*}\prod_{q=1}^{ML} [\b L]
{\hat g}^{(\leq N)}(\underline{k},q)\;.
\label{2.3}\eea
The Grassmann counterpart of the many-body interaction is:
\be V(\Psi) := \l\int_{0}^\b dx_{0}\sum_{x,y\in \Lambda_{L}} \sum_{\rho, \rho' =1}^{M} n_{\xx,\rho} \delta(x_{0} - y_{0})w_{\rho\rho'}(x, y) n_{\yy,\rho'}\;,
\ee
where $n_{\xx,\rho} = \Psi^{+}_{\xx,\rho} \Psi^{-}_{\xx,\rho}$ is the Grassmann counterpart of the density operator. With respect to (\ref{eq:Ham}), notice the absence of the $-1/2$ factors. The partition function of the Grassmann field theory is:
\begin{equation}\label{eq:partG}
\mathcal{Z}_{N,\beta, L} := \int P_{N}(d\Psi)\, e^{-V(\Psi)}\;.
\end{equation}
For finite $N, \beta, L$, the right-hand side of Eq. (\ref{eq:partG}) is a polynomial in $\lambda$ with bounded coefficients, due to the finiteness of the Grassmann algebra, and to the boundedness of the fermionic propagator (\ref{eq:prop}). Next, let us introduce the generating functional of correlations. Let ${\bf e}_{i} = (0, e_{i})$, for $i=1,2$. We define the Grassmann counterpart of the current operator as:
\bea
&& J_{1,\xx} := J_{\xx,\xx+{\bf e}_{1}} + \frac{1}{2}( J_{\xx, \xx+{\bf e}_{1}-{\bf e}_{2}} + J_{\xx,\xx+{\bf e}_{1}+{\bf e}_{2}} ) + \frac{1}{2}( J_{\xx-{\bf e}_{2}, \xx+{\bf e}_{1}} + J_{\xx+{\bf e}_{2}, \xx+{\bf e}_{1}} ) \nn\\
&&J_{2,\xx} := J_{\xx,\xx+{\bf e}_{2}} + \frac{1}{2}( J_{\xx,\xx-{\bf e}_{1}+{\bf e}_{2}} + J_{\xx, \xx+{\bf e}_{1}+{\bf e}_{2}} ) + \frac{1}{2}( J_{\xx-{\bf e}_{1}, \xx+{\bf e}_{2}} + J_{\xx+{\bf e}_{1}, \xx+{\bf e}_{2}} )\;,\nn
\eea
with the Grassmann bond current, recall Eq. (\ref{eq:bondJ}):
\be
J_{\xx,\yy} := i(\Psi^{+}_{\xx}, H(x;y) \Psi^{-}_{\yy}) -  i(\Psi^{+}_{\yy}, H(y;x) \Psi^{-}_{\xx})\;.
\ee
We then define the source terms as:
\be
B(\Psi; \phi) := \int_{0}^{\beta} dx_{0}\sum_{x\in \Lambda_{L}} \sum_{\rho=1}^{M} (\phi^{+}_{\xx,\rho} \Psi^{-}_{\xx,\rho} + \Psi^{+}_{\xx,\rho} \phi^{-}_{\xx,\rho})\;,\quad \G(\Psi; A) := \int_{0}^{\beta} dx_{0}\sum_{x\in \Lambda_{L}}  \sum_{\m = 0,1,2} A_{\m,\xx} J_{\m,\xx}\;,
\ee
where $\phi^{\pm}_{\xx},\, A_{\m,\xx}$  are, respectively, Grassmann and a complex valued external fields. The generating functional of correlations $\mathcal{W}_{N,\beta,L}(A,\phi)$ is defined as:
\be\label{eq:genfcn}
\mathcal{W}_{N, \beta,L}(A,\phi) := \log \int P_{ N} (d\Psi) e^{-V(\Psi) + \G(\Psi; A) + B(\Psi; \phi)}\;.
\ee
For small $\|A\|_{\infty}$ and for small $|\lambda|$, the argument of the $\log$ is nonzero, thanks to the finiteness of the Grassmann algebra.

It is well-known that the Schwinger functions of the Gibbs state at noncoinciding arguments can be obtained as functional derivatives of the generating functional, in the $N\to \infty$ limit; see, {\it e.g.}, \cite{GMP} for a discussion about this point. We have:
\bea\label{eq:equivss}
&&\log \mathcal{Z}_{\beta, L} = \lim_{N\to \infty} \mathcal{W}_{N}(0,0) \nonumber\\
&&\langle {\bf T} a^{\e_{1}}_{\xx_{1}, \rho_{1}}\,; \cdots \,; a^{\e_{n}}_{\xx_{n}, \rho_{n}}\,; j_{\m_{1}, \yy_{1}}\,; \cdots \,; j_{\m_{m}, \yy_{m}} \rangle_{\beta,L}  \nonumber\\
&&\qquad \qquad = \lim_{N\to \infty}\frac{\partial^{n+m}\mathcal{W}_{N,\beta,L}(A,\phi)}{\partial \phi^{\e_{1}}_{\xx_{1}, \rho_{1}} \cdots\, \partial \phi^{\e_{n}}_{\xx_{n}, \rho_{n}} \partial A^{\sharp_{1}}_{\m_{1}, \yy_{1}} \cdots\, \partial A^{\sharp_{m}}_{\m_{m}, \yy_{m}}} \Big|_{\substack{A=0 \\ \phi=0}}\;.
\eea
The identities (\ref{eq:equivss}) allow to use the properties of Grassmann Gaussian fields in order to investigate the Euclidean correlation functions of our model. Eqs. (\ref{eq:equivss}) hold in the analyticity domain of the Grassmann theory, which coincides with the analyticity domain of the Fock space model.

\subsection{Reduction to an effective one-dimensional model}\label{sec:red1d}

The Grassmann field theory can be studied using multiscale analysis and renormalization group. To do so, a preliminary step is to introduce a notion of `integration of bulk degrees of freedom', that allows to recast the problem into the study of an emergent $1+1$ dimensional QFT, describing the edge currents. The separation between `bulk' and `edge' degrees of freedom has been discussed in detail in \cite{AMP}. Here we shall outline the main ideas, referring to \cite{AMP}, Section 5.2, for further details.

The main advantage of the Grassmann integral formulation with respect to the original Fock space formulation is the addition principle of Grassmann variables. Suppose that $\Psi^{\pm}_{\xx}$ is a Grassmann field with propagator $g = g^{(1)} + g^{(2)}$. Let $\Psi^{(1)\pm}_{\xx}$, $\Psi^{(2)\pm}_{\xx}$ be two independent Grassmann fields, with propagators $g^{(1)}$, $g^{(2)}$. Let $\mathbb{E}$ be the Grassmann Gaussian average with respect to the field $\Psi$, and let $\mathbb{E}_{i}$ be the Grassmann Gaussian average with respect to the field $\Psi^{(i)}$, $i=1,2$. Then:
\begin{equation}
\mathbb{E} f(\Psi) = \mathbb{E}_{2} \mathbb{E}_{1} f(\Psi^{(1)} + \Psi^{(2)}) \equiv \mathbb{E}_{2} f^{(2)}(\Psi^{(2)})\;,
\end{equation}
where $f^{(2)}(\Psi^{(2)}) = \mathbb{E}_{1} f(\Psi^{(1)} + \Psi^{(2)})$. This simple identity is the starting point for the multiscale analysis of the Grassmann field theory.

In order to separate the `bulk' degrees of freedom from the `edge' modes, we rewrite the propagator $\hat g^{(\leq N)}(\underline{k},q)$ as:
\begin{equation}
\hat g^{(\leq N)}(\underline{k},q) = g^{(\text{edge})}(\underline{k},q) + g^{(\text{bulk})}(\underline{k}, q)\;,
\end{equation}
where:
\begin{eqnarray}\label{eq:geb}
g^{(\text{edge})}(\underline{k},q) &:=& \hat g^{(\leq N)}(\underline{k},q) \chi_{N}(k_{0})\chi\Big(\frac{4| e_{q}(k_{1}) - \mu |}{\delta}\Big) \nonumber\\
g^{(\text{bulk})}(\underline{k}, q) &:=& \hat g^{(\leq N)}(\underline{k},q) - g^{(\text{edge})}(\underline{k},q)\;.
\end{eqnarray}
For $\delta$ small enough, the last characteristic function in the first line of (\ref{eq:geb}) selects the energies at a distance $\delta/2$ from the Fermi level. By the Assumptions \ref{assu}, the only allowed energies in this domain are those of the edge modes. For later convenience, we slightly modify the defintion of edge propagator, as follows:
\begin{equation}\label{eq:gedgeinfo}
g^{(\text{edge})}(\underline{k},q) = \sum_{\omega} g^{(\text{edge})}(\underline{k},q) \chi_{N}(k_{0})\chi_{\omega}(k_{1})\;,
\end{equation}
where
\begin{equation}\label{eq:chik1}
\chi_{\omega}(k_{1}) = \chi\Big( 4\Big|\frac{\varepsilon_{\omega}(k_{1}) - \mu_{\omega}}{\delta}\Big| \Big)\;.
\end{equation}
The sum in (\ref{eq:gedgeinfo}) runs over all edge modes intersecting the Fermi level. Recall that the function $\varepsilon_{\omega}(k_{1})$ is the dispersion relation of the edge mode labelled by $\omega$, see discussion after (\ref{eq:gedgeintro}). Notice the presence of $\mu_{\omega}$ instead of $\mu$ in the characteristic function: the two parameters will differ of $O(\lambda)$. The parameters $\mu_{\omega}$ will play the role of counterterms, and will be fixed later on to control the flow of terms that are relevant in the RG sense.

The above momentum-space decomposition induces the following decomposition of the real-space propagator:
\begin{equation}\label{eq:split1}
g^{(\leq N)}(\xx,\yy) = g^{(\text{edge})}(\xx,\yy) + g^{(\text{bulk})}(\xx,\yy)\;,
\end{equation}
where, being $g^{(\text{bulk})}$ supported for energies away from the Fermi momentum:
\begin{equation}\label{eq:gdec}
| g_{\rho,\rho'}^{(\text{bulk})}(\xx,\yy) | \leq \frac{C_{n}}{1 + \| \xx - \yy \|_{\beta, L}^{n}}\;,\qquad 
\end{equation}
where the distance $\| \xx - \yy \|_{\beta, L}$ is defined as:
\begin{equation}
\| \xx - \yy \|_{\beta, L} := \inf_{n,m \in \mathbb{Z}} \| \xx - \yy - n \beta {\bf e}_{0} - nL {\bf e}_{1} \|\;,
\end{equation}
with ${\bf e}_{0} = (1,0)$, ${\bf e}_{i} = (0, e_{i})$.

Correspondingly to (\ref{eq:split1}), the field $\Psi^{\pm}$ is decomposed as $\Psi^{(\text{edge})\pm} + \Psi^{(\text{bulk})\pm}$. The addition principle is then used to write:
\begin{eqnarray}\label{eq:bulkinteg}
e^{\mathcal{W}_{N,\beta,L}(A,\phi)} &=& \mathbb{E}_{\leq N} \Big(e^{-V(\Psi) + \G(\Psi; A) + B(\Psi; \phi)}\Big)\nonumber\\
&=& e^{\mathcal{W}_{N,\beta,L}^{(\text{bulk})}(A,\phi)} \mathbb{E}_{\text{edge}} \Big(e^{V^{\text{(edge)}}(\Psi^{\text{(edge)}}; A, \phi)}\Big)\;,
\end{eqnarray}
for a suitable new functional $\mathcal{W}_{N,\beta,L}^{(\text{bulk})}$ and a suitable new effective interaction $V^{\text{(edge)}}$. The integration step is performed using fermionic cluster expansion, via the Brydges-Battle-Federbush formula, whose convergence for small $|\lambda|$ is ensured by the good decay properties (\ref{eq:gdec}). See \cite{GM} for a review of the method. The form of $\mathcal{W}_{N,\beta,L}^{(\text{bulk})}$, $V^{\text{(edge)}}$ is a priori explicit: both objects are given by a finite sum of monomials in the Grassmann fields and external fields, whose space-time dependent coefficients (also called kernels) have good locality properties, inherited by (\ref{eq:gdec}). For instance:
\begin{equation}\label{eq:vedge}
V^{\text{(edge)}}(\psi; A, \phi) = \sum_{\Gamma = (\Gamma_{\psi}, \Gamma_{A}, \Gamma_{\phi})} \int_{\beta, L} D{\bf X} D{\bf Y} D{\bf Z}\, \Psi_{\Gamma_{\psi}}({\bf X}) \phi_{\Gamma_{\phi}}({\bf Y})A_{\Gamma_{A}}({\bf Z}) W^{\text{(edge)}}_{\Gamma}({\bf X}, {\bf Y}, {\bf Z})\;,
\end{equation}
where: ${\bf X}$, ${\bf Y}$, ${\bf Z}$ collect all the variables at the arguments of $\psi, \phi, A$, respectively; $\Psi_{\Gamma_{\psi}}({\bf X})$, $\phi_{\Gamma_{\phi}}({\bf Y})$, $A_{\Gamma_{A}}({\bf Z})$ are monomials; $\Gamma_{\psi}$, $\Gamma_{A}$, $\Gamma_{\phi}$ label all possible monomials in the corresponding variables; the integral sign denotes integration over all times and summation over all space variables; the kernels are analytic functions of $\lambda$, and satisfy the following bounds, collecting all variables in ${\bf Q}$:
\begin{equation}
\int_{\beta, L} D{\bf Q}\, \prod_{i<j} \| {\bf q}_{i} - {\bf q}_{j} \|_{\beta, L}^{n_{ij}} | W^{\text{(edge)}}_{\Gamma}({\bf Q}) | \leq \beta L^{2}C(\{n_{ij}\})\;.
\end{equation}
Next, let us show how to rewrite the average in the right-hand side of (\ref{eq:bulkinteg}) in terms of a $1+1$ dimensional Grassmann field theory. The key observation is to notice that the edge propagator $g^{\text{(edge)}}$ only depends on the edge modes wave functions and dispersion relations, in proximity of the Fermi level. Explicitly:
\begin{equation}
g^{\text{(edge)}}(\xx;\yy) = \frac{1}{\beta L}\sum_{\underline{k} \in \mathbb{D}^{*}_{\beta,L}}\sum_{\omega} e^{-i\underline{k}\cdot (\underline{x} - \underline{y})} \xi^{\omega}(k_{1};x_{2}) \overline{\xi^{\omega}(k_{1}; y_{2})} \frac{\chi_{N}(k_{0})\chi_{\omega}(k_{1})}{-ik_{0} + \varepsilon_{\omega}(k_{1}) - \mu}\;.
\end{equation}
Let us define:
\begin{equation}
\check{\xi}^{\omega}_{\rho}(x) := \frac{1}{L} \sum_{k_{1} \in \mathbb{S}^{1}_{L}} e^{-i k_{1} x_{1}} \xi^{\omega}_{\rho}(k_{1}; x_{2}) \chi\Big( 2\frac{|\varepsilon_{\omega}(k_{1}) - \mu_{\omega}|}{\delta} \Big)\;.
\end{equation}
As a consequence of the smoothness of the cutoff function and of the assumptions (\ref{eq:bdedge}) on the edge modes, this function satisfies the following decay estimates, for some $\delta$-dependent constants:
\begin{equation}\label{eq:decxi}
|\check{\xi}^{\omega}_{\rho}(x)|\leq C_{n} \frac{e^{-c x_{2}}}{1 + |x_{1}|_{L}^{n}}\;\qquad \text{or}\qquad |\check{\xi}^{\omega}_{\rho}(x)|\leq C_{n} \frac{e^{-c |L - x_{2}|}}{1 + |x_{1}|_{L}^{n}}\;,
\end{equation}
depending on whether the edge mode is localized on the upper or lower edge.

Let $\underline{x} = (x_{0}, x_{1})$. We can view the propagator $g^{\text{(edge)}}(\xx;\yy)$ as the propagator of an effective $1d$ field, $\psi^{\pm}_{\underline{x},\omega}$ with $\underline{x} = (x_{0}, x_{1})$, convoluted with $\check{\xi}^{\omega}(x)$:
\begin{eqnarray}
g^{\text{(edge)}}(\xx;\yy) &=& \mathbb{E}_{\text{1d}} \Big( ( \psi^{-} * \check{\xi} )(\xx)( \psi^{+} * \check{\xi} )(\yy)  \Big) \\
(\psi^{-} * \check{\xi})_{\rho}(\underline{x}) &:=& \sum_{\omega} \sum_{z_{1}}  \psi^{-}_{(x_{0}, z_{1})} \check{\xi}_{\rho}^{\omega}((x_{1} - z_{1}, x_{2}))\nonumber\\ (\psi^{+} * \check{\xi})_{\rho}(\underline{x}) &:=& \sum_{\omega} \sum_{z_{1}}  \psi^{+}_{(x_{0}, z_{1})} \overline{\check{\xi}_{\rho}^{\omega}((x_{1} - z_{1}, x_{2}))}\;.\nonumber
\end{eqnarray}
The new Grassmann Gaussian integration $\mathbb{E}_{\text{1d}}$ is specified by the propagator:
\begin{eqnarray}\label{eq:g1d}
\mathbb{E}_{\text{1d}}( \psi^{-}_{\underline{x},\omega} \psi^{+}_{\underline{y},\omega'} ) &=& \delta_{\omega\omega'} g^{\text{(1d)}}_{\omega}(\underline{x} - \underline{y}) \nonumber\\
g^{\text{(1d)}}_{\omega}(\underline{x} - \underline{y}) &=& \frac{1}{\beta L} \sum_{\underline{k} \in \mathbb{D}^{*}_{\beta,L}} e^{-i\underline{k}\cdot (\underline{x} - \underline{y})} \frac{\chi_{N}(k_{0}) \chi_{\omega}(k_{1})}{-ik_{0} + \varepsilon_{\omega}(k_{1}) - \mu}\;.
\end{eqnarray}
This observation allows to rewrite the generating functional of correlations as:
\begin{eqnarray}\label{eq:map21d}
e^{\mathcal{W}_{N,\beta,L}(A,\phi)} &=& e^{\mathcal{W}_{N,\beta,L}^{(\text{bulk})}(A,\phi)} \mathbb{E}_{\text{1d}} \Big(e^{V^{\text{(edge)}}(\psi * \xi, A, \phi)}\Big) \nonumber\\
&\equiv& e^{\mathcal{W}_{N,\beta,L}^{(\text{bulk})}(A,\phi)} \mathbb{E}_{\text{1d}} \Big(e^{V^{\text{(1d)}}(\psi; A, \phi)}\Big)\;,
\end{eqnarray}
where $V^{\text{(1d)}}$ is the effective interaction for a $1+1$ dimensional Grassmann field. The new effective potential $V^{\text{(1d)}}$ is obtained replacing $\Psi^{(\text{edge})}$ with $\psi * \check\xi$ in (\ref{eq:vedge}); in doing so, we notice that the sum over all $x_{2}$ variables only acts on the edge states eigenfunctions. Thus, it can be performed explicity, and it gives rise to new kernels, which are `anchored' to either the lower or the upper edge,  thanks to the bounds (\ref{eq:decxi}). The new kernels thus obtained satisfy the estimate:
\begin{equation}\label{eq:bd1d}
\int_{\beta, L} D\underline{Q} D\widetilde{Q}_{2} \prod_{i<j} \| \underline{q}_{i} - \underline{q}_{j} \|^{n_{ij}}_{\beta, L} \prod_{l<k} | \tilde q_{l,2} - \tilde q_{k,2} |^{m_{lk}} | W^{\text{(1d)}}_{\Gamma}(\underline{Q}, \widetilde Q_{2}) | \leq \beta L C(\{n_{ij}\}, \{m_{lk}\})\;,
\end{equation}
where $\widetilde{Q}_{2}$ only collects the $y_{2}$, $z_{2}$ variables, associated to the external fields. 
\begin{rem}\label{rem:edges}
\begin{itemize}
\item[a)] Notice the presence of a $\beta L$ factor instead of $\beta L^{2}$; this is a consequence of the localization around $x_{2} = 0$ or $x_{2} = L-1$ introduced by the bounds (\ref{eq:decxi}).
\item[b)] Thanks to the estimates (\ref{eq:decxi}) for the edge mode eigenfunctions and to the decay of the bulk propagator (\ref{eq:gdec}), kernels associated to fields localized on opposite sides of the cylinder are negligible for $L$ large. More precisely, the bound (\ref{eq:bd1d}) has to be multiplied by a factor $C_{n} L^{-n}$. Since the limit $L\to \infty$ is taken before the limit $\beta \to \infty$, these kernels do not play an important role in the RG iteration.
\end{itemize}
\end{rem}
Being the propagator $g^{(\text{1d})}$ translation invariant, the formula (\ref{eq:map21d}) is a convenient starting point for the renormalization group analysis of the model.

\subsection{Sketch of the RG analysis}\label{sec:sketchRG}

The discussion below relies on the detailed RG analysis of \cite{AMP}, see Section 9 there, to which we refer for futher details.  Due to the edge modes, the propagator $g^{(\text{1d})}$ is massless, recall Eq. (\ref{eq:g1d}). In particular, linearization of the dispersion relation around the Fermi energy shows that the new Grassmann field theory obeys the same power counting as a $1+1$ dimensional relativistic model.

In order to integrate the field $\psi$ we proceed in a multiscale fashion. Let $k_{F}^{\omega}(\lambda) = k_{F}^{\omega} + O(\lambda)$ be the solution of:
\begin{equation}\label{eq:intkF}
\varepsilon_{\omega}(k_{F}^{\omega}(\lambda)) = \mu_{\omega}\;.
\end{equation}
Strictly speaking, this equation may not have a solution for finite $L$. Since we are eventually interested in taking the $L\to \infty$ limit, with a slight abuse of notation we shall denote by $k_{F}^{\omega}(\lambda)\in \mathbb{S}^{1}_{L}$ the best approximation of the $L\to \infty$ solution of (\ref{eq:intkF}). We write, for $\underline{k}_{F}^{\omega}(\lambda) = (0, k_{F}^{\omega}(\lambda))$:
\begin{equation}
\psi^{\pm}_{\underline{x},\omega} = \psi^{(uv)\pm}_{\underline{x},\omega}+ \sum_{\omega} \sum_{\substack{ h\in \mathbb{Z}_{-} \\ h_{\beta} \leq h \leq 0}} e^{\pm i\underline{k}_{F}^{\omega}(\lambda)\cdot \underline{x}} \psi^{(h)\pm}_{\underline{x},\omega}\;;
\end{equation}
the field $\psi^{(uv)\pm}$ is supported for momenta $k_{0}$ away from zero, say $|k_{0}| \geq 1$; the single-scale fields $\psi^{(h)\pm}_{\underline{x},\omega}$ depend on momenta $\underline{k}' = \underline{k} - \underline{k}_{F}^{\omega}(\lambda)$, such that:
\begin{equation}\label{eq:mom}
2^{h-1} \leq \sqrt{k_{0}^{2} + v_{\omega}^{2} {k^{'2}_{1}} } \leq 2^{h+1}\;.
\end{equation}
The parameter $h_{\beta} \sim \log \beta$ fixes the infrared cutoff of the theory; it is due to the fact that the smallest fermionic Matsubara frequency is $\frac{\pi}{\beta}$. The propagator of the single scale fields has the form, for $|r_{h,\omega}(k'_{1})|\leq C\|\underline{k}'\|^{\theta}$ with $0<\theta < 1$:
\begin{equation}\label{eq:gassu}
\hat g^{(h)}_{\omega}(\underline{k}') = \frac{1}{Z_{h,\omega}} \frac{f_{h,\omega}(\underline{k}')}{-ik_{0} + v_{h,\omega} k'_{1}}(1 + r_{h,\omega}(k'_{1}))\;;
\end{equation}
the function $f_{h,\omega}(\underline{k}')$ is a smooth cutoff function, that enforces the contraint (\ref{eq:mom}); the parameters $v_{h,\omega}$ are called the effective Fermi velocities, while the parameters $Z_{h,\omega}$ are called the wave function renormalizations. Controlling the flow as $h\to -\infty$ of  $Z_{h,\omega}, v_{h,\omega}$ is one of the main goals of the RG analysis. We shall inductively assume, and ultimately prove, that:
\begin{equation}\label{eq:assZv}
\Big| \frac{Z_{h,\omega}}{Z_{h-1,\omega}} \Big| \leq e^{c|\lambda|}\;,\qquad | v_{h,\omega} - v_{\omega} | \leq C|\lambda|\;,
\end{equation}
uniformly in $h$. On scale $h=0$ one has, as a consequence of (\ref{eq:g1d}):
\begin{equation}
Z_{0,\omega} = 1\;,\qquad v_{0,\omega} = \partial_{k_{1}} \varepsilon_{\omega}(k_{F}^{\omega}(\lambda)) = v_{\omega} + O(\lambda)\;.
\end{equation}
The bounds (\ref{eq:assZv}) are not difficult to check for the first few scales, using perturbation theory; the difficulty is to prove that they hold uniformly in the scale label. 

The field $\psi^{(uv)}$ is integrated in a straightforward way, in an essentially model-independent fashion; see {\it e.g.} Section 5.3 of \cite{GMP}. Due to the natural ultraviolet cutoff introduced by the lattice, this is much easier than the integration of the infrared degrees of freedom. The only minor technical difficulty is due to the slow decay in $k_{0}$ of the momentum space propagator (\ref{eq:g1d}). A careful analysis of perturbation theory shows that the only apparently divergent contribution is the one related to the tadpole graph; this graph can be evaluated explicity, and one finds a finite result.

The fields $\psi^{(h)}$ are then integrated in an iterative way, starting from $h=0$ down to $h = h_{\beta}$. The iterative integration produces a new effective interaction, and the goal is to prove that the sequence of effective interactions converges to a limit, in a suitable sense. Following the analysis of Section 9.2 of \cite{AMP}, one finds out that:
\begin{equation}
\mathbb{E}_{\text{1d}} \Big(e^{V^{\text{(1d)}}(\psi; A, \phi)}\Big) = e^{\mathcal{W}^{(h)}_{\beta,L}(A,\phi)}\mathbb{E}_{\leq h} \Big(e^{V^{(h)}(\sqrt{Z_{h}}\psi; A, \phi)}\Big)
\end{equation}
for a suitable new generating functional $\mathcal{W}^{(h)}_{\beta,L}$ and a suitable new effective interaction $V^{(h)}$. The Grassmann integration involves the fields on scales $\leq h$, and has covariance given by:
\begin{eqnarray}\label{eq:propEh}
\mathbb{E}_{\leq h}( \psi^{-}_{\underline{x},\omega} \psi^{+}_{\underline{y},\omega'} ) &=& \delta_{\omega\omega'} g^{(\leq h)}_{\omega}(\underline{x} - \underline{y}) \nonumber\\
g^{(\leq h)}_{\omega}(\underline{x} - \underline{y}) &=& \frac{1}{\beta L} \sum_{\underline{k}' \in \mathbb{D}^{*}_{\beta,L}} e^{-i\underline{k}'\cdot (\underline{x} - \underline{y})} \frac{1}{Z_{h,\omega}}\frac{\chi_{h,\omega}(\underline{k}')}{-ik_{0} + v_{h,\omega} k'_{1}}(1 + r_{h,\omega}(k'_{1}))\;,
\end{eqnarray}
where:
\begin{equation}
\chi_{h,\omega}(\underline{k}') = \chi\Big(4\frac{2^{-h}}{\delta} \sqrt{ k_{0}^{2} + v_{\omega}^{2} k_{1}^{'2} } \Big)\;.
\end{equation}
The effective interaction has the form:
\begin{eqnarray}
&&V^{(h)}(\sqrt{Z_{h}}\psi; A, \phi) \\
&&= \sum_{\Gamma = (\Gamma_{\psi}, \Gamma_{A}, \Gamma_{\phi})} \int_{\beta, L} D{\underline{X}} D{\bf Y} D{\bf Z}\, \Big[\prod_{f\in \Gamma_{\psi}} \sqrt{Z_{h,\omega(f)}}\Big]\Psi_{\Gamma_{\psi}}(\underline{X}) \phi_{\Gamma_{\phi}}({\bf Y})A_{\Gamma_{A}}({\bf Z}) W^{(h)}_{\Gamma}(\underline{X}, {\bf Y}, {\bf Z})\;. \nonumber
\end{eqnarray}
The new kernels are determined by the integration of the previous scales, and the RG map is equivalent to a recursion relation for these kernels. Let us neglect for simplicity the external fields.  As observed above, the one-dimensional field $\psi_{\underline{x}}^{\pm}$ obeys to the same power counting of relativistic $1+1$ dimensional fermions. It is well-known that for these relativistic models the only dangerous contributions to the effective actions are those introduced by the quadratic and quartic Grassmann monomials. See \cite{GM} for a review of RG for one-dimensional fermions. 

We write:
\begin{eqnarray}\label{eq:2+4}
V^{(h)}(\sqrt{Z_{h}}\psi) &=& \sum_{\underline{\omega}} \int_{\beta, L} d\underline{x} d\underline{y}\, \sqrt{Z_{h,\omega_{1}}} \sqrt{Z_{h,\omega_{2}}} \psi^{+}_{\underline{x},\omega} \psi^{-}_{\underline{y},\omega} W^{(h)}_{2;\underline{\omega}}(\underline{x},\underline{y})\nonumber\\
&& + \sum_{\underline{\omega}} \int_{\beta, L} d\underline{X}\, \Big[\prod_{i=1}^{4} \sqrt{Z_{h,\omega_{i}}}\Big]\psi^{+}_{\underline{x}_{1},\omega_{1}} \psi^{-}_{\underline{x}_{2},\omega_{2}} \psi^{+}_{\underline{x}_{3},\omega_{3}} \psi^{-}_{\underline{x}_{4},\omega_{4}} W^{(h)}_{4;\underline{\omega}}(\underline{X}) \nonumber\\
&& + \text{higher order monomials.}
\end{eqnarray}
For the sake of this discussion, we can safely suppose that the sum over the edge mode labels in (\ref{eq:2+4}) only involves edge modes localized on the same side of the cylinder. As commented in Remark \ref{rem:edges}, kernels associated to product of fields associated to edge modes localized on opposide sides of the cylinder are smaller than any power in $L$. Since the limit $L\to \infty$ is taken before the limit $\beta \to \infty$, and since the number of scales is finite for $\beta$ finite, the contribution of these kernels is negligible.

Let us focus on the quadratic terms. By Assumption \ref{assu}, the Fermi momenta of two different edge states are well separated; recall the first bound in (\ref{eq:asse2}). Since, for $\delta$ small enough, momentum conservation implies Fermi momentum conservation, we get:
\begin{equation}
W^{(h)}_{2;\underline{\omega}}(\underline{x},\underline{y}) = \delta_{\omega_{1},\omega_{2}} W^{(h)}_{2;\omega_{1},\omega_{1}}(\underline{x},\underline{y}) \equiv  \delta_{\omega_{1},\omega_{2}} W^{(h)}_{2;\omega_{1}}(\underline{x},\underline{y})\;.
\end{equation}
We then rewrite the quadratic terms as:
\begin{eqnarray}
&&\sum_{\omega} \int_{\beta, L} d\underline{x} d\underline{y}\, Z_{h,\omega} \psi^{+}_{\underline{x},\omega} \psi^{-}_{\underline{y},\omega} W^{(h)}_{\omega}(\underline{x},\underline{y}) = \sum_{\omega} \int_{\beta, L} d\underline{x} d\underline{y}\, Z_{h,\omega} \psi^{+}_{\underline{x},\omega} \psi^{-}_{\underline{x},\omega} W^{(h)}_{\omega}(\underline{x},\underline{y})\\&&\qquad + \sum_{\omega,\mu} \int_{\beta, L} d\underline{x} d\underline{y}\, Z_{h,\omega} \psi^{+}_{\underline{x},\omega} (\partial_{\mu}\psi^{-}_{\underline{x},\omega}) (x_{\mu} - y_{\mu})W^{(h)}_{\omega}(\underline{x},\underline{y}) + \text{higher order derivatives.}\nonumber
\end{eqnarray}
Using the translation-invariance of the kernels:
\begin{eqnarray}\label{eq:tri}
&&\sum_{\omega} \int_{\beta, L} d\underline{x} d\underline{y}\, Z_{h,\omega} \psi^{+}_{\underline{x},\omega} \psi^{-}_{\underline{y},\omega} W^{(h)}_{\omega}(\underline{x},\underline{y}) = \sum_{\omega} \int_{\beta, L} d\underline{x}\, Z_{h,\omega} 2^{h}\nu_{h,\omega} \psi^{+}_{\underline{x},\omega} \psi^{-}_{\underline{x},\omega} \\&&\qquad + \sum_{\omega} \int_{\beta, L} d\underline{x}\, Z_{h,\omega} (\psi^{+}_{\underline{x},\omega} z_{h,0,\omega} \partial_{0} \psi^{-}_{\underline{x},\omega} + i\psi^{+}_{\underline{x},\omega} z_{h,1,\omega} \partial_{1} \psi^{-}_{\underline{x},\omega}) + \text{higher order derivatives.}\nonumber
\end{eqnarray}
The first term in (\ref{eq:tri}) is {\it relevant} in the renormalization group sense: dimensional bounds would produce an unbounded flow for $\nu_{h,\omega}$. The second term is {\it marginal} in the RG sense. This term is used to redefine the Gaussian integration; it produces new wave function renormalizations and effective Fermi velocities, on scale $h-1$. Dimensional estimates cannot rule out a divergent flow in $h$, with a $\lambda$-dependent power law (anomalous exponent). Finally, the higher order derivaties are {\it irrelevant} in the RG sense: they shrink under iteration of the RG map.

A similar localization procedure is used to rewrite the quartic terms. By Assumption \ref{assu}, recall (\ref{eq:asse}), (\ref{eq:asse2}), together with momentum conservation, implies that:
\begin{equation}\label{eq:quartic}
W^{(h)}_{4;\underline{\omega}}(\underline{X}) = 0\qquad \text{unless $\omega_{1} = \omega_{2}$, $\omega_{3} = \omega_{4}$ or $\omega_{1} = \omega_{3}$, $\omega_{2} = \omega_{4}$.}
\end{equation}
This cancellation drastically reduces the number of quartic terms. We then have:
\begin{eqnarray}\label{eq:quartloc}
&&\sum_{\underline{\omega}} \int_{\beta, L} d\underline{X}\, \Big[\prod_{i=1}^{4} \sqrt{Z_{h,\omega_{i}}}\Big]\psi^{+}_{\underline{x}_{1},\omega_{1}} \psi^{-}_{\underline{x}_{2},\omega_{2}} \psi^{+}_{\underline{x}_{3},\omega_{3}} \psi^{-}_{\underline{x}_{4},\omega_{4}} W^{(h)}_{4;\underline{\omega}}(\underline{X})  \\
&&\qquad \equiv \sum_{\omega,\omega'} \int_{\beta, L} d\underline{X}\, Z_{h,\omega} Z_{h,\omega'} \psi^{+}_{\underline{x}_{1},\omega} \psi^{-}_{\underline{x}_{2},\omega} \psi^{+}_{\underline{x}_{3},\omega'} \psi^{-}_{\underline{x}_{4},\omega'} W^{(h)}_{4;\omega,\omega'}(\underline{X}) \nonumber\\
&&\qquad \equiv \sum_{\omega<\omega'} \int_{\beta, L} d\underline{x}\, Z_{h,\omega} Z_{h,\omega'} \psi^{+}_{\underline{x},\omega} \psi^{-}_{\underline{x},\omega} \psi^{+}_{\underline{x},\omega'} \psi^{-}_{\underline{x},\omega'} \lambda_{h,\omega,\omega'} + \text{higher order derivatives.}\nonumber
\end{eqnarray}
The parameters $\lambda_{h,\omega,\omega'}$ define the new effective coupling constants of the theory. Notice that there is no term with $\omega = \omega'$, due to the fact that the square of a Grassmann variable is zero. The effective coupling constants correspond to marginal directions in the RG sense: dimensional estimates would imply a flow which diverges linearly in $h$. Finally, higher order derivatives are instead 
irrelevant in the RG sense. 

The cancellation (\ref{eq:quartic}) is crucial for our present analysis. It allows to rewrite the effective interaction as a density-density interaction for chiral fermions, at all scales. In the physics literature, this type of interaction can be analyzed via bosonization methods, which can be used to derive exact solutions for relativistic models. Our analysis will not rely on these ideas; instead, we will use the emergent chiral symmetry of the QFT to derive exact identities for the correlation functions of the scaling limit of the model, which ultimately allow to control the RG flow. 

To conclude, the control of the RG map is equivalent to the control of the following finite-dimensional dynamical system:
\begin{eqnarray}\label{eq:betaflow0}
\frac{Z_{h-1,\omega}}{Z_{h,\omega}} &=& 1 + z_{h,0,\omega} \equiv 1 + \beta^{z}_{h,\omega} \nonumber\\
v_{h-1,\omega} &=& \frac{Z_{h,\omega}}{Z_{h-1,\omega}} (v_{h,\omega} + z_{h,1,\omega}) \equiv v_{h} + \beta^{v}_{h,\omega}\nonumber\\
2^{h} \nu_{h,\omega} &=& \frac{Z_{h,\omega}}{Z_{h-1,\omega}} \widehat{W}^{(h)}_{2;\omega}(\underline{0}) \equiv 2^{h+1} \nu_{h+1,\omega} + 2^{h+1}\beta^{\nu}_{h+1,\omega}\nonumber\\
\lambda_{h,\omega,\omega'} &=& \Big[ \prod_{i=1}^{4} \sqrt{\frac{Z_{h,\omega}}{ Z_{h-1,\omega}}} \Big] \widehat{W}^{(h)}_{4;\omega,\omega,\omega',\omega'}(\underline{0},\underline{0},\underline{0}) \equiv \lambda_{h+1,\omega,\omega'} + \beta^{\lambda}_{h+1,\omega,\omega'}\;.
\end{eqnarray}
The initial data of the dynamical system are:
\begin{equation}
Z_{0,\omega} = 1\;,\qquad v_{0,\omega} = v_{\omega} + O(\lambda)\;,\qquad \nu_{0,\omega} = \mu_{\omega} + O(\lambda)\;,\qquad \lambda_{0,\omega, \omega'} = A_{\omega, \omega'}\lambda + O(\lambda^{2})\;,
\end{equation}
with 
\begin{equation}
A_{\omega,\omega'} = \sum_{x_{2}, y_{2}} \sum_{\rho, \rho'}w_{\rho\rho'}(\underline{0}; x_{2}, y_{2}) \xi_{\rho}(k_{F}^{\omega}, x_{2}) \overline{\xi_{\rho}(k_{F}^{\omega}, x_{2})} \xi_{\rho'}(k_{F}^{\omega'}, y_{2}) \overline{\xi_{\rho'}(k_{F}^{\omega'}, y_{2})}\;;
\end{equation}
the higher-order terms also take into account the bulk and of the ultraviolet degrees of freedom. The function: 
\begin{eqnarray}
\beta_{h,\underline{\omega}} &=& \beta_{h,\underline{\omega}}\big( \{ Z_{h,\omega}, v_{k,\omega}, \nu_{k,\omega}, \lambda_{k,\underline{\omega}} \}_{k=h+1}^{0} \big)  \nonumber\\
&=& ( \beta^{z}_{h,\omega}, \beta^{v}_{h,\omega}, \beta^{\nu}_{h,\omega}, \beta^{\lambda}_{h,\omega} )
\end{eqnarray}
is called the {\it beta function} of the theory. As it is clear from (\ref{eq:betaflow0}), a precise control of this function is needed in order to safely iterate the equations in (\ref{eq:betaflow0}) for $h\to -\infty$ (zero temperature limit). In \cite{AMP}, control of this dynamical system was achieved in the case of a single-edge mode, spin degenerate (which implies in particular that the quartic interaction is nontrivial). In \cite{MPhelical}, instead, the same control has been obtained for a configuration of edge states corresponding to a pair of counterpropagating modes, with opposite spins (a setting relevant for time-reveral invariant systems).

In both \cite{AMP, MPhelical}, the control of the dynamical system was achieved comparing the beta function in (\ref{eq:betaflow}) with the beta function of a model that can be exactly solved using renormalization group. We write:
\begin{equation}
\beta_{h,\underline{\omega}} = \beta^{\text{ref}}_{h,\underline{\omega}} + O(\lambda_{\geq h}^{2} 2^{\theta h})\;,
\end{equation} 
where $\beta^{(\text{ref})}_{h,\underline{\omega}}$ is the beta function of a suitable {\it reference model}, which describes the scaling limit for the original lattice model. In \cite{AMP}, the reference model coincided with the chiral Luttinger model, while in \cite{MPhelical} it coincided with the helical Luttinger model. In both cases, the beta function of the reference model can be proved to be exponentially small in the scale label: this is what we call the (asymptotic) vanishing of the beta function.

In the present setting, due to the arbitrary number of edge modes, we have to generalize the notion of reference model. We shall consider the multi-channel Luttinger model, describing  arbitrary number of chiral relativistic fermions in $1+1$ dimensions, interacting via a density-density quartic interaction. For this class of models, we will prove the vanishing of the beta function property by extending previous results, \cite{BMWI, BMchiral, BFM2}.

The following result allows to control the iteration of the RG map for the lattice model, and in particular to prove that $|\lambda_{k,\omega,\omega'}| \leq C|\lambda|$ together with the bounds (\ref{eq:assZv}).
\begin{prop}[Bounds for the beta function.]\label{prp:flow} There exists choices of $\mu_{\omega} = \mu + O(\lambda)$ such that, for $|\lambda|$ small enough, the following is true. Let $0 < \theta < 1$, $2^{\theta h_{\beta}} \geq Ce^{-cL}$. Let $|\lambda_{k}| = \max_{\omega,\omega'} |\lambda_{k,\omega,\omega'}|$. Then, the following estimates hold:
\begin{equation}\label{eq:betaflow}
|\beta^{v}_{k,\omega}| \leq C_{\theta} |\lambda_{k}| 2^{\theta k}\;,\qquad |\beta^{\lambda}_{k,\omega}| \leq C_{\theta} |\lambda_{k}|^{2} 2^{\theta k}\;,\qquad |\nu_{k,\omega}| \leq C_{\theta} |\lambda_{k}| 2^{\theta k}\;.
\end{equation}
\end{prop}
\begin{rem} The first two estimates allow to control the flow of the effective velocity and of the effective couplings. The proof follows \cite{AMP}, and it will be reviewed in Section \ref{sec:compare}. It is based on the comparison of the beta function of the lattice model with the beta function of the reference model, which turns out to be asymptotically vanishing as $h\to -\infty$. This key property will be proven in Section \ref{sec:flow}, adapting the existing argument developed for the usual Luttinger model \cite{BMWI, BMchiral, BFM2}. Notice that this comparison only allows to control the flow of the marginal couplings: the control of the flow of the relevant term $\nu_{k,\omega}$ is achieved by suitably choosing $\mu_{\omega} = \mu + O(\lambda)$, as in Proposition 9.4 of \cite{AMP}, see Section 4.1 of \cite{BM} for more details.
\end{rem}

\section{Definition of the reference model}\label{sec:defref}

\begin{rem} From this section until Section \ref{sec:compare}, with a slight abuse of notation we will use the symbols $Z_{\omega}$, $v_{\omega}$, $\lambda_{\omega,\omega'}$ to denote the bare parameters of the reference model. Until Section \ref{sec:compare}, these quantities will be unrelated from the analogous parameters of the lattice model. Later, we will replace them by $Z^{\text{ref}}_{\omega}$, $v^{\text{ref}}_{\omega}$, $\lambda^{\text{ref}}_{\omega,\omega'}$, to distinguish them from the analogous quantities appearing in the analysis of the lattice model.
\end{rem}

This section is devoted to the construction of a one-dimensional quantum field theory model, that underlies the infrared scaling limit of the edge correlation functions of our class of $2d$ condensed matter systems. The definition is motivated by the analysis of Section \ref{sec:sketchRG}: it is given by the multi-channel Luttinger model in the presence of a momentum cut-off, to be removed at the end. Formally, we want to construct a QFT with the following partition function:
\begin{eqnarray}\label{eq:formal}
\mathcal{Z} &=& \int D\Psi\, e^{-\mathcal{S}(\Psi)} \nonumber\\
\mathcal{S}(\Psi) &=& \sum_{\omega} \int_{\mathbb{R}^{2}} d\underline{x}\, Z_{\omega}\psi^{+}_{\underline{x},\omega} (\partial_{0} + iv_{\omega} \partial_{1}) \psi^{-}_{\underline{x},\omega}\nonumber\\&& + \sum_{\omega,\omega'} \lambda_{\omega,\omega'} Z_{\omega} Z_{\omega'} \int_{\mathbb{R}^{2}\times \mathbb{R}^{2}} d\underline{x}d\underline{y}\, \psi^{+}_{\underline{x},\omega} \psi^{-}_{\underline{y},\omega} \psi^{+}_{\underline{y},\omega'} \psi^{-}_{\underline{x},\omega'} v(\underline{x} - \underline{y})\;,
\end{eqnarray}
for $v(\cdot)$ smooth and short-ranged, such that $\hat v(\underline{0}) = 1$. The parameters $Z_{\omega}, v_{\omega}, \lambda_{\omega,\omega'}$ have to be fine tuned, in order to guarantee that the above QFT correctly describes the scaling limit of the edge correlation functions of the lattice model. A nice feature of the above theory is that it is invariant under chiral $U(1)$ gauge transformations:
\begin{equation}\label{eq:chisym}
\psi^{\pm}_{\underline{x},\omega}\to e^{\pm i \alpha_{\omega}} \psi^{\pm}_{\underline{x},\omega}\;.
\end{equation}
This symmetry is {\it absent} in the original lattice model. As we will see, this symmetry will play an important role in establishing the vanishing of the beta function.

The problem with the above definition of QFT is that the expression in (\ref{eq:formal}) is meaningless; in order to correctly formulate the model in a statistical mechanics setting, the functional integral has to be defined as a limit of well-defined objects involving finitely many degrees of freedom. Technically, this means that one has to introduce infrared and ultraviolet cutoffs. The presence of these unavoidable regularizations will have important  consequences: they will affect the conservation of the chiral current, which at the classical level is guaranteed by (\ref{eq:chisym}). More precisely, the Ward identities of the theory will be {\it anomalous.}

\paragraph{Sets and cutoffs.} Given even integers $L>0$, $\frak{N}>0$, and defining $\frak{a} := L/\frak{N}$, we introduce the space-time lattice of side $L$ and mesh $\frak{a}$ as:
\begin{equation}
\Lambda_{L,\frak{a}} := \Big\{ \underline{x} = ( n_{1} \frak{a}, n_{2} \frak{a} )\, \mid \, 0\leq n_{i} \leq \frak{N}-1 \Big\}\;. 
\end{equation}
We will consider Grassmann fields defined on this finite lattice, and we will extended them antiperiodically to the whole set $\frak{a} \mathbb{Z}^{2}$. To this end, it is convenient to define the following set of momenta, needed in order to define the Fourier series for antiperiodic functions:
\begin{equation}
\mathbb{D}_{L,\frak{a}}^{\text{a}} := \Big\{ \underline{k} = \frac{2\pi}{L}\Big( m_{1} + \frac{1}{2}, m_{2} + \frac{1}{2} \Big) \, \big| \, 0 \leq m_{i} \leq \frak{N}-1 \Big\}\;.
\end{equation}
Notice that the sets $\Lambda_{L,\frak{a}}$ and $\mathbb{D}_{L,\frak{a}}^{\text{a}}$ have the same number of elements.

Let $\chi(\cdot)$ be a cutoff function as in Section \ref{sec:grass}. For technical reasons, it will be convenient to introduce a $\varepsilon$-deformation $\chi^{\varepsilon}$, as follows:
\begin{equation}
\chi^{\varepsilon}(t) = C_{\varepsilon} \int_{0}^{\infty} ds\, e^{-| t - s |^{2} / \varepsilon} \chi(s)\;,
\end{equation}
with $C_{\varepsilon}$ such that $C_{\varepsilon}  \int_{0}^{\infty} ds\,e^{-| t - s |^{2} / \varepsilon} = 1$. One has $\lim_{\varepsilon \to 0} \chi^{\varepsilon}(t) = \chi(t)$; also, for $\varepsilon >0$, the function $\chi^{\varepsilon}(t)$ is never vanishing. Next, given $v_{\omega} \in \mathbb{R}$, $\omega = 1, \ldots, n_{\text{e}}$, $v_{\omega} \neq 0$, we define the Euclidean norm:
\begin{equation}\label{eq:dist0}
| \underline{k} |_{\omega} := \sqrt{ k_{0}^{2} + v^{2}_{\omega} k_{1}^{2} }\;.
\end{equation}
The parameters $v_{\omega}$ will play the role of bare velocities for the reference model. We will consider functions on the finite set of momenta $\mathbb{D}_{L,\frak{a}}^{\text{a}}$, that will be extended periodically over the whole $\frac{2\pi}{L} (\mathbb{Z} + \frac{1}{2})^{2}$. Thus, it is convenient to introduce the following notion of distance on $\frac{2\pi}{L} (\mathbb{Z} + \frac{1}{2})^{2}$:
\begin{equation}
\| \underline{k} \|_{\omega} := \inf_{a_{1}, a_{2} \in \mathbb{Z}} | \underline{k} - a_{1} \underline{G}_{1} - a_{2} \underline{G}_{2} |_{\omega}\;,
\end{equation}
with $\underline{G}_{1}$, $\underline{G}_{2}$ the basis vectors of the reciprocal lattice $(\frak{a} \mathbb{Z}^{2})^{*}$:
\begin{equation}
\underline{G}_{1} = \frac{2\pi}{\frak{a}}(1, 0)\;,\qquad \underline{G}_{2} = \frac{2\pi}{\frak{a}}(0, 1)\;.
\end{equation}
Finally, given two integers $h < 0$, $N > 0$, we define the cutoff function:
\begin{equation}\label{eq:cutoffe}
\chi^{\omega, \varepsilon}_{[h, N]}(\underline{k}) :=  ( 1 - \chi^{\varepsilon}(2^{-h} \| \underline{k} \|_{\omega}) ) \chi^{\varepsilon}(2^{-N} \| \underline{k} \|_{\omega})\;.
\end{equation}
As $\varepsilon \to 0$, $\chi^{\omega, \varepsilon}_{[h, N]}(\underline{k}) \to \chi^{\omega}_{[h, N]}(\underline{k})$, a function supported on $\underline{k}$ such that $2^{h-1} \leq \| \underline{k} \|_{\omega} \leq 2^{N+1}$. The function in (\ref{eq:cutoffe}) introduces an infrared and an ultraviolet cutoff.

\paragraph{Fields and propagators.} We define a $1+1$ dimensional Grassmann field starting from its (finitely many) Fourier components. In the following, we are interested in defining a suitable lattice version of the Grassmann field with relativistic propagator $Z_{\omega}^{-1}\chi^{\omega, \varepsilon}_{[h, N]}(\underline{k}) / (-ik_{0} + v_{\omega} k_{1})$. 

To each $\underline{k} \in \mathbb{D}_{L,\frak{a}}^{\text{a}}$, we associate Grassmann variables $\psi^{\pm}_{\underline{k}, \omega}$, $\omega = 1, \ldots, n_{\text{e}}$. We extend periodically the field $\hat \psi_{\underline{k}, \omega}$ to the whole $\frac{2\pi}{L} ( \mathbb{Z} + 1/2 )^{2}$ by:
\begin{equation}
\hat \psi^{\pm}_{\underline{k} + a_{1}\underline{G}_{1} + a_{2} \underline{G}_{2}, \omega} := \hat\psi^{\pm}_{\underline{k}, \omega}\;.
\end{equation}
For $\underline{x} \in \frak{a} \mathbb{Z}^{2}$, we define the configuration-space field, as:
\begin{equation}\label{eq:confpsi}
\psi^{\pm}_{\underline{x}, \omega} := \frac{1}{L^{2}} \sum_{\underline{k} \in \mathbb{D}^{\text{a}}_{L,a}} e^{\pm ik\cdot x} \hat \psi^{\pm}_{\underline{k}, \omega}\;;
\end{equation}
the configuration-space field satisfies the following antiperiodicity condition:
\begin{equation}
\psi^{\pm}_{\underline{x} + a_{1} L \underline{e}_{1} + a_{2} L \underline{e}_{2}, \omega} := (-1)^{a_{1} + a_{2}} \psi^{\pm}_{\underline{x}, \omega}\;.
\end{equation}
Finally, using that:
\begin{equation}
\sum_{\underline{x} \in \Lambda_{L,\frak{a}}} e^{i\underline{k}\cdot \underline{x}} = \sum_{\underline{J} \in (\frak{a} \mathbb{Z}^{2})^{*}} \delta_{\underline{k} + \underline{J}, \underline{0}} \frak{N}^{2}\;,
\end{equation}
with $\delta_{\underline{k}, \underline{0}}$ the Kronecker delta function, the relation in (\ref{eq:confpsi}) can be inverted as:
\begin{equation}
\hat \psi^{\pm}_{\underline{k}, \omega} = \frak{a}^{2} \sum_{\underline{x} \in \Lambda_{L, \frak{a}}} e^{\mp i k\cdot x} \psi^{\pm}_{\underline{x},\omega}\;,\qquad \forall \underline{k} \in \mathbb{D}_{L, \frak{a}}^{\text{a}}\;.
\end{equation}
Associated to the Grassmann field, we define the regularized, momentum space propagator as:
\begin{eqnarray}\label{eq:latchi}
\hat g^{[h, N]}_{\omega, \e, \frak{a}}(\underline{k}) &:=& \frac{1}{Z_{\omega}}\frac{\chi^{\omega, \e}_{[h,N]}(\underline{k})}{\frak{D}_{\omega, \frak{a}}(\underline{k})}\;,\qquad \forall \underline{k} \in \mathbb{D}^{\text{a}}_{L, \frak{a}} \nonumber\\
\frak{D}_{\omega, \frak{a}}(\underline{k}) &:=& -i \frac{1}{\frak{a}} \sin(\frak{a} k_{0}) + v_{\omega} \frac{1}{\frak{a}} \sin(\frak{a} k_{1})\;,
\end{eqnarray}
for suitable real parameters $Z_{\omega}$, $v_{\omega}$, $\omega = 1,\ldots, n_{\text{e}}$, to be chosen later on. Notice that:
\begin{equation}\label{eq:perio}
\hat g^{[h, N]}_{\omega, \e, \frak{a}}(\underline{k}) = \hat g^{[h, N]}_{\omega, \e, \frak{a}}(\underline{k} + a_{1} \underline{G}_{1} + a_{2} \underline{G}_{2})\;.
\end{equation}
Thanks to this periodicity property, we extend $\hat g^{[h, N]}_{\omega, \e, \frak{a}}(\underline{k})$ to the whole $\frac{2\pi}{L} ( \mathbb{Z} + 1/2 )^{2}$.

The denominator of $\hat g^{[h, N]}_{\omega, \e, \frak{a}}(\underline{k})$ vanishes at points $\underline{k}$ such that:
\begin{equation}\label{eq:singpts}
\underline{k} = (0, 0)\;,\quad (\pi/ \frak{a}, 0)\;,\quad (0, \pi/\frak{a})\;,\quad (\pi/ \frak{a}, \pi/\frak{a})
\end{equation}
and all translations of these points by $a_{1} \underline{G}_{1} + a_{2} \underline{G}_{2}$. It is important to notice that the points in (\ref{eq:singpts}) do not belong\footnote{The number $0$ cannot be realized as $m + 1/2$ with $m$ integer. Also, the number $\pi/\frak{a}$ cannot be realized as $\frac{\pi (2m + 1)}{L}$; this is equivalent to $\frak{N} = 2m + 1$, which is impossible since $\frak{N}$ is even by assumption.} to $\mathbb{D}_{L,\frak{a}}^{\text{a}}$: thus, for $\frak{a} > 0$, the propagator is a bounded function. We see that, as $\frak{a}\to 0$, the lattice propagator in (\ref{eq:latchi}) correctly approximates the relativistic propagator we are interested in. However, there are three other singular points in this limit, which are not present in the relativistic approximation; notice however that for $\frak{a}$ small enough the last three points in (\ref{eq:singpts}) are outside the support of $\chi^{\omega}_{[h,N]}(\underline{k})$. More precisely, for $\frak{a}$, $\varepsilon$ small enough, the contribution from these points to the lattice propagator is arbitrarily small. That is, in the limit $\varepsilon \to 0$, $\frak{a}\to 0$, the singular behavior of the momentum space propagator is only due to the point $(0,0)$:
%
\begin{eqnarray}
\lim_{\varepsilon \to 0}\lim_{\frak{a}\to 0}\hat g^{[h, N]}_{\omega, \e, \frak{a}}(\underline{k}) &=& \frac{1}{Z_{\omega}}\frac{\chi^{\omega}_{[h,N]}(\underline{k})}{-ik_{0} + v_{\omega} k_{1}} \nonumber\\
&=:& \hat g^{[h, N]}_{\omega}(\underline{k})\;.
\end{eqnarray}
Next, we then define the configuration-space propagator as:
\begin{equation}\label{eq:realg}
g^{[h,N]}_{\omega, \e, \frak{a}}(\underline{x}) := \frac{1}{L^{2}} \sum_{\underline{k} \in \mathbb{D}^{\text{a}}_{L, \frak{a}}} e^{ik\cdot x} \hat g^{[h,N]}_{\omega, \e, \frak{a}}(\underline{k})\;,\qquad \forall \underline{x} \in \frak{a} \mathbb{Z}^{2}\;.
\end{equation}

\paragraph{Generating functional of correlations.} Having defined fields and propagators, we introduce the Grassmann Gaussian integration as follows:
\begin{eqnarray}\label{eq:Cpsi}
\mu_{[h,N]}(d\psi) &:=& \mathcal{N}^{-1} \Big[ \prod_{\omega = 1}^{n_{\text{e}}} \prod_{\underline{k} \in \mathbb{D}^{\text{a}}_{L, \frak{a}}} d\hat\psi_{\underline{k},\omega}^{+} d\hat\psi^{-}_{\underline{k},\omega}\Big] e^{-C(\psi)} \nonumber\\
C(\psi) &:=& \frac{1}{L^{2}} \sum_{\omega}\sum_{\underline{k} \in \mathbb{D}^{\text{a}}_{L, \frak{a}}} \hat \psi^{+}_{\underline{k}, \omega} g^{[h,N]}_{\omega, \varepsilon, \frak{a}}(\underline{k})^{-1} \hat \psi^{-}_{\underline{k},\omega}\;,
\end{eqnarray}
where the parameter $\mathcal{N}$ is a normalization factor, chosen so that:
\begin{equation}
\int \mu_{[h,N]}(d\psi) = 1\;.
\end{equation}
Thanks to $\chi^{\omega, \e}_{[h,N]}(\underline{k}) > 0$, the covariance $C(\psi)$ is well defined. From the rules of Grassmann integration, it is easy to see that:
\begin{equation}
\int \mu_{[h,N]}(d\psi)\, \hat \psi^{-}_{\underline{k}, \omega} \hat \psi^{+}_{\underline{q}, \omega'} = L^{2} \delta_{\underline{k},\underline{q}} \delta_{\omega, \omega'} \hat g^{[h,N]}_{\omega, \e, \frak{a}}(\underline{k})\;.
\end{equation}
This implies that, by linearity of the Gaussian expectation:
\begin{equation}
\int \mu_{[h,N]}(d\psi)\, \psi^{-}_{\underline{x}, \omega} \psi^{+}_{\underline{y}, \omega'} = \delta_{\omega, \omega'} g_{\omega, \e, \frak{a}}^{[h,N]}(\underline{x} - \underline{y})\;.
\end{equation}
Next, we introduce the many-body interaction as:
\begin{equation}\label{eq:Vdef}
V(\psi) := \frac{\frak{a}^{4}}{2} \sum_{\underline{x}, \underline{y} \in \Lambda_{L, \frak{a}}} \sum_{\omega, \omega'} \lambda_{\omega,\omega'}Z_{\omega} Z_{\omega'} n_{\underline{x}, \omega} n_{\underline{y},\omega'} v(\underline{x} - \underline{y})
\end{equation}
where: $\lambda_{\omega, \omega'}\in \mathbb{R}$, $\lambda_{\omega, \omega'} = \lambda_{\omega',\omega}$, $\lambda_{\omega,\omega} = 0$; $v(\cdot)$ satisfies the periodicity condition:
\begin{equation}\label{eq:periov}
v(\underline{x} + a_{1} L \underline{e}_{1} + a_{2} L \underline{e}_{2}) = v(\underline{x})\;;
\end{equation}
and $n_{\underline{x},\omega}$ is the Grassmann counterpart of the density operator, $n_{\underline{x},\omega} = \psi^{+}_{\underline{x},\omega} \psi^{-}_{\underline{x},\omega}$. We will suppose that
\begin{equation}
|\lambda_{\omega, \omega'}| \leq C|\lambda|\;,
\end{equation}
for $\lambda$ small enough. Without loss of generality, we will also assume that $\hat v(0) = 1$. Let $\mathbb{D}^{\text{p}}_{L}$ be the set of momenta compatible with the periodicity condition (\ref{eq:periov}):
\begin{equation}
\mathbb{D}_{L}^{\text{p}} := \Big\{ \underline{p} = \Big(\frac{2\pi}{L} n_{0}, \frac{2\pi}{L} n_{1}\Big)\; \Big|\; n_{i} \in \mathbb{Z} \Big\}\;.
\end{equation}
We have:
\begin{equation}
v(\underline{x}) = \frac{1}{L^{2}} \sum_{\underline{p} \in \mathbb{D}_{L}^{\text{p}}} e^{-i\underline{p} \cdot \underline{x}} \hat v(\underline{p})\;,\qquad \hat v(\underline{p}) = \frak{a}^{2}\sum_{x\in \Lambda_{L}} e^{i\underline{p} \cdot x} v(\underline{x})\;.
\end{equation}
We will suppose that the function $\hat v(\cdot)$ is the restriction to $\mathbb{D}_{L}^{\text{p}}$ of an even, smooth and compactly supported function on $\mathbb{R}^{2}$. Thus, denoting by $\| \cdot \|$ the periodic distance on $\frak{a} \mathbb{Z}^{2}$:
\begin{equation}
\| \underline{x} \| := \inf_{a_{1}, a_{2} \in \mathbb{Z}} | \underline{x} - a_{1} L \underline{e}_{1} - a_{2} L \underline{e}_{2} |\;,
\end{equation}
we easily get:
\begin{equation}
\| \underline{x} \|^{n} |v(\underline{x})| \leq C_{n}\;, \qquad \forall n\in \mathbb{N}\;.
\end{equation}
In Fourier space, the many-body interaction reads:
\begin{equation}
V(\psi) = \frac{1}{2 L^{2}} \sum_{\underline{p} \in \mathbb{D}_{L}^{\text{p}}} \sum_{\omega, \omega'} \lambda_{\omega, \omega'} Z_{\omega} Z_{\omega'} \hat v(\underline{p}) \hat n_{\underline{p},\omega} \hat n_{-\underline{p}, \omega'}\;,
\end{equation}
with momentum-space Grassmann density:
\begin{equation}
\hat n_{\underline{p}, \omega} = \frak{a}^{2} \sum_{\underline{x}\in \Lambda_{L, \frak{a}}} e^{i\underline{p}\cdot \underline{x}} n_{\underline{x}, \omega} = \frac{1}{L^{2}} \sum_{\underline{k} \in \mathbb{D}^{\text{a}}_{L, \frak{a}}} \hat \psi^{+}_{\underline{k} - \underline{p}, \omega} \hat \psi^{-}_{\underline{k}, \omega}\;.
\end{equation}
With these notations, we define the partition function of the theory as:
\begin{equation}\label{eq:part}
\mathcal{Z}_{L, \frak{a}, \e, h, N} = \int \mu_{[h,N]}(d\psi)\, e^{-V(\psi)}\;.
\end{equation}
As long as all cutoffs are kept finite, the right-hand side of (\ref{eq:part}) is a finite sum of Grassmann monomials, with bounded coefficients. Next, let us introduce the generating functional of correlations. To this end, we define the source terms as:
\begin{eqnarray}
\Gamma(\psi; A) &:=& \sum_{\omega = 1}^{n_{\text{e}}} \sum_{\mu = 0,1} \sum_{x_{2} = 0}^{L-1} \frak{a}^{2} \sum_{\underline{x}\in \Lambda_{L, \frak{a}}} Z_{\mu,\omega}(x_{2}) A_{\underline{x}, \mu, \omega}(x_{2}) n_{\underline{x},\omega}\\
B(\psi; \phi) &:=& \sum_{\omega = 1}^{n_{\text{e}}} \sum_{x_{2} = 0}^{L-1} \frak{a}^{2} \sum_{\underline{x} \in \Lambda_{L, \frak{a}}} \Big[ \psi^{+}_{\underline{x},\omega} Q_{\omega}(x_{2}) \phi^{-}_{\underline{x},\omega}(x_{2}) + \phi^{+}_{\underline{x},\omega}(x_{2}) \overline{Q_{\omega}(x_{2})} \psi^{-}_{\underline{x},\omega}\Big]\;,\nonumber
\end{eqnarray}
where $A_{\underline{x},\mu,\omega}(x_{2})$, $\phi^{\pm}_{\underline{x}, \omega}$ are respectively a real-valued and a Grassmann external field, of the form: 
\begin{equation}
A_{\underline{x},\mu,\omega}(x_{2}) = \frac{1}{L^2} \sum_{\underline{p} \in \mathbb{D}_{L}^{\text{p}}} e^{-i\underline{p}\cdot \underline{x}} \hat A_{\underline{p}, \mu, \omega}(x_{2})\;,\qquad \phi^{\pm}_{\underline{x}, \omega}(x_{2}) = \frac{1}{L^2} \sum_{\underline{k} \in \mathbb{D}^{\text{a}}_{L, \frak{a}}} e^{\pm ik\cdot x} \hat \phi^{\pm}_{\underline{k}, \omega}(x_{2})\;.
\end{equation}
Equivalently, the source terms can be written in momentum space as:
\begin{eqnarray}\label{eq:sourceFourier}
\Gamma(\psi; A) &=& \sum_{\omega = 1}^{n_{\text{e}}} \sum_{\mu = 0,1} \sum_{x_{2} = 0}^{L-1} \frac{1}{L^{2}} \sum_{\underline{p}\in \mathbb{D}^{\text{p}}_{L}} Z_{\mu,\omega}(x_{2}) \hat A_{-\underline{p}, \mu, \omega}(x_{2}) \hat n_{\underline{p},\omega} \\
B(\psi; \phi) &=& \sum_{\omega = 1}^{n_{\text{e}}} \sum_{x_{2} = 0}^{L-1} \frac{1}{L^{2}} \sum_{\underline{k} \in \mathbb{D}^{\text{a}}_{L, \frak{a}}} \Big[ \hat \psi^{+}_{\underline{k},\omega} Q_{\omega}(x_{2}) \hat \phi^{-}_{\underline{k},\omega}(x_{2}) + \hat \phi^{+}_{\underline{k},\omega}(x_{2}) \overline{Q_{\omega}(x_{2})} \hat \psi^{-}_{\underline{k},\omega}\Big]\;.\nonumber
\end{eqnarray}
Finally, we set:
\begin{eqnarray}
\mathcal{Z}_{L, \frak{a}, \e, h, N}(A, \phi) &:=& \int \mu_{[h,N]}(d\psi)\, e^{-V(\psi) + \Gamma(\psi; A) + B(\psi; \phi)} \nonumber\\
\mathcal{W}_{L, \frak{a}, \e, h, N}(A, \phi) &:=&  \log \frac{\mathcal{Z}_{L, \frak{a}, \e, h, N}(A, \phi)}{\mathcal{Z}_{L, \frak{a}, \e, h, N}(0, 0)}\;.
\end{eqnarray}
The derivatives with respect to $A$ and to $\phi^{\pm}$ of $\mathcal{W}_{L, \frak{a}, \e, h, N}$ define the connected correlation functions of the model. We shall be interested in the $L\to \infty$, $\e\to 0$, $\frak{a} \to 0$ limit of such correlations.

%
%

%

\section{Renormalization group analysis of the reference model}\label{sec:RG}

As for the original lattice model, the generating functional of correlations of the reference model is evaluated via a multiscale procedure, which we shall outline here. The content of this section is a straightforward adaptation of the analysis that has already been done for the Luttinger model; see {\it e.g.} \cite{BFM, BFM3}. We repeat it for completeness, referring to the literature for technical details. We will consider separately the ultraviolet and the infrared regime.

\subsection{Ultraviolet regime}

We write:
\begin{equation}
\chi^{\omega,\e}_{[h,N]}(\underline{k}) = \chi^{\omega,\e}_{[h,0]}(\underline{k}) + \sum_{j = 1}^{N} f_{j}^{\omega, \e}(\underline{k})\;,
\end{equation}
with $f_{j}^{\omega, \e}(\underline{k}) = \chi^{\omega,\e}_{[h,j]}(\underline{k}) - \chi^{\omega,\e}_{[h,j-1]}(\underline{k})$, and we decompose the propagator and the fields as:
\begin{equation}
\hat g^{[h,N]}_{\omega, \e, \frak{a}}(\underline{k}) = \hat g^{[h, 0]}_{\omega, \e, \frak{a}}(\underline{k}) + \sum_{j = 1}^{N} \hat g^{(j)}_{\omega, \e, \frak{a}}(\underline{k})\;,\qquad \hat \psi^{\pm}_{\underline{k},\omega} = \hat\psi^{(\leq 0)\pm}_{\underline{k},\omega} + \sum_{j=1}^{N} \hat\psi^{(j)\pm}_{\underline{k},\omega}\;,
\end{equation}
with $\hat g^{[h, 0]}_{\omega, \e, \frak{a}}(\underline{k})$, $\hat g^{(j)}_{\omega, \e, \frak{a}}(\underline{k})$ given by $\hat g^{[h,N]}_{\omega, \e, \frak{a}}(\underline{k})$ after replacing $\chi^{\omega,\e}_{[h,N]}(\underline{k})$ by $\chi^{\omega,\e}_{[h,0]}(\underline{k})$, $f_{j}^{\omega, \e}(\underline{k})$, respectively. In configuration space, the single-scale propagators satisfy the following bounds, for universal constants $C_{n}>0$, and for all $j\geq 1$:
\begin{equation}
(2^{j}\| x \|)^{n} | g^{(j)}_{\omega, \e, \frak{a}}(x) | \leq C_{n} 2^{j}\;,\qquad \forall n\in \mathbb{N}\;,\quad \forall \underline{x} \in \frak{a} \mathbb{Z}^{2}\;.
\end{equation}
From this bound, we easily get:
\begin{equation}\label{eq:singscaleuv}
\| g^{(j)}_{\omega, \e, \frak{a}} \|_{\ell^\infty(\Lambda_{L, \frak{a}})} \leq C2^{j}\;,\qquad \| g^{(j)}_{\omega, \e, \frak{a}} \|_{\ell^{1}(\Lambda_{L, \frak{a}})} \leq C 2^{-j}\;.
\end{equation}
By the addition principle of Grassmann variables,
\begin{equation}
\mathcal{Z}_{L, \frak{a}, \e, h, N}(A, \phi) = \int \mu_{[h,0]}(d\psi^{(\leq 0)}) \int \mu_{1}(d\psi^{(1)}) \cdots \int \mu_{N}(d\psi^{(N)}) e^{V (\psi^{(\leq 0)} + \psi^{(1)} + \ldots + \psi^{(N)}; A, \phi)}\;,
\end{equation}
where:
\begin{equation}
V (\psi; A, \phi) = -V(\psi) + \Gamma(\psi; A) + B(\psi; \phi)\;.
\end{equation}
We integrate the single-scale fields in an iterative fashion, from the scale $N$ until the scale $0$. Every step of integration is performed via fermionic cluster expansion, see {\it e.g.} \cite{GM} for a review. The convergence of the expansion is guaranteed by the bounds (\ref{eq:singscaleuv}). We get, for any $j\geq 0$:
\begin{equation}
\mathcal{Z}_{L, \frak{a}, \e, h, N}(A, \phi) = e^{\mathcal{W}^{(j)}(A,\phi)} \int \mu_{[h,j]}(d\psi)  e^{V^{(j)}(\psi; A, \phi)}\;,
\end{equation}
where the prefactor is chosen so that $V^{(j)}(0; A, \phi) = 0$. The effective potential on scale $j$ has the form, for suitable kernels $W^{(j)}_{\Gamma}$:
\begin{equation}\label{eq:N-1}
V^{(j)}(\psi; A,\phi)  = \sum_{\Gamma} \frak{a}^{2|\Gamma|} \sum_{\underline{X}, \underline{Y}, \underline{Z}}  \psi_{\Gamma}(\underline{X}) A_{\Gamma}(\underline{Y}) \phi_{\Gamma}(\underline{Z})\, W^{(j)}_{\Gamma}(\underline{X}, \underline{Y}, \underline{Z})\;,
\end{equation}
where: $\underline{X}, \underline{Y}, \underline{Z}$ denote collective position variables, e.g. $\underline{X} = (\underline{x}_{1}, \ldots, \underline{x}_{n})$, with $n$ the order of the $\psi$-monomial; $\Gamma = (\Gamma_{\psi}, \Gamma_{A}, \Gamma_{\phi})$ collects the indices labelling the monomials in the variables $\psi$, $A$, $\phi$, respectively. Also,
\begin{equation}
\psi_{\Gamma}(\underline{X})  \equiv \psi_{\Gamma_{\psi}}(\underline{X})\;,\qquad  A_{\Gamma}(\underline{Y}) \equiv A_{\Gamma_{A}}(\underline{Y})\;,\qquad \phi_{\Gamma}(\underline{Z})\equiv \phi_{\Gamma_{\phi}}(\underline{Z})
\end{equation}
with:
\begin{equation}
\psi_{\Gamma_{\psi}}(\underline{X}) = \prod_{f\in \Gamma_{\psi}} \psi^{\varepsilon(f)}_{\omega(f), \underline{x}(f)}\;,\; A_{\Gamma_{A}}(\underline{Y}) = \prod_{f\in \Gamma_{A}} A_{\mu(f), \underline{y}(f)}(x_{2}(f))\;,\; \phi_{\Gamma_{\phi}}(\underline{Z}) = \prod_{f\in \Gamma_{\phi}} \phi^{\varepsilon(f)}_{\omega(f), \underline{z}(f)}(z_{2}(f))\;.
\end{equation}
The generating functional on scale $j$, $\mathcal{W}^{(j)}(A,\phi)$, has a form similar to (\ref{eq:N-1}), but with $\psi \equiv 0$. Let us define the weighted norms:
\begin{equation}
\| W_{\Gamma}^{(j)} \|_{1,k} := \frac{1}{|\Lambda_{L,\frak{a}}|} \frak{a}^{2|\Gamma|} \sum_{\underline{X}, \underline{Y}, \underline{Z}} \big | W_{\Gamma}^{(j)}(\underline{X}, \underline{Y}, \underline{Z}) \big| w_{k}(\underline{X}, \underline{Y}, \underline{Z})\;,
\end{equation}
where, setting $(\underline{X}, \underline{Y}, \underline{Z}) \equiv \underline{Q} = (\underline{q}_{1}, \ldots, \underline{q}_{|\Gamma|})$:
\begin{equation}
w_{k}(\underline{X}, \underline{Y}, \underline{Z}) = \sum^{*}_{\{m_{ij}\}}\prod_{i<j} \| \underline{q}_{i} - \underline{q}_{j} \|^{m_{ij}}\;;
\end{equation}
the sum is over all natural numbers $m_{ij}$, labelled by $1\leq i,j\leq |\Gamma|$, and the asterisk denotes the constraint $\sum_{ij} m_{ij} = k$. The kernels can be estimated inductively. On scale $N-1$, the bounds (\ref{eq:singscaleuv}) allow to show that:
\begin{equation}\label{eq:estN-1}
\| W_{\Gamma}^{(N-1)} \|_{1,k} \leq C_{\Gamma, k} 2^{-k(N-1)}2^{(N-1) D(\Gamma)} \Big[\prod_{f\in \Gamma_{A}} |Z_{\mu(f),\omega(f)}(y_{2}(f))|\Big]\Big[ \prod_{f\in \Gamma_{\phi}} |Q^{\varepsilon(f)}_{\omega(f)} (z_{2}(f))|\Big]\;,
\end{equation}
where $D(\Gamma)$ is the {\it scaling dimension:}
\begin{equation}
D(\Gamma) := 2 - \frac{|\Gamma_{\psi}|}{2} - |\Gamma_{A}| - \frac{3 |\Gamma_{\phi}|}{2}\;.
\end{equation}
It turns out that the above estimate cannot be reproduced on all scales $j\geq 0$. The next proposition provides a more precise bound, which can be iterated on all scales $j\geq 0$. It has been proven in \cite{BFM}, Theorem 3.1, for the case of fermions with two chiralities, see also \cite{MaQED}, Lemma 2, for QED in $1+1$ dimensions. The proof can be immediately adapted to our setting, and it will be omitted.
\begin{prop}[Improved ultraviolet bounds.]\label{prp:gainUV} For $|\lambda_{\omega,\omega'}| \leq |\lambda|$ with $|\lambda|$ small enough uniformly in all cutoff parameters, the following is true. Let $\Gamma$ be such that $D(\Gamma)\geq 0$ and $\Gamma_{\psi} \neq \emptyset$. For all $j\in [0,N]$, there exists $C>0$ such that:
\begin{eqnarray}\label{eq:bduv}
\| W_{\Gamma}^{(j)} - W_{\Gamma}^{(N)} \|_{1,k} \leq C |\lambda| 2^{-kj}2^{j (D(\Gamma) - \theta(\Gamma))} \Big[\prod_{f\in \Gamma_{A}} |Z_{\mu(f),\omega(f)}(y_{2}(f))|\Big]\Big[ \prod_{f\in \Gamma_{\phi}} |Q^{\varepsilon(f)}_{\omega(f)} (z_{2}(f))|\Big]
\end{eqnarray}
where: $\theta(\Gamma) = 2$ if $\Gamma = \Gamma_{\psi}$ and $|\Gamma_{\psi}| = 2$; $\theta(\Gamma) = 1$ if $|\Gamma_{A}| = 1$, $|\Gamma_{\psi}| = 2$; $\theta(\Gamma) = 1$ if $\Gamma = \Gamma_{\psi}$, $|\Gamma_{\psi}| = 4$. For all the other kernels, the estimate (\ref{eq:estN-1}) with $N-1$ replaced by $j$ holds true. 
\end{prop} 
Proposition \ref{prp:gainUV} shows that the scaling dimension of all monomials is actually {\it strictly negative} in the ultraviolet regime: all kernels are irrelevant in the renormalization group sense. 
\begin{rem}
As discussed in \cite{BFM, MaQED}, the key ingredient for the dimensional improvement of (\ref{eq:bduv}) is the vanishing of the bubble diagram in the continuum limit. We have:
\begin{equation}\label{eq:bubblee}
\lim_{L\to \infty} \lim_{\e\to 0} \lim_{\frak{a} \to 0} \frak{a}^{2} \sum_{\underline{x} \in \Lambda_{L, \frak{a}}} g^{[j,N]}_{\omega, \e, \frak{a}}(\underline{x})^{2} = \int_{\mathbb{R}^{2}} \frac{d\underline{k}}{(2\pi)^{2}}\, g^{[j,N]}_{\omega}(\underline{k})^{2}\;.
\end{equation}
Performing the change of variables $v_{\omega} k_{1} \to k_{1}$, and setting $\chi_{[j,N]}(\underline{k}) \equiv \chi^{\omega}_{[j,N]}(\underline{k})|_{v_{\omega} = 1}$, we get:
\begin{equation}
(\ref{eq:bubblee}) = \frac{1}{Z_{\omega}|v_{\omega}|} \int_{\mathbb{R}^{2}} \frac{d\underline{k}}{(2\pi)^{2}} \frac{(\chi_{[j,N]}(\underline{k}))^{2}}{(-ik_{0} + k_{1})^{2}}\;.
\end{equation}
Consider the rotations $\underline{k} \mapsto R_{\alpha} \underline{k}$ , with $R_{\alpha}$ the counterclockwise rotation with angle $\alpha$. The Jacobian of the transformation is $1$, and $\chi_{[h,N]}$ is invariant; instead, the function $1/(-ik_{0} + k_{1})$ transforms as $1/(-ik_{0} + k_{1}) \to e^{-i\alpha} /(-ik_{0} + k_{1})$. Choosing $\alpha\neq \pi$, this symmetry immediately implies:
\begin{equation}
\int \frac{d\underline{k}}{(2\pi)^{2}} \frac{(\chi_{[j,N]}(\underline{k}))^{2}}{(-ik_{0} + k_{1})^{2}} = 0\;.
\end{equation}
\end{rem}
%
%
\subsection{Infrared regime: the partition function}\label{sec:RGIR}
The starting point is the identity obtained after the integration of the ultraviolet degrees of freedom:
\begin{equation}
\mathcal{Z}_{L, \frak{a}, \e, h, N}(A, \phi) = e^{\mathcal{W}^{(0)}(A,\phi)} \int \mu_{[h,0]}(d\psi)\, e^{V^{(0)}(\psi; A, \phi)}\;.
\end{equation}
We will start by considering the partition function, corresponding to $A = \phi = 0$. Later, we will discuss how to adapt the analysis to the case $A,\phi \neq 0$. We will use the following notations:
\begin{equation}
\mathcal{W}^{(j)} \equiv \mathcal{W}^{(j)}(0,0)\;,\qquad V^{(j)}(\psi) \equiv V^{(j)}(\psi, 0, 0)\;.
\end{equation}
We shall suppose inductively that the partition function can be rewritten as, for all scales $r\geq j > h$:

\begin{equation}
\mathcal{Z}_{L, \frak{a}, \e, h, N}(0,0) = e^{\mathcal{W}^{(r)}} \int \mu_{[h,r]}(d\psi)\, e^{V^{(r)}(\sqrt{Z_{r}}\psi)}\;.
\end{equation}
The Grassmann Gaussian integration $\mu_{[h,r]}(d\psi)$ has covariance:
\begin{eqnarray}
\hat g^{[h,r]}_{\omega, \e, \frak{a}}(\underline{k}) &=& \frac{1}{Z_{r,\omega}(\underline{k})} \frac{\chi^{\omega,\e}_{[h,r]}(\underline{k})}{\frak{D}^{(r)}_{\omega, \e, \frak{a}}(\underline{k})} \nonumber\\
\frak{D}^{(r)}_{\omega, \e, \frak{a}}(\underline{k}) &=&  -i \frac{1}{\frak{a}} \sin(\frak{a} k_{0}) + v_{r,\omega}(\underline{k}) \frac{1}{\frak{a}} \sin(\frak{a} k_{1})\;, 
\end{eqnarray}
where the renormalizations $Z_{r,\omega}(\underline{k})$, $v_{r,\omega}(\underline{k})$ are even functions of $\underline{k}$; they satisfy a recursion relation, that will be obtained below. The effective interaction is given by a (finite) linear combination of Grassmann monomials:
\begin{equation}
V^{(r)}(\sqrt{Z_{r}}\psi) = \sum_{\Gamma} \frak{a}^{2|\Gamma|} \sum_{\underline{X}} \Big[\prod_{f\in \Gamma} \sqrt{Z_{r,\omega(f)}}\Big] \psi_{\Gamma}(\underline{X}) W^{(r)}_{\Gamma}(\underline{X})\;,
\end{equation}
for suitable kernels to be determined inductively. They are analytic functions of $\lambda_{\omega, \omega'}$, and they are translation-invariant in the space-time arguments. We shall suppose that:
\begin{equation}\label{eq:bdir}
\| W_{\Gamma}^{(r)} \|_{1,k} \leq C_{\Gamma, k} |\lambda| 2^{-k r}2^{r D(\Gamma)}\;.
\end{equation}
All these assumptions are true on scale $0$. In order to prepare the integration of the scale $j$, we define a localization and renormalization procedure, as follows. We write:
\begin{equation}\label{eq:LR0}
V^{(j)}(\sqrt{Z_{j}}\psi) = \mathcal{L} V^{(j)}(\sqrt{Z_{j}}\psi) + \mathcal{R} V^{(j)}(\sqrt{Z_{j}}\psi)\;,
\end{equation}
for suitable operators $\mathcal{L}$ and $\mathcal{R}$, that we shall now introduce. We rewrite the effective action on scale zero in Fourier space as:
\begin{equation}\label{eq:potFu}
V^{(j)}(\sqrt{Z_{j}}\psi) = \sum_{n\geq 1} \frac{1}{L^{4n}} \sum_{\underline{K}, \underline{\omega}} \Big[ \prod_{i=1}^{2n} \sqrt{Z_{j,\omega_{i}}} \Big]\Big[\prod_{i=1}^{2n} \hat \psi^{+}_{\underline{k}_{2i-1}, \omega_{2i-1}} \hat \psi^{-}_{\underline{k}_{2i}, \omega_{2i}} \Big] \widehat W^{(j)}_{2n; \underline{\omega}}(\underline{K})\delta(\underline{K})\;,
\end{equation}
where: $\underline{K} = (\underline{k}_{1}, \ldots, \underline{k}_{2n})$, $\delta(\underline{K}) = L^{2}\delta_{L, \frak{a}}(\sum_{i=1}^{2n} (-1)^{i} \underline{k}_{i})$ with $\delta_{L,\frak{a}}$ the Kronecker delta function periodic on $\mathbb{D}^{\text{a}}_{L, \frak{a}}$, and
\begin{eqnarray}
\widehat{W}^{(j)}_{2n;\underline{\omega}}(\underline{K}) &:=& \frac{1}{L^{2}}\widehat{W}^{(j)}_{2n;\underline{\omega}}(\underline{k}_{1}, \underline{k}_{2}, \ldots, \underline{k}_{2n-1}, -\underline{k}_{1} - \ldots - \underline{k}_{2n-1}) \nonumber\\&\equiv& \widehat{W}^{(j)}_{2n;\underline{\omega}}(\underline{k}_{1}, \underline{k}_{2}, \ldots, \underline{k}_{2n-1})\;.
\end{eqnarray}
The expression (\ref{eq:potFu}) follows from the translation-invariance of the configuration-space kernels. We shall also set:
\begin{equation}
\widehat{W}^{(j)}_{2n;\underline{\omega};\infty}(\underline{K}) := \lim_{L\to \infty} \lim_{\varepsilon\to 0} \lim_{\frak{a}\to 0} \widehat{W}^{(j)}_{2n;\underline{\omega}}(\underline{K})\;.
\end{equation}
We shall suppose that $\widehat{W}^{(j)}_{2;\underline{\omega};\infty}(\underline{0}) = 0$. This assumption is true on scale $j=0$; it follows from $\hat g^{(\ell)}_{\omega}(\underline{k}) = -\hat g^{(\ell)}_{\omega}(-\underline{k})$  for all $\ell>0$, from $\hat v(\underline{p}) = \hat v(-\underline{p})$, and from the fact that $\widehat{W}^{(0)}_{2;\underline{\omega};\infty}(\underline{0})$ is given by a sum of Feynman graphs containing an odd number of fermionic propagators.

The localization operator $\mathcal{L}$ is a  linear operator acting on the kernels of the effective interaction $V^{(j)}$ as:
\begin{eqnarray}\label{eq:local}
\mathcal{L} \widehat W^{(j)}_{2;\underline{\omega}}(\underline{k}) &=& \widehat W^{(j)}_{2;\underline{\omega}; \infty}(\underline{0}) + \sum_{i=0}^{1} \frac{1}{\frak{a}} \sin (k_{i} \frak{a}) \partial_{i} \widehat W^{(j)}_{2;\underline{\omega};\infty}(\underline{0}) \nonumber\\
\mathcal{L} \widehat W^{(j)}_{4;\underline{\omega}}(\underline{k}_{1}, \underline{k}_{2}, \underline{k}_{3}) &=& \widehat W^{(j)}_{4;\underline{\omega};\infty}(\underline{0}, \underline{0}, \underline{0}) \nonumber\\
\mathcal{L} \widehat W^{(j)}_{2n;\underline{\omega}}(\underline{K}) &=& 0\qquad \text{otherwise.}
\end{eqnarray}
The role of the localization operator is to isolate the potentially dangerous contributions to the effective action, that tend to expand under the renormalization group iteration. The terms $\widehat W^{(j)}_{2;\underline{\omega};\infty}(\underline{0})$ are relevant, while the terms $\partial_{i} \widehat W^{(j)}_{2;\underline{\omega};\infty}(\underline{0})$, $\widehat W^{(j)}_{4;\underline{\omega};\infty}(\underline{0}, \underline{0}, \underline{0})$ are marginal. Instead, the operator $\mathcal{R}$ is defined as $1 - \mathcal{L}$; the kernels appearing in $\mathcal{R} V^{(j)}$ are by definition irrelevant. If we could replace $V^{(j)}$ by $\mathcal{R} V^{(j)}$, the inductive assumptions could be safely iterated on all scales.

By our inductive assumptions, the relevant terms are exactly zero.  Consider the marginal terms in the first of (\ref{eq:local}). By global chiral gauge symmetry, $\omega_{1} = \omega_{2}$ at the argument of $\widehat W^{(j)}_{2;\underline{\omega}}$, and we shall set $\widehat W^{(j)}_{2;\omega}\equiv \widehat W^{(j)}_{2;(\omega,\omega)}$. Let:
\begin{equation}
z_{j,0,\omega} := i\frac{\partial}{\partial k_{0}} \widehat W^{(j)}_{2,\omega;\infty}(\underline{0})\;,\qquad z_{j,1,\omega} :=  \frac{\partial}{\partial k_{1}} \widehat W^{(j)}_{2,\omega;\infty}(\underline{0})\;.
\end{equation}
Thanks to the estimates (\ref{eq:bdir}):
\begin{equation}
|z_{j,0,\omega} | \leq C|\lambda|\;,\qquad |z_{j,1,\omega} | \leq C|\lambda|\;.
\end{equation}
The corresponding contribution to the effective interaction is:
\begin{eqnarray}
\mathcal{L}_{2} V^{(j)}(\sqrt{Z_{j}}\psi) = \sum_{\omega} Z_{j,\omega} \frac{1}{L^{2}}\sum_{\underline{k} \in \mathbb{D}^{\text{a}}_{L, \frak{a}}} \hat \psi^{+}_{\underline{k}, \omega} \Big( -iz_{j,0,\omega} \frac{1}{\frak{a}}\sin(a k_{0}) + z_{j,1,\omega} \frac{1}{\frak{a}}\sin(a k_{1}) \Big) \hat \psi^{-}_{\underline{k},\omega}\;.\nonumber
\end{eqnarray}
This term has the same structure of the covariance of the integration $\mu_{[h,j]}(d\psi)$. It is convenient to reabsorb it in the Grassmann Gaussian integration, as follows:
\begin{eqnarray}
&&\mu_{[h,j]}(d\psi) \exp\Big\{ \sum_{\omega} Z_{j,\omega}  \frac{1}{L^{2}} \sum_{\underline{k} \in \mathbb{D}^{\text{a}}_{L, \frak{a}}} \hat \psi^{+}_{\underline{k}, \omega} \Big( -iz_{j,0,\omega} \frac{1}{\frak{a}}\sin(a k_{0}) + z_{j,1,\omega} \frac{1}{\frak{a}}\sin(a k_{1}) \Big) \hat \psi^{-}_{\underline{k},\omega}  \Big\} \nonumber\\
&&\qquad \equiv \frac{1}{n_{j}}\nu_{[h,j]}(d\psi)\;,
\end{eqnarray}
for a suitable normalization factor $n_{j}$. The new Grassmann Gaussian integration can be written as:
\begin{equation}
\nu_{[h,j]}(d\psi) = \frac{1}{\widetilde{\mathcal{N}}}d\psi\, e^{- \widetilde{C}_{[h,j]}(\psi)}\;,
\end{equation}
with:
\begin{eqnarray}
\widetilde{C}_{[h,j]}(\psi) &=& \sum_{\omega} \frac{1}{L^{2}} \sum_{\underline{k} \in \mathbb{D}^{\text{a}}_{L, \frak{a}}}\hat \psi^{+}_{\underline{k},\omega} \Big[ \hat g^{[h,j]}_{\omega, \e, \frak{a}}(\underline{k})^{-1} - ik_{0} Z_{j,\omega}z_{j,0,\omega} + k_{1} Z_{j,\omega}z_{j,1,\omega} \Big] \hat \psi^{-}_{\underline{k},\omega} \nonumber\\
&\equiv& \sum_{\omega}\frac{1}{L^{2}} \sum_{\underline{k} \in \mathbb{D}^{\text{a}}_{L, \frak{a}}} \hat \psi^{+}_{\underline{k},\omega} \tilde g^{[h,j]}_{\omega, \e, \frak{a}}(\underline{k})^{-1} \hat \psi^{-}_{\underline{k},\omega}\;,
\end{eqnarray}
where the new `dressed' propagator is given by:
\begin{eqnarray}
\tilde g^{[h,j]}_{\omega, \e, \frak{a}}(\underline{k}) &=& \frac{1}{Z_{j-1,\omega}(\underline{k})} \frac{\chi^{\omega,\e}_{[h,j]}(\underline{k})}{\frak{D}^{(j)}_{\omega, \e, \frak{a}}(\underline{k})}\nonumber\\
\frak{D}^{(j)}_{\omega, \e, \frak{a}}(\underline{k}) &=& -i \frac{1}{\frak{a}} \sin(\frak{a} k_{0}) + v_{j-1,\omega}(\underline{k}) \frac{1}{\frak{a}} \sin(\frak{a} k_{1})\;,
\end{eqnarray}
with new parameters:
\begin{eqnarray}\label{eq:flow0}
Z_{j-1,\omega}(\underline{k}) &=& Z_{j,\omega}(1 + z_{j,0,\omega} \chi_{[h,j]}^{\omega,\e}(\underline{k})) \nonumber\\
Z_{j-1,\omega}(\underline{k}) v_{j-1,\omega}(\underline{k}) &=& Z_{j,\omega} (v_{j,\omega} +  z_{j,1,\omega} \chi_{[h,j]}^{\omega,\e}(\underline{k}))\;.
\end{eqnarray}
For later use, we also define:
\begin{equation}\label{eq:flow1}
Z_{j-1,\omega} = Z_{j,\omega}(1 + z_{j,0,\omega} )\;,\qquad Z_{j-1,\omega} v_{j-1,\omega} = Z_{j,\omega} (v_{j,\omega} +  z_{j,1,\omega} )\;.
\end{equation}
\begin{rem}\label{rem:diffbeh}
Notice that for $\|\underline{k}\|_{\omega} > 2^{h}$, the functions (\ref{eq:flow0}) behave qualitatively as (\ref{eq:flow1}). Instead, for $\|\underline{k}\|_{\omega} \leq 2^{h}$, the behavior of (\ref{eq:flow0}) for $j=h$ is different from the one of (\ref{eq:flow1}); to see this, consider the $\varepsilon\to 0$ limit. We can rewrite (\ref{eq:flow0}) as:
\begin{eqnarray}\label{eq:flow0h}
Z_{j-1,\omega}(\underline{k}) &=& Z_{j,\omega}(\underline{k}) + Z_{j,\omega} z_{j,0,\omega} f_{h}^{\omega}(\underline{k})\nonumber\\
Z_{j-1,\omega}(\underline{k}) v_{j-1,\omega}(\underline{k}) &=& Z_{j,\omega}(\underline{k}) v_{j,\omega}(\underline{k}) +  Z_{j,\omega}z_{j,1,\omega} f_{h}^{\omega}(\underline{k})\;.
\end{eqnarray}
Iterating the expressions, we end up with:
\begin{eqnarray}\label{eq:subtle}
Z_{j-1,\omega}(\underline{k}) &=& Z_{\omega} + (Z_{j-1,\omega} - Z_{\omega}) f_{h}^{\omega}(\underline{k})\nonumber\\
Z_{j-1,\omega}(\underline{k})v_{h-1,\omega}(\underline{k})  &=& Z_{\omega} v_{\omega} + (Z_{j-1,\omega} v_{j-1,\omega} - Z_{\omega} v_{\omega}) f_{h}^{\omega}(\underline{k})\;.\label{eq:h-1}
\end{eqnarray}
Nevertheless, as we will comment later on, this will not affect the estimates for the single-scale propagators on scale $h$.
\end{rem}
Thus,
\begin{equation}\label{eq:resum}
\mu_{[h,j]}(d\psi) e^{V^{(j)}(\sqrt{Z}\psi)} = \frac{1}{n_{j}} \nu_{[h,j]}(d\psi) e^{\mathcal{L}_{4} V^{(j)}(\sqrt{Z_{j}}\psi) + \mathcal{R} V^{(j)}(\sqrt{Z_{j}}\psi)}\;,
\end{equation}
where $\mathcal{L}_{4} V^{(j)}$ only collects the quartic contributions to $\mathcal{L} V^{(j)}$. Before integrating the field on scale $j$, we rescale the effective interaction as follows:
\begin{equation}
\mathcal{L}_{4} V^{(j)}(\sqrt{Z_{j}}\psi) + \mathcal{R} V^{(j)}(\sqrt{Z_{j}}\psi) \equiv \mathcal{L}_{4} \widetilde V^{(j)}(\sqrt{Z_{j-1}} \psi) + \mathcal{R} \widetilde V^{(j)}(\sqrt{Z_{j-1}} \psi)\;.
\end{equation}
The relation between the new and the old kernels is:
\begin{equation}
\widetilde{W}^{(j)}_{\Gamma_{\psi}}(\underline{X}) := \Big[\prod_{f\in \Gamma_{\psi}} \sqrt{\frac{Z_{j,\omega(f)}}{Z_{j-1,\omega(f)}}} \Big] W^{(j)}_{\Gamma_{\psi}}(\underline{X})\;.
\end{equation}
In particular, the local quartic term is:
\begin{equation}\label{eq:loc4}
\mathcal{L}_{4} \widetilde V^{(j)}(\sqrt{Z_{j-1}} \psi) = \sum_{\underline{\omega}} \Big[\prod_{i=1}^{4} \sqrt{Z_{j-1,\omega_{i}}}\Big] \frak{a}^{2} \sum_{\underline{x}} \psi^{+}_{\underline{x},\omega_{1}} \psi^{-}_{\underline{x},\omega_{2}} \psi^{+}_{\underline{x},\omega_{3}}\psi^{+}_{\underline{x},\omega_{4}} \lambda_{j, \underline{\omega}}\;,
\end{equation}
with the effective coupling on scale $j$ given by:
\begin{equation}\label{eq:coupling0}
\lambda_{j, \underline{\omega}} = \Big[\prod_{i=1}^{4} \sqrt{ \frac{Z_{j,\omega_{i}}}{Z_{j-1,\omega_{i}}} } \Big] \widehat W^{(j)}_{4; \underline{\omega}}(\underline{0}, \underline{0}, \underline{0})\;.
\end{equation}
By global chiral gauge symmetry, $\lambda_{j,\underline{\omega}}$ is zero unless $\omega_{1} = \omega_{2}$ and $\omega_{3} = \omega_{4}$ or $\omega_{1} = \omega_{4}$ and $\omega_{2} = \omega_{3}$. We shall denote by $\lambda_{j, \omega, \omega'}$ the independent couplings. For $j=0$, one finds:
\begin{equation}
|\lambda_{0, \omega,\omega'} - \lambda_{\omega,\omega'}|\leq C|\lambda|^{2}\;.
\end{equation}
In general, the iterative integration gives rise to a recursion relation for the local quartic coupling:
\begin{equation}\label{eq:flow2}
\lambda_{j, \omega, \omega'} = \lambda_{j+1, \omega, \omega'} + \beta^{\lambda}_{j+1, \omega, \omega'}\;,
\end{equation}
where $\beta^{\lambda}_{j+1, \omega, \omega'}$ is the beta function of the marginal quartic coupling. Controlling the flow generated by (\ref{eq:flow2}) is a major task in our analysis. In terms of the effective coupling constant, we rewrite (\ref{eq:loc4}) as:
\begin{equation}\label{eq:L40}
\mathcal{L}_{4} \widetilde V^{(j)}(\sqrt{Z_{j-1}} \psi) = \sum_{\omega, \omega'} Z_{j-1,\omega} Z_{j-1, \omega'} \frak{a}^{2} \sum_{\underline{x}} \psi^{+}_{\underline{x},\omega} \psi^{-}_{\underline{x},\omega} \psi^{+}_{\underline{x},\omega'}\psi^{+}_{\underline{x},\omega'} \lambda_{j, \omega, \omega'}\;.
\end{equation}
We are now ready to integrate the scale $j$. We split:
\begin{equation}
\tilde g^{[h,j]}_{\omega, \e, \frak{a}}(\underline{k}) = \hat g^{[h,j-1]}_{\omega, \e, \frak{a}}(\underline{k}) + \hat g^{(j)}_{\omega, \e, \frak{a}}(\underline{k})\;,
\end{equation}
with:
\begin{equation}
\hat g^{(j)}_{\omega, \e, \frak{a}}(\underline{k}) = \frac{1}{Z_{j-1,\omega}(\underline{k})} \frac{f^{\omega, \e}_{j}(\underline{k})}{\frak{D}^{(j)}_{\omega, \e, \frak{a}}(\underline{k})}
\end{equation}
and where $f^{\omega, \e}_{j}(\underline{k}) = \chi^{\omega, \e}_{[h,j]}(\underline{k}) - \chi^{\omega, \e}_{[h,j-1]}(\underline{k})$. In the $L\to \infty$, $\varepsilon \to 0$, $\frak{a}\to 0$ limit, the propagator on scale $j$ is a compactly supported function, for $2^{j-1}\leq \|\underline{k}\|_{\omega} \leq 2^{j+1}$, and it satisfies the following estimate, for all scales $j\geq h$:
\begin{eqnarray}\label{eq:singscale}
\| \partial_{k}^{\ell}\hat g^{(j)}_{\omega, \infty}(\underline{k}) \|_{\infty} \leq C_{\ell}\frac{2^{-j(1 + \ell)}}{Z_{j-1,\omega}}\;, \qquad \forall \ell\geq 0\;.
\end{eqnarray}
This bound easily implies the following decay estimate for the configuration space propagator:
\begin{equation}\label{eq:bdxj}
2^{j \ell} \|\underline{x}\|^{\ell} | g_{\omega,\infty}^{(j)} (\underline{x})|\leq C_{\ell} \frac{2^{j}}{Z_{j-1,\omega}}\;,\qquad \forall \ell\geq 0\;.
\end{equation}
\begin{rem}
\begin{itemize}
\item[a)] Despite Remark \ref{rem:diffbeh}, it is not difficult to check that the above bounds also hold on scale $j=h$, without losing the factor $Z_{h-1,\omega}$. This is due to the fact that whenever $Z_{h-1,\omega}(\underline{k})$ is close to $Z_{\omega}$, the support function $f^{\omega}_{h}(\underline{k})$ is close to zero. See discussion around Eq. (63) of \cite{BMdensity} for a proof.
\item[b)] For $L$ large enough, and for $\e, \frak{a}$ small enough, the propagator $g^{(j)}_{\omega, \frak{a}, \varepsilon}(\underline{x})$ satisfies the same bounds (\ref{eq:singscale}), (\ref{eq:bdxj}), provided $\|\underline{x}\|$ is replaced by the distance on the torus of side $L$. 
\item[c)] Furthermore, we observe that the propagator on scale $j$ satisfies the parity property:
\begin{equation}
\hat g^{(j)}_{\omega, \e, \frak{a}}(\underline{k}) = - \hat g^{(j)}_{\omega, \e, \frak{a}}(-\underline{k})\;.
\end{equation}
This property will allow to reproduce the cancellation of the relevant terms on the next scale. 
\end{itemize}
\end{rem}
By the addition principle of Grassmann variables, we write:
\begin{eqnarray}\label{eq:intj-1}
\mathcal{Z}_{L, \frak{a}, \e, h, N}(0,0) &=& e^{\mathcal{W}^{(j)}} \frac{1}{n_{j}}\int \nu_{[h,j]}(d\psi)\, e^{\mathcal{L}_{4} \widetilde{V}^{(j)}(\sqrt{Z_{j-1}} \psi) + \mathcal{R} \widetilde{V}^{(j)}(\sqrt{Z_{j-1}}\psi)} \nonumber\\
&=& e^{\mathcal{W}^{(j)}} \frac{1}{n_{j}} \int \mu_{[h,j-1]}(d\psi) \int \nu_{0}(d\xi)\, e^{\mathcal{L}_{4} \widetilde{V}^{(j)}(\sqrt{Z_{j-1}} (\psi + \xi)) + \mathcal{R} \widetilde{V}^{(j)}(\sqrt{Z_{j-1}}(\psi + \xi))} \nonumber\\
&\equiv& e^{\mathcal{W}^{(j-1)}} \int \mu_{[h,j-1]}(d\psi)\, e^{V^{(j-1)}(\sqrt{Z_{j-1}}\psi)}\;,
\end{eqnarray}
where the prefactor is chosen so that $V^{(j-1)}(0) = 0$. The new effective interaction has the form:
\begin{equation}
V^{(j-1)}(\sqrt{Z_{j-1}}\psi) = \sum_{\Gamma} \frak{a}^{2|\Gamma|} \Big[ \prod_{f\in \Gamma} \sqrt{Z_{j-1, \omega(f)}} \Big] \sum_{\underline{X}} \psi_{\Gamma}(\underline{X}) W^{(j-1)}_{\Gamma}(\underline{X})\;,
\end{equation}
for suitable new kernels $W^{(j-1)}_{\Gamma}$. The iterative scale integration can be conveniently organized in terms of {\it Gallavotti-Nicol\`o trees}: we refer the reader to \cite{GM} for a review. The outcome of this representation is a convergent expansion for the kernels of the effective potentials, in powers of the running coupling constants $\lambda_{r, \omega, \omega'}$, with $r\geq j$. The new kernels also depend on the values of the wave function renormalization $Z_{r,\omega}(\underline{k})$ and of the effective Fermi velocity $v_{r, \omega}(\underline{k})$. It turns out that the expansion converges uniformly in all cutoff parameters provided:
\begin{equation}\label{eq:controlrcc}
\Big|\frac{Z_{r,\omega}}{Z_{r-1,\omega}}\Big| \leq e^{c|\lambda|}\;,\qquad | v_{r, \omega} - v_{\omega} | \leq C |\lambda|\;,\qquad |\lambda_{r,\omega, \omega'}| \leq C|\lambda|\;.
\end{equation}
If these bounds hold, the inductive assumptions can be propagated to the next scale. In particular, the free energy of the model is analytic at small coupling.

The proof of the bounds (\ref{eq:controlrcc}) is our main technical task. These properties are highly nontrivial, and rely on subtle cancellations. In fact, convergence of perturbation theory only implies that:
\begin{equation}\label{eq:betafunction}
|z_{j,\mu,\omega}| \leq C \lambda^{2}_{\geq j+1}\;,\qquad |\beta^{\lambda}_{j, \omega, \omega'}| \leq C\lambda_{\geq j+1}^{2}\;.
\end{equation}
These bounds cannot be used to prove the smallness of $\lambda_{\geq j}$ uniformly in $j$: to show this, we need to exhibit a cancellation in the expansion. Based on the bounds (\ref{eq:betafunction}), it is however possible to prove that the bounds (\ref{eq:controlrcc}) hold true for a finite number of scales, say for all $j$ such that $|j| |\lambda|^{\frac{1}{2} - \epsilon}\leq 1$. More precisely, for these scales one can prove that both the effective interaction and the Fermi velocity are close to the bare ones:
\begin{equation}\label{eq:lambdak}
| \lambda_{j,\omega,\omega'} - \lambda_{\omega,\omega'}  | \leq |\lambda|^{\frac{3}{2}}\;,\qquad | v_{j,\omega} - v_{\omega} | \leq |\lambda|^{\frac{1}{2}}\;.
\end{equation}
Notice that the validity of these bounds implies that:
\begin{equation}
Z_{j-1,\omega}  = Z_{\omega} \prod_{k\geq j} ( 1+ z_{0,k,\omega}) \sim Z_{\omega} 2^{\eta_{\omega} j}\;,\qquad \eta_{\omega} = O(\lambda^{2})\;. 
\end{equation}
In Section \ref{sec:flow}, we will establish the bounds (\ref{eq:lambdak}) for all $j\leq 0$. Furthermore, we will prove the vanishing of the beta function, for the effective coupling constants and for the Fermi velocity. Let us rewrite the flow of the Fermi velocity and of the effective coupling as:
\begin{equation}
v_{j-1,\omega} = v_{j,\omega} + \beta^{v}_{j,\omega}\;,\qquad \lambda_{j-1,\omega,\omega'} = \lambda_{j,\omega,\omega'} + \beta_{j,\omega,\omega'}^{\lambda}\;,
\end{equation}
where $\beta_{j,\omega}^{v}$, $\beta_{j,\omega,\omega'}^{\lambda}$ are analytic functions of the couplings $\lambda_{k} = \{ \lambda_{k,\omega,\omega'} \}$ for $k\geq j$,
\begin{equation}
\beta_{j,\omega}^{v} \equiv \beta_{j,\omega}^{v}(\lambda_{j}, \ldots, \lambda_{0})\;,\qquad \beta_{j,\omega,\omega'}^{\lambda} \equiv \beta_{j,\omega,\omega'}^{\lambda}(\lambda_{j}, \ldots, \lambda_{0})
\end{equation}
\begin{thm}[Vanishing of the beta function]\label{thm:vanishing}
There exists $0< \theta < 1$ such that the following is true. Let $s = \{s_{\omega,\omega'}\}$, with $s_{\omega,\omega'} \in \mathbb{C}$ and with $|s| = \max_{\omega,\omega'} |s_{\omega,\omega'}|$ small enough. Then, the following bound holds, for all $j\leq 0$:
\begin{equation}
|\beta_{j,\omega}^{v}(s, \ldots, s)| \leq C|s| 2^{\theta j}\;,\qquad  |\beta_{j,\omega,\omega'}^{\lambda}(s, \ldots, s)| \leq C|s|^{2} 2^{\theta j}\;.
\end{equation}
\end{thm}
\begin{rem}
Theorem \ref{thm:vanishing} has been proved in \cite{BMWI, BMchiral}, for the case of relativistic fermions with two chiralities. See also \cite{FM} for the simpler case of relativistic fermions with a single chirality. In Section \ref{sec:flow} we will review the proof, and we will discuss how to adapt it to the case of arbitrarily many chiralities.
\end{rem}

\subsection{Infrared regime: the correlation functions}\label{sec:corrfunc}

To conclude the discussion about the multiscale analysis of the generating functional of correlations, here we shall discuss how to adapt the analysis of the previous section to the multiscale evalutation of the generating functional of the correlation functions, $A\neq 0$, $\phi \neq 0$. We shall omit the $x_{2}, y_{2}$ dependence of the external fields: taking them into account simply amounts to introducing another label in the following discussion.

The generating functional of the correlations on scale $j \leq 0$ has the form:
\begin{equation}\label{eq:grasscor}
\mathcal{Z}_{L, \frak{a}, \e, h, N}(A,\phi) = e^{\mathcal{W}^{(j)}(A,\phi)} \int \mu_{[h,j]}(d\psi) e^{V^{(j)}(\sqrt{Z_{j}}\psi; A, \phi)}
\end{equation}
where the effective interaction has the form:
\begin{equation}
V^{(j)}(\sqrt{Z_{j}}\psi; A, \phi) = \sum_{\Gamma} \frak{a}^{2|\Gamma|} \sum_{\underline{X}, \underline{Y}, \underline{Z}} \Big[\prod_{f\in \Gamma_{\psi}} \sqrt{Z_{j, \omega(f)}}\Big] \psi_{\Gamma}(\underline{X}) A_{\Gamma}(\underline{Y}) \phi_{\Gamma}(\underline{Z})\, W^{(j)}_{\Gamma}(\underline{X}, \underline{Y}, \underline{Z})\;,
\end{equation}
where the kernels $W^{(j)}_{\Gamma}(\underline{X}, \underline{Y}, \underline{Z})$ are defined inductively. The Grassmann integral (\ref{eq:grasscor}) is evaluated in a multiscale fashion, as for the partition function. In order to control the flow of the new effective interactions, we need to extend the definition of localization operator, so to include terms that depend on the external fields, and behave dimensionally as relevant and marginal terms. To this end, we rewrite the effective interaction on scale $j$ in Fourier space as:
\begin{eqnarray}
&&V^{(j)}(\sqrt{Z_{j}}\psi; A, \phi) = \\
&& \sum_{n, m, s} \frac{1}{L^{4n + 2m + 2q}} \sum_{\substack{\underline{\omega}, \underline{\mu} \\ \underline{\rho}, \underline{\varepsilon}}}\sum_{\substack{\underline{K}, \underline{P} \\ \underline{Q}}} \widehat \psi_{\underline{\omega}}(\underline{K}) \widehat A_{\underline{\mu}}(\underline{P}) \widehat \Phi^{\underline{\varepsilon}}_{\underline{\rho}}(\underline{Q})\Big[ \prod_{i=1}^{2n} \sqrt{Z_{j,\omega_{i}}} \Big] \widehat W^{(j)}_{2n, m, s; \underline{\omega}, \underline{\mu}, \underline{\rho}, \underline{\varepsilon}}(\underline{K}, \underline{P}, \underline{Q})\delta(\underline{K}, \underline{P}, \underline{Q})\nonumber
\end{eqnarray}
where:
\begin{equation}
\widehat \psi_{\underline{\omega}}(\underline{K}) = \prod_{i=1}^{n} \hat \psi^{+}_{\underline{k}_{2i-1}, \omega_{2i-1}} \hat \psi^{-}_{\underline{k}_{2i}, \omega_{2i}}\;,\qquad \widehat A_{\underline{\mu}}(\underline{P}) = \prod_{i=1}^{m} \hat A_{\mu_{i},\underline{p}}\;,\qquad \widehat \Phi^{\underline{\varepsilon}}_{\underline{\rho}}(\underline{Q}) = \prod_{i=1}^{s} \phi^{\varepsilon_{i}}_{\underline{q},\rho_{i}}\;;
\end{equation}
the symbol $\delta(\underline{K}, \underline{P}, \underline{Q})$ denotes the Kronecker delta function; and:
\begin{eqnarray}
\widehat W^{(j)}_{2n, m, s; \underline{\omega}, \underline{\mu}, \underline{\rho}, \underline{\varepsilon}}(\underline{K}, \underline{P}, \underline{Q}) &=& \frac{1}{L^{2}} \widehat W^{(j)}_{2n, m, s; \underline{\omega}, \underline{\mu}, \underline{\rho}, \underline{\varepsilon}}(\underline{k}_{1}, \ldots, \underline{k}_{2n-1}, \underline{k}_{*}, \underline{p}_{1}, \ldots, \underline{p}_{m}, \underline{q}_{1}, \ldots, \underline{q}_{s}) \nonumber\\
&\equiv& \widehat W^{(j)}_{2n, m, s; \underline{\omega}, \underline{\mu}, \underline{\rho}, \underline{\varepsilon}}(\underline{k}_{1}, \ldots, \underline{k}_{2n-1}, \underline{p}_{1}, \ldots, \underline{p}_{m}, \underline{q}_{1}, \ldots, \underline{q}_{s})\;,
\end{eqnarray}
where $\underline{k}_{*}$ is chosen so to enforce momentum conservation. We complement the definition of the localization operator in (\ref{eq:local}), with:
\begin{equation}\label{eq:loccorr}
\mathcal{L} \widehat W^{(j)}_{2,1,0}(\underline{k}, \underline{p}) = \widehat W^{(j)}_{2,1,0;\infty}(\underline{0}, \underline{0})\;,\qquad 
\mathcal{L} \widehat W^{(j)}_{1,0,1}(\underline{k}) = \widehat W^{(j)}_{1,0,1;\infty}(\underline{k})\;,
\end{equation}
and we update accordingly the definition of $\mathcal{R}$. The first term in (\ref{eq:loccorr}) is marginal, while the second is relevant. We proceed as for the multiscale evaluation of the partition function, with the difference that in Eq. (\ref{eq:intj-1}) the argument of the first exponential is replaced by:
\begin{eqnarray}
&&\mathcal{L}_{1,0,1} \widetilde{V}^{(j)}(\sqrt{Z_{j-1}} \psi; A, \phi) + \mathcal{L}_{2,1,0} \widetilde{V}^{(j)}(\sqrt{Z_{j-1}} \psi; A, \phi) \nonumber\\ &&\qquad + \mathcal{L}_{4,0,0} \widetilde{V}^{(j)}(\sqrt{Z_{j-1}} \psi; A, \phi) + \mathcal{R} \widetilde{V}^{(j)}(\sqrt{Z_{j-1}}\psi; A, \phi)\nonumber
\end{eqnarray}
where:
\begin{eqnarray}
\mathcal{L}_{1,0,1} \widetilde{V}^{(j)}(\sqrt{Z_{j-1}} \psi; A, \phi) &:=& \sum_{\omega} \frac{1}{L^{2}} \sum_{\underline{k}} [\hat\psi_{\underline{k},\omega}^{+} Q^{(j)}_{\omega}(\underline{k}) \hat \phi^{-}_{\underline{k},\omega} + \hat\phi_{\underline{k},\omega}^{+} \overline{Q^{(j)}_{\omega}(\underline{k})} \hat \psi^{-}_{\underline{k},\omega}\big]\nonumber\\
\mathcal{L}_{2,1,0} \widetilde{V}^{(j)}(\sqrt{Z_{j-1}} \psi; A, \phi) &:=& \sum_{\omega, \mu} \frac{1}{L^{4}} \sum_{\underline{k}, \underline{p}} \hat A_{\mu, \underline{p}}\hat \psi^{+}_{\underline{k} + \underline{p}, \omega} \hat \psi^{-}_{\underline{k}, \omega} Z_{j-1,\omega} Z_{j, \mu, \omega}\nonumber\\
\mathcal{L}_{4,0,0} \widetilde{V}^{(j)}(\sqrt{Z_{j-1}} \psi; A, \phi) &:=& \mathcal{L}_{4} \widetilde{V}^{(j)}(\sqrt{Z_{j-1}} \psi)\;.
\end{eqnarray}
The new running coupling constants evolve as follows \cite{GM}, for $j\leq 0$:
\begin{eqnarray}
Q^{(j)}_{\omega}(\underline{k}) &=& Q^{(j+1)}_{\omega}(\underline{k}) + Z_{j,\omega}\widehat W^{(j)}_{2,0,0;\omega,\omega}(\underline{k}) \sum_{r=j+1}^{N} \hat g^{(r)}_{\omega;\infty}(\underline{k}) Q^{(r)}_{\omega}(\underline{k})\nonumber\\
Z_{j-1,\omega} Z_{\mu,j,\omega} &=& Z_{j,\omega} Z_{\mu,j+1,\omega} + Z_{j,\omega} \widehat W^{(j)}_{2,1,0;\omega,\omega,\mu;\infty}(\underline{0}, \underline{0})\;.
\end{eqnarray}
Notice that, for $\underline{k}$ in the support of $\hat g^{(j)}_{\omega;\infty}(\underline{k})$, the first equation becomes:
\begin{equation}
Q^{(j)}_{\omega}(\underline{k}) = Q_{\omega}(\underline{k}) [ 1 + Z_{j,\omega}\widehat W^{(j)}_{2,0,0;\omega,\omega}(\underline{k}) \widehat g^{(j+1)}_{\omega;\infty}(\underline{k})]\;.
\end{equation}
Instead, the second equation has to be solved for $Z_{\mu,0,\omega} = Z_{\mu,\omega} ( 1+ O(\lambda))$, where the $O(\lambda)$ corrections are determined by the integration of the ultraviolet degrees if freedom. The single-scale integration is performed via fermionic cluster expansion, as for the partition function, which can be conveniently organized in terms of Gallavotti-Nicol\`o trees. We refer to \cite{GM} for details. The next result will provide important information on various correlation functions of the reference model. We shall use the following notation:
\begin{equation}
\langle \cdot \rangle_{h,N} = \lim_{L\to \infty} \lim_{\varepsilon_{\to 0}}\lim_{\frak{a}\to 0} \frac{1}{L^{2}}\langle \cdot \rangle_{L, \frak{a}, \e, h, N}
\end{equation}
\begin{prop}[Bounds for the correlation functions.]\label{prp:correlations}\label{prp:propcorr} Suppose that the bounds (\ref{eq:controlrcc}) hold. Then, the following is true, uniformly in $N$.
\vspace{.2cm}

\noindent{\underline{\emph{Bound for the two-point function.}}} Let $\underline{k}$ be such that $\|\underline{k}\|_{\omega} = 2^{h}$. Then,
\begin{equation}\label{eq:2ptf}
\langle \psi^{-}_{\underline{k},\omega} \psi^{+}_{\underline{k},\omega} \rangle_{h,N} = \frac{1}{Z_{h-1,\omega} D_{h-1,\omega}(\underline{p})}(1 + o(1))\;,
\end{equation}
where $o(1)$ denotes error terms bounded by $C2^{\theta h}$ or by $C|\lambda|$, for a universal constant $C>0$.

\vspace{.2cm}

\noindent{\underline{\emph{Bound for the vertex function.}}} Let $\underline{k}$, $\underline{k}+\underline{p}$ be such that $\|\underline{k}\|_{\omega'} = \|\underline{k} + \underline{p}\|_{\omega'} = 2^{h}$. Then:
\begin{equation}\label{eq:bdvert}
\langle \hat n_{\underline{p},\omega}\;; \hat \psi^{-}_{\underline{k},\omega'}\;; \hat \psi^{+}_{\underline{k}+\underline{p},\omega'} \rangle_{h,N} = Z_{0,h-1,\omega'} Z_{h-1,\omega} \langle \hat \psi^{-}_{\underline{k},\omega} \hat \psi^{+}_{\underline{k},\omega} \rangle_{h,N} \langle \hat \psi^{-}_{\underline{k} + \underline{p},\omega} \hat\psi^{+}_{\underline{k} + \underline{p},\omega} \rangle_{h,N} (\delta_{\omega,\omega'} + o(1))\;.
\end{equation}
\vspace{.2cm}

\noindent{\underline{\emph{Bounds for the four-point function.}}} Let $\underline{k}_{1}$, $\underline{k}_{2}$, $\underline{k}_{3}$, $\underline{k}_{4}$ be such that $\| \underline{k}_{1} \|_{\omega} = \| \underline{k}_{2} \|_{\omega} = \| \underline{k}_{3} \|_{\omega'} = \| \underline{k}_{4} \|_{\omega'} = 2^{h}$. Then:
\begin{eqnarray}\label{eq:4ptf}
&&\langle \hat \psi^{+}_{\underline{k}_{1}, \omega}\;; \hat \psi^{-}_{\underline{k}_{2}, \omega}\;; \hat \psi^{+}_{\underline{k}_{3}, \omega'}\;; \hat \psi^{-}_{\underline{k}_{4}, \omega'} \rangle_{h,N} \\
&& = -\langle \psi^{-}_{\underline{k}_{1},\omega} \psi^{+}_{\underline{k}_{1},\omega} \rangle_{h,N} \langle \psi^{-}_{\underline{k}_{2},\omega} \psi^{+}_{\underline{k}_{2},\omega} \rangle_{h,N} \langle \psi^{-}_{\underline{k}_{3},\omega'} \psi^{+}_{\underline{k}_{3},\omega'} \rangle_{h,N} \langle \psi^{-}_{\underline{k}_{4},\omega'} \psi^{+}_{\underline{k}_{4},\omega'} \rangle_{h,N}\nonumber\\&&\quad\cdot \big( Z_{h-1,\omega} Z_{h-1,\omega'} (\lambda_{h,\omega,\omega'} + O(\lambda^{2}_{\geq h})) \big)\nonumber
\end{eqnarray}
Moreover, suppose that $\|\underline{k}_{1}\|_{\omega} = 2^{h_1}$, $\| \underline{k}_{2} \|_{\omega} = 2^{h_2}$. Then:
\begin{equation}\label{eq:4ptdiffscales}
| \langle \hat \psi^{+}_{\underline{k}_{1}, \omega}\;; \hat \psi^{-}_{\underline{k}_{2}, \omega}\;; \hat \psi^{+}_{\underline{k}_{3}, \omega'}\;; \hat \psi^{-}_{\underline{k}_{4}, \omega'} \rangle_{h,N} | \leq C \lambda_{\geq h} \frac{2^{-h_{1} - h_{2} - 2h}}{\sqrt{Z_{h_{1},\omega_{1}} Z_{h_{2},\omega_{2}}} Z_{h,\omega}}\;.
\end{equation}
\vspace{.2cm}

\noindent{\underline{\emph{Bound for the $(1,4)$-point function.}}} Let $\|\underline{p}\|_{\omega'} \leq C2^{h}$, and let $\|\underline{k}_{i}\| = 2^{h}$. Then:
\begin{equation}\label{eq:41bd}
|\langle \hat n_{\underline{p},\tilde \omega}\;; \hat \psi^{-}_{\underline{k}_{4} + \underline{p}, \omega'} \;; \hat \psi^{+}_{\underline{k}_{3},\omega'} \;; \hat \psi^{-}_{\underline{k}_{2}, \omega}\;; \hat \psi^{+}_{\underline{k}_{1}, \omega}\rangle_{h,N}|\leq C\lambda_{\geq h} \frac{2^{-5h}}{Z_{h-1,\omega} Z_{h-1,\omega'}} Z_{h-1,0,\tilde \omega}\;. 
\end{equation}
\end{prop}
\begin{rem}
\begin{itemize}
\item[a)] The proof of this proposition follows from the tree expansion for the correlations. In the setting of the usual Luttinger model, the proofs of the bounds (\ref{eq:2ptf}), (\ref{eq:bdvert}), (\ref{eq:4ptf}) can be found in \cite{BMchiral}, Theorem 1,  or in \cite{BMdensity}, Section 3.5. A proof of the bound (\ref{eq:4ptdiffscales}), (\ref{eq:41bd}) can be found in Appendix A1 of \cite{BMWI}. The adaptation to the present setting is straightforward, and will be omitted.
\item[b)] The $Z$-dependence of the bounds is easily understood recalling that the contraction with the fields $\phi$ produce propagators on the scale of the external momenta, which carry a factor $1/Z$. Also, the fields that are contracted with the monomials proportional to $\phi$ carry a factor $\sqrt{Z}$.
\end{itemize}
\end{rem}

\section{Vanishing of the beta function of the reference model}\label{sec:flow}

In this section we will prove the bounds (\ref{eq:lambdak}) and Theorem \ref{thm:vanishing}, adapting the strategy of \cite{BMWI, BMchiral}. We will proceed by induction: we will assume the validity of the bounds (\ref{eq:lambdak}) on scales $\geq h+1$, and we will show that they hold true for the running coupling constants on scale $h$ (recall that $h$ is the scale of the infrared cutoff). Recall that, as discussed at the end of Section \ref{sec:RGIR}, the estimates (\ref{eq:lambdak}) are true on scales $h$ such that $|h| |\lambda|^{\frac{1}{2} - \epsilon} \leq 1$.

All the statements below have to be understood for finite $L, h, N$ and nonzero $\frak{a}, \varepsilon$. To simplify the expressions, we shall temporarily use the notation:
\begin{equation}
\langle \cdot \rangle \equiv \langle \cdot \rangle_{h,N,L,\frak{a},\varepsilon}\;.
\end{equation}
Also, we shall set:
\begin{equation}
\mathcal{Z}(\phi) \equiv \mathcal{Z}_{h,N,L,\frak{a},\varepsilon}(0,\phi)\;,\qquad \mathcal{W}(\phi) \equiv \mathcal{W}_{h,N,L,\frak{a},\varepsilon}(0,\phi)\;,
\end{equation}
and we shall denote by $\langle \cdot \rangle(\phi)$ the Gibbs state of the reference model in the presence of a nonzero external field $\phi$. In what follows, it will be enough to study the generating functional of the correlations for the following choice of parameters:
\begin{eqnarray}
Q_{\omega}^{\pm} = 1\;,\qquad Z_{\mu,\omega} = \delta_{\mu,0}\;.
\end{eqnarray}

\subsection{Schwinger-Dyson equation}
The analysis relies on the combination of Schwinger-Dyson equations and Ward identities. Recall the definition of four-point function:
\begin{equation}
\langle \hat \psi^{+}_{\underline{k}_{1}, \omega_{1}}\;; \hat \psi^{-}_{\underline{k}_{2}, \omega_{2}}\;; \hat \psi^{+}_{\underline{k}_{3}, \omega_{3}}\;; \hat \psi^{-}_{\underline{k}_{4}, \omega_{4}}\rangle := L^{8}\frac{\partial^{4}}{\partial \hat \phi^{-}_{\underline{k}_{1},\omega_{1}} \partial \hat \phi^{+}_{\underline{k}_{2},\omega_{2}} \partial \hat \phi^{-}_{\underline{k}_{3},\omega_{3}} \partial \hat \phi^{+}_{\underline{k}_{4},\omega_{4}} } \mathcal{W}(\phi) \Big|_{\phi = 0}\;.
\end{equation}
By conservation of momentum, the four point function is zero unless $\underline{k}_{1} - \underline{k}_{2} + \underline{k}_{3} -\underline{k}_{4} = 0$ (modulo translations by $\underline{G}_{1}$ and $\underline{G}_{2}$). Also, by global chiral gauge symmetry, the four-point function is zero unless:
\begin{equation}
\omega_{1} = \omega_{2}\;,\quad \omega_{3} = \omega_{4} \qquad \text{or}\qquad \omega_{1} = \omega_{4}\;,\quad \omega_{2} = \omega_{3}\;.
\end{equation}
The study of the four-point function will be central in order to prove the vanishing of the beta function of the marginal couplings $\lambda_{h,\omega, \omega'}$. The starting point is the next proposition.
\begin{prop}[Schwinger-Dyson equation for the four-point function.]\label{prp:SD} Let:
\begin{equation}
\omega_{1} = \omega_{2} \equiv \omega\;,\quad \omega_{3} = \omega_{4} \equiv \omega'\;,\quad \omega \neq \omega'\;,\quad \underline{k}_{1} - \underline{k}_{2} = \underline{k}_{4} - \underline{k}_{3}\;.
\end{equation}
Then, the following identity holds true:
\begin{eqnarray}\label{eq:SD}
&&\langle \hat \psi^{+}_{\underline{k}_{1}, \omega}\;; \hat \psi^{-}_{\underline{k}_{2}, \omega}\;; \hat \psi^{+}_{\underline{k}_{3}, \omega'}\;; \hat \psi^{-}_{\underline{k}_{4}, \omega'} \rangle \\
&& = - \frac{\hat g_{\omega'}(\underline{k}_{4})}{L^{2}} \sum_{\tilde \omega:\, \tilde\omega \neq \omega'} \lambda_{\tilde\omega, \omega'} Z_{\tilde\omega} Z_{\omega'}  \sum_{\underline{p} \in \mathbb{D}_{L}^{\text{s}}} \hat v(\underline{p}) \langle  \hat n_{\underline{p},\tilde\omega}\;; \hat \psi^{+}_{\underline{k}_{1}, \omega}\;; \hat \psi^{-}_{\underline{k}_{2}, \omega}\;; \hat \psi^{+}_{\underline{k}_{3},\omega'}\;;  \hat \psi^{-}_{\underline{k}_{4} - \underline{p}, \omega'} \rangle\nonumber\\
&& \quad - \frac{\hat g_{\omega'}(\underline{k}_{4})}{L^{2}} \sum_{\tilde\omega:\, \tilde\omega \neq \omega'} \lambda_{\tilde\omega, \omega'} Z_{\tilde\omega} Z_{\omega'} \hat v(\underline{k}_{2} - \underline{k}_{1}) \langle \hat \psi^{-}_{\underline{k}_{3}, \omega'} \;; \hat \psi^{+}_{\underline{k}_{3},\omega'}\rangle \langle \hat n_{\underline{k}_{1} - \underline{k}_{2},\tilde\omega}\;; \hat \psi^{-}_{\underline{k}_{2}, \omega}\;; \hat \psi^{+}_{\underline{k}_{1}, \omega}\rangle\;.\nonumber
\end{eqnarray}
\end{prop}
\begin{proof} We start by computing:
\begin{equation}
\frac{\partial}{\partial \hat \phi^{+}_{\underline{k}_{4},\omega_{4}}} \log \mathcal{Z}(\phi) = \frac{1}{L^{2}}\frac{1}{\mathcal{Z}(\phi)} \int \mu(d\psi)\, e^{-V(\psi) + B(\psi; \phi)} \hat \psi^{-}_{\underline{k}_{4}, \omega_{4}}\;.
\end{equation}
Next, we use that:
\begin{equation}
\frac{1}{L^{2}}e^{-C(\psi)} \hat \psi^{-}_{\underline{k}_{4}, \omega_{4}} = - \hat g_{\omega_{4}}(\underline{k}_{4}) \Big( \frac{\partial}{\partial \psi^{+}_{\underline{k}_{4}, \omega_{4}}} e^{-C(\psi)}\Big)\;.
\end{equation}
Hence:
\begin{eqnarray}\label{eq:intparts}
\frac{\partial}{\partial \hat \phi^{+}_{\underline{k}_{4},\omega_{4}}} \log \mathcal{Z}(\phi) &=& - \frac{\hat g_{\omega_{4}}(\underline{k}_{4})}{\mathcal{Z}(\phi)} \int D\psi\, \Big( \frac{\partial}{\partial \psi^{+}_{\underline{k}_{4}, \omega_{4}}} e^{-C(\psi)}\Big) e^{-V(\psi) + B(\psi; \phi)} \nonumber\\
&=&  \frac{\hat g_{\omega_{4}}(\underline{k}_{4})}{\mathcal{Z}(\phi)} \int D\psi\, e^{-C(\psi)} \frac{\partial}{\partial \psi^{+}_{\underline{k}_{4}, \omega_{4}}} e^{-V(\psi) + B(\psi; \phi)}\\
&=&  \frac{\hat g_{\omega_{4}}(\underline{k}_{4})}{\mathcal{Z}(\phi)} \int \mu(d\psi)\, e^{-V(\psi) + B(\psi; \phi)} \Big( -\frac{\partial}{\partial \hat\psi^{+}_{\underline{k}_{4}, \omega_{4}}} V(\psi) + \frac{1}{L^{2}} \hat\phi^{-}_{\underline{k}_{4}, \omega_{4}} \Big)\;.\nonumber
\end{eqnarray}
The second step in Eq. (\ref{eq:intparts}) follows from Grassmann integration by parts, and the last step follows from the fact that $V(\psi), B(\psi; \phi)$ are given by sums of even monomials in the Grassmann variables (hence their exponentials commute with any Grassmann monomial). Using that:
\begin{eqnarray}
\frac{\partial}{\partial \hat \psi^{+}_{\underline{k}_{4}, \omega_{4}}} \hat n_{\underline{p}, \tilde\omega} &=& \frac{\partial}{\partial \hat \psi^{+}_{\underline{k}_{4}, \omega_{4}}} \Big( \frac{1}{L^{2}} \sum_{\underline{k} \in \mathbb{D}^{\text{a}}_{L,\frak{a}}} \hat \psi^{+}_{\underline{k} - \underline{p}, \tilde\omega} \hat \psi^{-}_{\underline{k}, \tilde\omega} \Big) \nonumber\\
&=& \frac{\delta_{\tilde\omega, \omega_{4}}}{L^{2}} \hat \psi^{-}_{\underline{k}_{4} + \underline{p}, \tilde\omega}\;,
\end{eqnarray}
we explicitly find, recalling that $\hat v(p) = \hat v(-p)$, $\lambda_{\omega, \omega'} = \lambda_{\omega', \omega}$:
\begin{equation}
\frac{\partial}{\partial \hat\psi^{+}_{\underline{k}_{4}, \omega_{4}}} V(\psi) = \sum_{\tilde\omega:\, \tilde\omega\neq \omega_{4}} \lambda_{\tilde\omega, \omega_{4}} Z_{\tilde\omega} Z_{\omega_{4}} \frac{1}{L^{4}} \sum_{\underline{p} \in \mathbb{D}_{L}^{\text{s}}} \hat v(\underline{p}) \hat n_{\underline{p},\tilde\omega} \hat \psi^{-}_{\underline{k}_{4} - \underline{p}, \omega_{4}}\;.
\end{equation}
From now on, we will enforce the constraints:
\begin{equation}\label{eq:chiralchoices}
\omega_{1} = \omega_{2}\;,\quad \omega_{3} = \omega_{4}\;,\quad \omega_{1} \neq \omega_{3}\;.
\end{equation}
%
%
%
We then get\footnote{The linear term in $\hat \phi^{-}_{\underline{k}_{4},\omega_{4}}$ gives a vanishing contribution to the four point function.}:
\begin{eqnarray}\label{eq:SD0}
&&\langle \hat \psi^{+}_{\underline{k}_{1}, \omega_{1}}\;; \hat \psi^{-}_{\underline{k}_{2}, \omega_{2}}\;; \hat \psi^{+}_{\underline{k}_{3}, \omega_{3}}\;; \hat \psi^{-}_{\underline{k}_{4}, \omega_{4}} \rangle \\
&& = -\hat g_{\omega_{4}}(\underline{k}_{4}) \sum_{\tilde\omega:\, \tilde\omega \neq \omega_{4}} \lambda_{\tilde\omega, \omega_{4}} Z_{\tilde\omega} Z_{\omega_{4}} \frac{1}{L^{2}} \sum_{\underline{p} \in \mathbb{D}_{L}^{\text{s}}} \hat v(\underline{p}) \frac{L^{6}\partial^{3}}{\partial \hat \phi^{-}_{\underline{k}_{1},\omega_{1}} \partial \hat \phi^{+}_{\underline{k}_{2},\omega_{2}} \partial \hat \phi^{-}_{\underline{k}_{3},\omega_{3}} } \langle \hat n_{\underline{p},\tilde\omega} \hat \psi^{-}_{\underline{k}_{4} - \underline{p}, \omega_{4}}\rangle \nonumber(\phi)\Big|_{\phi=0}\;.
\end{eqnarray}
Next, let us evaluate the remaining derivatives in terms of connected correlation functions. We have:
\begin{eqnarray}
\frac{L^{6}\partial^{3}}{\partial \hat \phi^{-}_{\underline{k}_{1},\omega_{1}} \partial \hat \phi^{+}_{\underline{k}_{2},\omega_{2}} \partial \hat \phi^{-}_{\underline{k}_{3},\omega_{3}} } \langle \hat n_{\underline{p},\tilde\omega} \hat \psi^{-}_{\underline{k}_{4} - \underline{p}, \omega_{4}}\rangle \nonumber(\phi)\Big|_{\phi=0} = \langle  \hat \psi^{+}_{\underline{k}_{1}, \omega_{1}}\;; \hat \psi^{-}_{\underline{k}_{2}, \omega_{2}}\;; \hat \psi^{+}_{\underline{k}_{3},\omega_{3}}\;; \hat n_{\underline{p},\tilde\omega} \hat \psi^{-}_{\underline{k}_{4} - \underline{p}, \omega_{4}} \rangle\;,
\end{eqnarray}
which we rewrite as, recalling the definition of cumulant, thanks to the choices (\ref{eq:chiralchoices}):
\begin{eqnarray}\label{eq:cumu}
&&\frac{L^{6}\partial^{3}}{\partial \hat \phi^{-}_{\underline{k}_{1},\omega_{1}} \partial \hat \phi^{+}_{\underline{k}_{2},\omega_{2}} \partial \hat \phi^{-}_{\underline{k}_{3},\omega_{3}} } \langle \hat n_{\underline{p},\tilde\omega} \hat \psi^{-}_{\underline{k}_{4} - \underline{p}, \omega_{4}}\rangle(\phi)\Big|_{\phi=0}\\
&&\qquad\qquad = \langle  \hat n_{\underline{p},\tilde\omega}\;; \hat \psi^{+}_{\underline{k}_{1}, \omega_{1}}\;; \hat \psi^{-}_{\underline{k}_{2}, \omega_{2}}\;; \hat \psi^{+}_{\underline{k}_{3},\omega_{3}}\;;  \hat \psi^{-}_{\underline{k}_{4} - \underline{p}, \omega_{4}} \rangle\nonumber\\
&& \qquad\qquad\quad + \langle \hat \psi^{-}_{\underline{k}_{4} - \underline{p}, \omega_{4}} \;; \hat \psi^{+}_{\underline{k}_{3},\omega_{3}}\rangle \langle \hat n_{\underline{p},\tilde\omega}\;; \hat \psi^{-}_{\underline{k}_{2}, \omega_{2}}\;; \hat \psi^{+}_{\underline{k}_{1}, \omega_{1}}\rangle\;.\nonumber
\end{eqnarray}
Plugging (\ref{eq:cumu}) in (\ref{eq:SD0}) we get:
\begin{eqnarray}
&&\langle \hat \psi^{+}_{\underline{k}_{1}, \omega_{1}}\;; \hat \psi^{-}_{\underline{k}_{2}, \omega_{2}}\;; \hat \psi^{+}_{\underline{k}_{3}, \omega_{3}}\;; \hat \psi^{-}_{\underline{k}_{4}, \omega_{4}} \rangle \\
&& = -\hat g_{\omega_{4}}(\underline{k}_{4}) \sum_{\tilde\omega:\, \tilde\omega \neq \omega_{4}} \lambda_{\tilde\omega, \omega_{4}} Z_{\tilde\omega} Z_{\omega_{4}} \frac{1}{L^{2}} \sum_{\underline{p} \in \mathbb{D}_{L}^{\text{s}}} \hat v(\underline{p})\langle  \hat n_{\underline{p},\tilde\omega}\;; \hat \psi^{+}_{\underline{k}_{1}, \omega_{1}}\;; \hat \psi^{-}_{\underline{k}_{2}, \omega_{2}}\;; \hat \psi^{+}_{\underline{k}_{3},\omega_{3}}\;;  \hat \psi^{-}_{\underline{k}_{4} - \underline{p}, \omega_{4}} \rangle\nonumber\nonumber\\
&& \quad - \hat g_{\omega_{4}}(\underline{k}_{4}) \sum_{\tilde\omega:\, \tilde\omega \neq \omega_{4}} \lambda_{\tilde\omega, \omega_{4}} Z_{\tilde\omega} Z_{\omega_{4}} \frac{1}{L^{2}} \sum_{\underline{p} \in \mathbb{D}_{L}^{\text{s}}} \hat v(\underline{p})\langle \hat \psi^{-}_{\underline{k}_{4} - \underline{p}, \omega_{4}} \;; \hat \psi^{+}_{\underline{k}_{3},\omega_{3}}\rangle \langle \hat n_{\underline{p},\tilde\omega}\;; \hat \psi^{-}_{\underline{k}_{2}, \omega_{2}}\;; \hat \psi^{+}_{\underline{k}_{1}, \omega_{1}}\rangle\;.\nonumber
\end{eqnarray}
This expression can be further simplified, noting that the last correlation function is zero unless $\underline{p} - \underline{k}_{1} + \underline{k}_{2} = 0$. Furthermore, by our choice of external momenta $\underline{k}_{3}$ and $\underline{k}_{4}$, we also have $\underline{p} = \underline{k}_{4} - \underline{k}_{3}$. In conclusion:
\begin{eqnarray}\label{eq:SD4pt}
&&\langle \hat \psi^{+}_{\underline{k}_{1}, \omega_{1}}\;; \hat \psi^{-}_{\underline{k}_{2}, \omega_{2}}\;; \hat \psi^{+}_{\underline{k}_{3}, \omega_{3}}\;; \hat \psi^{-}_{\underline{k}_{4}, \omega_{4}} \rangle \\
&& = -\hat g_{\omega_{4}}(\underline{k}_{4}) \sum_{\tilde\omega:\, \tilde\omega \neq \omega_{4}} \lambda_{\tilde\omega, \omega_{4}} Z_{\tilde\omega} Z_{\omega_{4}} \frac{1}{L^{2}} \sum_{\underline{p} \in \mathbb{D}_{L}^{\text{s}}} \hat v(\underline{p})\langle  \hat n_{\underline{p},\tilde\omega}\;; \hat \psi^{+}_{\underline{k}_{1}, \omega_{1}}\;; \hat \psi^{-}_{\underline{k}_{2}, \omega_{2}}\;; \hat \psi^{+}_{\underline{k}_{3},\omega_{3}}\;;  \hat \psi^{-}_{\underline{k}_{4} - \underline{p}, \omega_{4}} \rangle\nonumber\nonumber\\
&& \quad - \hat g_{\omega_{4}}(\underline{k}_{4}) \sum_{\tilde\omega:\, \tilde\omega \neq \omega_{4}} \lambda_{\tilde\omega, \omega_{4}} Z_{\tilde\omega} Z_{\omega_{4}} \frac{1}{L^{2}} \hat v(\underline{k}_{2} - \underline{k}_{1})\langle \hat \psi^{-}_{\underline{k}_{3}, \omega_{4}} \;; \hat \psi^{+}_{\underline{k}_{3},\omega_{3}}\rangle \langle \hat n_{\underline{k}_{1} - \underline{k}_{2},\tilde\omega}\;; \hat \psi^{-}_{\underline{k}_{2}, \omega_{2}}\;; \hat \psi^{+}_{\underline{k}_{1}, \omega_{1}}\rangle\;.\nonumber
\end{eqnarray}
This concludes the proof of Proposition \ref{prp:SD}.
\end{proof}
\subsection{Anomalous Ward identities}\label{sec:Ward}
The importance of the identity (\ref{eq:SD}) relies on the fact that, for external momenta on scale $2^{h}$, the left-hand side gives access to $\lambda_{h,\omega,\omega'}$ (recall Proposition \ref{prp:correlations}), while the right-hand side is manifestly $O(\lambda)$. Therefore, if we can show that, at leading order, the right-hand side of (\ref{eq:SD}) has the same momentum dependence of the four-point function, we will be able to bound $\lambda_{h,\omega,\omega'}$ by $C|\lambda|$ uniformly in $h$.

In order to gain insight on the momentum-dependence of the right-hand side, we shall derive Ward identities for the reference model. These identities follow from the behavior of the QFT under chiral local gauge transformations:
\begin{equation}\label{eq:chiralphase}
\psi^{\pm}_{\underline{x},\omega} \to e^{\pm i \alpha_{\omega}(\underline{x})} \psi^{\pm}_{\underline{x}, \omega}\;,
\end{equation}
with $\alpha_{\omega}(\underline{x})$ a function on $\frak{a} \mathbb{Z}^{2}$, with the periodicity of $\Lambda_{L,\frak{a}}$:
\begin{equation}
\alpha_{\omega}(\underline{x} + n_{1} e_{1} L + n_{2} e_{2} L) = \alpha_{\omega}(\underline{x})\;,
\end{equation}
which we can represent in terms of its Fourier series as:
\begin{equation}
\alpha_{\omega}(\underline{x}) = \frac{1}{L^{2}} \sum_{\underline{p} \in \mathbb{D}_{L}^{\text{p}}} e^{-i\underline{p}\cdot \underline{x}} \hat \alpha_{\omega}(\underline{p})\;.
\end{equation}
The starting point is the following proposition.
\begin{prop}[Local phase transformations.] Let $Q(\psi^{+}, \psi^{-})$ be a monomial in the Grassmann variables $\psi^{+}_{\underline{x},\omega}$, $\psi^{-}_{\underline{x},\omega}$. Let $Q_{\alpha}(\psi^{+}, \psi^{-})$ be the monomial obtained performing the following replacement in $Q(\psi^{+}, \psi^{-})$:
\begin{equation}\label{eq:phase}
\psi^{\pm}_{\underline{x},\omega} \to e^{\pm i \alpha_{\omega}(\underline{x})} \psi^{\pm}_{\underline{x}, \omega}\;.
\end{equation}
Then, the following identity holds true:
\begin{equation}\label{eq:Jac}
\int D\psi\, Q(\psi^{+}, \psi^{-}) = \int D\psi\, Q_{\alpha}(\psi^{+}, \psi^{-})\;.
\end{equation}
\end{prop}
\begin{rem} Since all functions on a finite Grassmann algebra are polynomials, Eq. (\ref{eq:Jac}) proves that the Jacobian of (\ref{eq:phase}) is one.
\end{rem}
\begin{proof} The proof is well-known, and we spell it out for completeness. To begin, notice that both the left-hand side and the right-hand side of (\ref{eq:Jac}) are zero if the same Grassmann field $\psi^{\pm}_{\underline{x},\omega}$ appears more than once in the monomial. Also, representing the fields $\psi^{\pm}_{\underline{x},\omega}$ in Fourier space as in (\ref{eq:confpsi}), and using that $|\Lambda_{L,\frak{a}}| = |\mathbb{D}^{\text{a}}_{L,\frak{a}}|$, it is clear that left-hand side and right-hand side of (\ref{eq:Jac}) are zero unless the number of $\psi^{+}_{\underline{x},\omega}$ fields and of $\psi^{-}_{\underline{x},\omega}$ fields are both equal to $n_{\text{e}}\times |\Lambda_{L,\frak{a}}|$.

Hence the fields $\psi^{+}_{\underline{x},\omega}$, $\psi^{-}_{\underline{x},\omega}$ come in pairs, and the identity (\ref{eq:Jac}) follows from the invariance of $\psi^{+}_{\underline{x},\omega} \psi^{-}_{\underline{x},\omega}$ under (\ref{eq:phase}).
\end{proof}
By linearity of the Grassmann integration, the property (\ref{eq:Jac}) implies the following identity, valid for any function $f$ on the finite Grassmann algebra:
\begin{equation}\label{eq:Jac2}
\int D\psi\, f(\psi) = \int D\psi\, f_{\alpha}(\psi)\;,
\end{equation}
with $f_{\alpha}(\psi)$ the function obtained from $f(\psi)$, after the chiral gauge transformation (\ref{eq:chiralphase}). The next proposition is a direct consequence of this identity.
\begin{prop}[Generating Ward identity.] The following identity holds true:
\begin{equation}\label{eq:masterWI}
0 = \int \mu(d\psi) e^{- V(\psi) + B(\psi; \phi) + \Gamma(\psi; A)} \Big[ Z_{\omega} \frak{D}_{\omega, \frak{a}}(\underline{p}) \hat n_{\underline{p},\omega} - \Delta_{\underline{p},\omega}(\psi) + B_{\underline{p},\omega}(\psi;\phi) \Big]\;,
\end{equation}
where:
\begin{eqnarray}\label{eq:Delta}
\Delta_{\underline{p},\omega}(\psi) &:=& \frac{1}{L^{2}} \sum_{\underline{k} \in \mathbb{D}^{\text{a}}_{L,\frak{a}}} \hat \psi^{+}_{\underline{k} - \underline{p},\omega}\Delta_{\omega}(\underline{k}, \underline{p})\hat \psi^{-}_{\underline{k},\omega} \\
\Delta_{\omega}(\underline{k}, \underline{p}) &:=& -\chi^{\omega, \varepsilon}_{[h,N]}(\underline{k} - \underline{p})^{-1} Z_{\omega} \frak{D}_{\omega,\frak{a}}(\underline{k} - \underline{p}) + \chi^{\omega, \varepsilon}_{[h,N]}(\underline{k})^{-1} Z_{\omega} \frak{D}_{\omega,\frak{a}}(\underline{k}) + Z_{\omega} \frak{D}_{\omega,\frak{a}}(\underline{p})\;,\nonumber
\end{eqnarray}
and:
\begin{equation}\label{eq:Bp}
B_{\underline{p},\omega}(\psi;\phi)  := \frac{1}{L^{2}} \sum_{\underline{k} \in \mathbb{D}^{\text{a}}_{L,\frak{a}}} \Big[ \hat \psi^{+}_{\underline{k} - \underline{p},\omega} \hat \phi^{-}_{\underline{k},\omega} - \hat \phi^{+}_{\underline{k},\omega} \hat \psi^{-}_{\underline{k} + \underline{p},\omega}\Big]\;.
\end{equation}
\end{prop}
%
%
%
\begin{proof}
We apply the symmetry (\ref{eq:Jac2}) to the generating functional of correlations. We get:
\begin{eqnarray}
\int \mu(d\psi)\, e^{-V(\psi) + B(\psi; \phi) + \Gamma(\psi; A)} &\equiv& \frac{1}{\mathcal{N}}\int D\psi\, e^{-C(\psi) - V(\psi) + B(\psi; \phi) + \Gamma(\psi; A)} \nonumber\\
&=& \frac{1}{\mathcal{N}}\int D\psi\, e^{-C_{\alpha}(\psi) - V_{\alpha}(\psi) + B_{\alpha}(\psi; \phi) + \Gamma_{\alpha}(\psi; A)} \\
&=& \int \mu(d\psi)e^{-(C_{\alpha}(\psi) - C(\psi)) - V(\psi) + B_{\alpha}(\psi; \phi) + \Gamma(\psi; A)}\;,\nonumber
\end{eqnarray}
where the second identity follows from (\ref{eq:Jac2}), and the last from the invariance of $V(\psi), \Gamma(\psi;A)$ under chiral gauge transformations. Let $\hat \alpha_{\omega}(\underline{p})$ be the Fourier coefficients of $\alpha_{\omega}(\underline{x})$, $\underline{p} \in \mathbb{D}_{L}^{\text{p}}$. Differentiating with respect to $\hat \alpha_{\omega}(\underline{p})$ we get:
\begin{equation}
0 = \frac{\partial}{\partial \hat \alpha_{\omega}(\underline{p})} \int \mu(d\psi)e^{-(C_{\alpha}(\psi) - C(\psi)) - V(\psi) + B_{\alpha}(\psi; \phi) + \Gamma(\psi;A)}\;.
\end{equation}
In particular:
\begin{equation}\label{eq:id}
0 = \int \mu(d\psi)e^{- V(\psi) + B(\psi; \phi) + \Gamma(\psi;A)}\Big[ -\frac{\partial}{\partial \hat \alpha_{\omega}(\underline{p})} C_{\alpha}(\psi) + \frac{\partial}{\partial \hat \alpha_{\omega}(\underline{p})} B_{\alpha}(\psi; \phi) \Big]_{\alpha = 0} \;.
\end{equation}
Let us now compute the derivatives. To begin, it is convenient to linearize the gauge transformation in momentum space:
\begin{equation}
\widehat{\Big(e^{i\alpha_{\omega}(\underline{x})} \psi^{+}_{\underline{x},\omega}\Big)}(\underline{k}) = \hat \psi^{+}_{\underline{k}, \omega} + \frac{i}{L^{2}} \sum_{\underline{p} \in \mathbb{D}_{L}^{\text{p}}} \hat \alpha_{\omega}(\underline{p}) \hat \psi^{+}_{\underline{k}+\underline{p},\omega} + O(\hat\alpha^{2})\;.
\end{equation}
Differentiating with respect to $\hat \alpha_{\omega}(\underline{p})$:
\begin{equation}\label{eq:paralpha1}
\frac{\partial}{\partial \hat \alpha_{\omega}(\underline{p})} \widehat{\Big(e^{i\alpha_{\omega}(\underline{x})} \psi^{+}_{\underline{x},\omega}\Big)}(\underline{k})\Big|_{\alpha = 0} = \frac{i}{L^{2}} \hat \psi^{+}_{\underline{k}+ \underline{p},\omega}\;.
\end{equation}
Similarly,
\begin{equation}\label{eq:paralpha2} 
\frac{\partial}{\partial \hat \alpha_{\omega}(\underline{p})} \widehat{\Big(e^{-i\alpha_{\omega}(\underline{x})} \psi^{-}_{\underline{x},\omega}\Big)}(\underline{k})\Big|_{\alpha = 0} = -\frac{i}{L^{2}} \hat \psi^{-}_{\underline{k}- \underline{p},\omega}\;.
\end{equation}
Let us now compute the derivatives in (\ref{eq:id}). Recalling (\ref{eq:sourceFourier}), and using (\ref{eq:paralpha1}), (\ref{eq:paralpha2}), we get:
\begin{eqnarray}
\frac{\partial}{\partial \hat \alpha_{\omega}(\underline{p})} B_{\alpha}(\psi; \phi) \Big|_{\alpha = 0} &=& \frac{i}{L^{4}} \sum_{\underline{k} \in \mathbb{D}^{\text{a}}_{L,\frak{a}}} \Big[ \hat \psi^{+}_{\underline{k} + \underline{p},\omega} \hat \phi^{-}_{\underline{k},\omega} - \hat \phi^{+}_{\underline{k},\omega} \hat \psi^{-}_{\underline{k} - \underline{p},\omega}\Big] \nonumber\\
&\equiv& \frac{i}{L^{2}}B_{-\underline{p},\omega}(\psi;\phi)\;, 
\end{eqnarray}
with $B_{\underline{p},\omega}(\psi;\phi)$ given by (\ref{eq:Bp}). Finally, consider now the first term in (\ref{eq:id}). Recalling the expression (\ref{eq:Cpsi}) for the covariance $C(\psi)$, and using (\ref{eq:paralpha1}), (\ref{eq:paralpha2}), we get:
\begin{eqnarray}
-\frac{\partial}{\partial \hat \alpha_{\omega}(\underline{p})} C_{\alpha}(\psi)\Big|_{\alpha = 0} &=& \frac{i}{L^{4}} \sum_{\underline{k} \in \mathbb{D}^{\text{a}}_{L,\frak{a}}} \Big[ \hat \psi^{+}_{\underline{k} + \underline{p},\omega}\Big( \hat g_{\omega}(\underline{k}+\underline{p})^{-1} - \hat g_{\omega}(\underline{k})^{-1} \Big) \hat \psi^{-}_{\underline{k},\omega} \Big]\nonumber\\
&\equiv& -\frac{i}{L^{2}}Z_{\omega} \frak{D}_{\omega,\frak{a}}(-\underline{p}) \hat n_{-\underline{p},\omega} - \frac{i}{L^{2}}\Delta_{-\underline{p},\omega}(\psi)\label{eq:parCalpha}
\end{eqnarray}
with $\Delta_{\underline{p},\omega}(\psi)$ given by (\ref{eq:Delta}). Changing $\underline{p}\to -\underline{p}$, and using $\frak{D}_{\omega, \frak{a}}(\underline{p}) = -\frak{D}_{\omega,\frak{a}}(-\underline{p})$, the claim follows.
\end{proof}
Differentiating the identity (\ref{eq:masterWI}) with respect to the external fields $A$ and to $\phi$ we obtain nontrivial relations between correlation functions. Suppose that:
\begin{equation}\label{eq:Zneq0}
\mathcal{Z}(0, 0) \neq 0\;;
\end{equation}
this is certainly true for $|\lambda|$ small enough, due to the finiteness of the Grassmann algebra. The bounds (\ref{eq:controlrcc}) also imply the validity of (\ref{eq:Zneq0}) for a range of $\lambda_{\omega,\omega'}$ which is uniform in the dimension of the algebra. This follows from the analyticity of the free energy, which rules out zeroes of the partition function. 

Differentiating with respect to the external fields the following identity:
\begin{equation}\label{eq:master}
0 = \frac{1}{\mathcal{Z}(A, \phi)}\int \mu(d\psi) e^{- V(\psi) + B(\psi; \phi) + \Gamma(\psi; A)} \Big[ Z_{\omega} \frak{D}_{\omega, \frak{a}}(\underline{p}) \hat n_{\underline{p},\omega} - \Delta_{\underline{p},\omega}(\psi) + B_{\underline{p},\omega}(\psi;\phi) \Big]
\end{equation}
we obtain relations between connected correlation functions. The prefactor $1/\mathcal{Z}(A, \phi)$ is a polynomial in $\phi^{\pm}$, with bounded coefficients thanks to (\ref{eq:Zneq0}), for $\|A\|_{\infty}$ small enough. The identities obtained in this way hold as long as (\ref{eq:Zneq0}) holds true. In the next proposition we collect the identities we will need in our analysis.
\begin{prop}[Ward identities for the correlation functions.]\label{prp:WI} Assume the bounds (\ref{eq:lambdak}) on scales $\geq h+1$. Then, the following identities hold true.
\vspace{.2cm}

\noindent{\em \underline{Density-density identity:}}
\begin{equation}\label{eq:density2anom}
Z_{\omega} \frak{D}_{\omega,\frak{a}} (\underline{p}) \langle \hat n_{\underline{p},\omega}\;; \hat n_{-\underline{p},\omega_{2}} \rangle = \langle \Delta_{\underline{p}, \omega}\;;\hat n_{-\underline{p},\omega_{2}} \rangle\;.
\end{equation}
{\em \underline{Vertex identity:}}
\begin{equation}\label{eq:vertexWI}
Z_{\omega} \frak{D}_{\omega,\frak{a}}(\underline{p}) \langle \hat n_{\underline{p},\omega}\;; \hat \psi^{-}_{\underline{k}, \omega_{2}} \;; \hat \psi^{+}_{\underline{k} + \underline{p},\omega_{2}} \rangle = \delta_{\omega,\omega_{2}}(\langle \hat \psi^{-}_{\underline{k},\omega} \hat \psi^{+}_{\underline{k},\omega} \rangle - \langle  \hat \psi^{-}_{\underline{k} + \underline{p},\omega} \hat \psi^{+}_{\underline{k} + \underline{p},\omega}\rangle) + \langle \Delta_{\underline{p}, \omega} \;; \hat \psi^{-}_{\underline{k}, \omega_{2}}\;; \hat \psi^{+}_{\underline{k} + \underline{p},\omega_{2}}\rangle\;.
\end{equation}
{\em \underline{$(1,4)$-point identity:}} 
\begin{eqnarray}\label{eq:WI6}
&&Z_{\omega} \frak{D}_{\omega,\frak{a}}(\underline{p}) \langle \hat n_{\underline{p},\omega}\;; \hat \psi^{+}_{\underline{k}, \omega_{2}}\;; \hat \psi^{-}_{\underline{q}, \omega_{2}}\;; \hat \psi^{+}_{\underline{k}', \omega_{3}} \;; \hat \psi^{-}_{\underline{q}_{*}, \omega_{3}} \rangle \nonumber\\
&&\qquad = \delta_{\omega,\omega_{3}}(\langle \hat \psi^{+}_{\underline{k}, \omega_{2}}\;; \hat \psi^{-}_{\underline{q}, \omega_{2}}\;; \hat \psi^{+}_{\underline{k}', \omega_{3}} \;; \hat \psi^{-}_{\underline{q}_{*} + \underline{p}, \omega_{3}} \rangle - \langle \hat \psi^{+}_{\underline{k}, \omega_{2}}\;; \hat \psi^{-}_{\underline{q}, \omega_{2}}\;; \hat \psi^{+}_{\underline{k}' - \underline{p}, \omega_{3}} \;; \hat \psi^{-}_{\underline{q}_{*}, \omega_{3}} \rangle) \nonumber\\
&&\qquad\quad + \delta_{\omega,\omega_{2}}(\langle \hat \psi^{+}_{\underline{k}, \omega_{2}}\;; \hat \psi^{-}_{\underline{q} + \underline{p}, \omega_{2}}\;; \hat \psi^{+}_{\underline{k}', \omega_{3}} \;; \hat \psi^{-}_{\underline{q}_{*}, \omega_{3}} \rangle - \langle \hat \psi^{+}_{\underline{k} - \underline{p}, \omega_{2}}\;; \hat \psi^{-}_{\underline{q}, \omega_{2}}\;; \hat \psi^{+}_{\underline{k}', \omega_{3}} \;; \hat \psi^{-}_{\underline{q}_{*}, \omega_{3}} \rangle) \nonumber\\
&&\qquad\quad+ \langle \Delta_{\underline{p},\omega}\;; \hat \psi^{+}_{\underline{k}, \omega_{2}}\;; \hat \psi^{-}_{\underline{q}, \omega_{2}}\;; \hat \psi^{+}_{\underline{k}', \omega_{3}} \;; \hat \psi^{-}_{\underline{q}_{*}, \omega_{3}} \rangle \;,
\end{eqnarray}
with the understanding that $\underline{q}_{*} = -\underline{p} + \underline{k} - \underline{q} + \underline{k}'$.
\end{prop}
\begin{proof}
The density-density identity is obtained differentiating Eq. (\ref{eq:master}) with respect to $\hat A_{\underline{p}, \omega_{2}}$, and setting $A = \phi = 0$. The vertex identity is obtained differentiating Eq. (\ref{eq:master}) with respect to $\hat \phi^{-}_{\underline{k} + \underline{p}, \omega_{2}}$ and then with respect to $\hat \phi^{+}_{\underline{k},\omega_{2}}$, and setting $A=\phi=0$. The $(1,4)$-point identity is obtained differentiating Eq. (\ref{eq:master}) with respect to (in this order, from left to right): $\hat \phi^{+}_{\underline{q}_{*},\omega_{3}}$, $\hat \phi^{-}_{\underline{k}', \omega_{3}}$, $\hat \phi^{+}_{\underline{q},\omega_{2}}$, $\hat \phi^{-}_{\underline{k},\omega_{2}}$, and then setting $A=\phi=0$.
\end{proof}
It is instructive to compute the right-hand side of (\ref{eq:density2anom}) in the absence of interactions, $\lambda = 0$. In the limit $L\to \infty$, $\varepsilon\to 0$, $\frak{a}\to 0$ we get:
\begin{eqnarray}\label{eq:buban}
\frac{1}{L^{2}}\langle \Delta_{\underline{p}, \omega}\;; \hat n_{-\underline{p},\omega_{2}} \rangle^{0} &=& -\delta_{\omega,\omega_{2}} \int \frac{d\underline{k}}{(2\pi)^{2}}\, \hat g_{\omega}(\underline{k} - \underline{p}) \hat g_{\omega}(\underline{k}) \Delta_{\omega}(\underline{k},\underline{p})\nonumber\\
&\equiv& \delta_{\omega,\omega_{2}} \frac{1}{Z_{\omega}}\frak{B}^{h,N}_{\omega}(\underline{p})\;.
\end{eqnarray}
We shall refer to $\frak{B}^{h,N}_{\omega}(\underline{p})$ as the {\it anomalous bubble diagram}. In order to evaluate this quantity, we proceed as follows. We write:
\begin{eqnarray}
&&\chi^{\omega}_{[h,N]}(\underline{k} - \underline{p}) \chi^{\omega}_{[h,N]}(\underline{k})\Big[ -(\chi^{\omega}_{[h,N]}(\underline{k} - \underline{p})^{-1} - 1) D_{\omega}(\underline{k} - \underline{p}) + (\chi^{\omega}_{[h,N]}(\underline{k})^{-1} - 1) D_{\omega}(\underline{k}) \Big] \nonumber\\
&& = -(1 - \chi^{\omega}_{[h,N]}(\underline{k} - \underline{p})) \chi^{\omega}_{[h,N]}(\underline{k}) D_{\omega}(\underline{k} - \underline{p}) + (1 - \chi^{\omega}_{[h,N]}(\underline{k}))\chi^{\omega}_{[h,N]}(\underline{k} - \underline{p}) D_{\omega}(\underline{k})\;.
\end{eqnarray}
We shall be interested in the value of this integral for $\underline{p}$ fixed, and as $h\to -\infty$, $N\to \infty$. Hence,
\begin{eqnarray}
\frak{B}^{h,N}_{\omega}(\underline{p}) &=&  \int \frac{d\underline{k}}{(2\pi)^{2}}\, \frac{1}{D_{\omega}(\underline{k})} \chi^{\omega}_{[h,N]}(\underline{k}) ( \chi^{\omega}_{[h,N]}(\underline{k} + \underline{p}) - \chi^{\omega}_{[h,N]}(\underline{k} - \underline{p}) )\;.
\end{eqnarray}
We see that the argument of the integral is nonzero only if $\|\underline{k} \pm \underline{p} \|\sim 2^{h}$, or if $\| \underline{k} \pm \underline{p} \| \sim 2^{N}$. In the first case, being $\underline{p}$ fixed, we either have $\|\underline{k} + \underline{p} \| \sim 2^{h}$ or $\|\underline{k} - \underline{p} \| \sim 2^{h}$. In both cases, the volume of the region of integration is estimated as $\|\underline{p}\| 2^{h}$, and the integrand is bounded by $\|\underline{p}\|^{-1}$. Hence, the corresponding contribution vanishes as $h\to -\infty$. 

In the latter case, we approximate:
\begin{equation}\label{eq:taylor}
\chi^{\omega}_{[h,N]}(\underline{k} + \underline{p}) - \chi^{\omega}_{[h,N]}(\underline{k} - \underline{p}) = 2\underline{p}\cdot \nabla_{\underline{k}}  \chi^{\omega}_{N}(\underline{k}) + O(2^{-2N})\;,
\end{equation}
where we used that the effect of the infrared cutoff is invisible on the chosen momentum range, and the fact that every derivative of the cutoff function is bounded by $2^{-N}$. The error term in (\ref{eq:taylor}) gives a vanishing contribution to the integral as $N\to \infty$. Concerning the main term, we rewrite it as, again neglecting the effect of the infrared cutoff:
\begin{equation}
\int \frac{d\underline{k}}{(2\pi)^{2}}\, \frac{1}{D_{\omega}(\underline{k})} \chi^{\omega}_{N}(\underline{k}) 2\underline{p}\cdot \nabla_{\underline{k}}  \chi^{\omega}_{N}(\underline{k}) = \int \frac{d\underline{k}}{(2\pi)^{2}}\, \frac{1}{D_{\omega}(\underline{k})} \chi^{\omega}(\underline{k}) 2\underline{p}\cdot \nabla_{\underline{k}}  \chi^{\omega}(\underline{k})\;,
\end{equation}
where we set $\chi^{\omega}(\underline{k})\equiv \chi_{1}^{\omega}(\underline{k})$ and we used that the integral is scale invariant. Next, we write:
\begin{eqnarray}
\underline{p}\cdot \nabla_{\underline{k}} \chi^{\omega}(\underline{k}) &=& (p_{0} \partial_{0} + p_{1} \partial_{1})\chi^{\omega}(\underline{k}) \nonumber\\
&=& \Big( \frac{p_{0} k_{0}}{\|\underline{k}\|_{\omega}} + \frac{p_{1} |v_{\omega}|^{2} k_{1}}{\|\underline{k}\|_{\omega}} \Big) \chi'(\|\underline{k}\|_{\omega})\;.
\end{eqnarray}
Hence,
\begin{eqnarray}
\int \frac{d\underline{k}}{(2\pi)^{2}}\, \frac{1}{D_{\omega}(\underline{k})} \chi^{\omega}_{N}(\underline{k}) 2\underline{p}\cdot \nabla_{\underline{k}}  \chi^{\omega}_{N}(\underline{k}) &=& \int \frac{d\underline{k}}{(2\pi)^{2}}\, \frac{2}{D_{\omega}(\underline{k})} \chi(\|\underline{k}\|_{\omega}) \Big( \frac{p_{0} k_{0}}{\|\underline{k}\|_{\omega}} + \frac{p_{1} |v_{\omega}|^{2} k_{1}}{\|\underline{k}\|_{\omega}} \Big) \chi'(\|\underline{k}\|_{\omega})\nonumber\\
&=& \frac{1}{|v_{\omega}|}\int \frac{d\underline{k}}{(2\pi)^{2}}\, \frac{2}{-ik_{0} + k_{1}} \chi(\|\underline{k}\|) \Big( \frac{p_{0} k_{0}}{\|\underline{k}\|} + \frac{p_{1} v_{\omega} k_{1}}{\|\underline{k}\|} \Big) \chi'(\|\underline{k}\|)\nonumber
\end{eqnarray}
where in the last step we performed the change of variables $v_{\omega} k_{1} \to k_{1}$, and we used that $|v_{\omega}|^{2} / v_{\omega} = v_{\omega}$. Finally, using that:
\begin{equation}
\int \frac{d\underline{k}}{(2\pi)^{2}}\, \frac{1}{-ik_{0} + k_{1}} \chi(\|\underline{k}\|) \frac{-ik_{0}}{\|\underline{k}\|} \chi'(\|\underline{k}\|) = \int \frac{d\underline{k}}{(2\pi)^{2}}\, \frac{1}{-ik_{0} + k_{1}} \chi(\|\underline{k}\|) \frac{k_{1}}{\|\underline{k}\|} \chi'(\|\underline{k}\|)
\end{equation}
we rewrite:
\begin{eqnarray}
\int \frac{d\underline{k}}{(2\pi)^{2}}\, \frac{1}{D_{\omega}(\underline{k})} \chi^{\omega}_{N}(\underline{k}) 2\underline{p}\cdot \nabla_{\underline{k}}  \chi^{\omega}_{N}(\underline{k}) &=& \frac{(i p_{0} + v_{\omega} p_{1})}{|v_{\omega}|} \int \frac{d\underline{k}}{(2\pi)^{2}}\, \frac{2}{-ik_{0} + k_{1}} \chi(\|\underline{k}\|) \frac{k_{1}}{\|\underline{k}\|} \chi'(\|\underline{k}\|) \nonumber\\
&=& \frac{(i p_{0} + v_{\omega} p_{1})}{|v_{\omega}|} \int \frac{d\underline{k}}{(2\pi)^{2}}\, \frac{1}{\|\underline{k}\|} \chi(\|\underline{k}\|)  \chi'(\|\underline{k}\|)\nonumber\\
&=& \frac{(i p_{0} + v_{\omega} p_{1})}{4\pi |v_{\omega}|}\;.
\end{eqnarray}
Therefore, setting $\tilde{\underline{p}} = (p_{0}, -p_{1})$, we obtain:
\begin{equation}\label{eq:bubble}
\frak{B}_{\omega}(\underline{p}) := \lim_{h\to -\infty}\lim_{N\to \infty} \frak{B}^{h,N}_{\omega}(\underline{p}) = -\frac{D_{\omega}(\tilde{\underline{p}})}{4\pi |v_{\omega}|}\;,\qquad \tilde{\underline{p}} = (p_{0}, -p_{1})\;.
\end{equation}
Together with (\ref{eq:density2anom}), (\ref{eq:buban}), this result shows that, for $\lambda = 0$:
\begin{equation}
\lim_{h\to -\infty}\lim_{N\to \infty}  \lim_{L\to \infty} \lim_{\varepsilon \to 0}\lim_{\frak{a}\to 0} \frac{1}{L^{2}}\langle \hat n_{\underline{p},\omega}\;; \hat n_{-\underline{p},\omega_{2}} \rangle^{0} =  -\delta_{\omega,\omega_{2}} \frac{1}{Z^{2}_{\omega}} \frac{1}{4\pi |v_{\omega}|} \frac{D_{\omega}(\tilde{\underline{p}})}{D_{\omega} (\underline{p}) }\;. 
\end{equation}
Later, we will extend this computation to the case of interacting fermions, $\lambda\neq 0$.

The next proposition will provide useful information about the structure of the correction terms in the Ward identities (\ref{prp:WI}), for the infrared and ultraviolet regularized theory. We will only state results for the corrections to the vertex Ward identity and to the $(1,4)$-point function Ward identity, Eqs. (\ref{eq:vertexWI}), (\ref{eq:WI6}). Thanks to these results, we will be able to close the Schwinger-Dyson equation (\ref{eq:SD}), and to control the flow of the marginal couplings. Recall the notation:
\begin{equation}
\langle \cdot \rangle_{h,N} = \lim_{L\to \infty} \lim_{\varepsilon\to 0} \lim_{\frak{a}\to 0} L^{-2}\langle \cdot \rangle_{h,N,L,\frak{a},\varepsilon}\;.
\end{equation}
The existence of the limit is guaranteed by our construction.
\begin{prop}[Analysis of the correction terms.]\label{prp:corrections} Suppose that the running coupling constants on scales $\geq h+1$ satisfy the estimates (\ref{eq:controlrcc}). Then, for $\lambda$ small enough, the following holds true:
\vspace{.2cm}

%
%
\noindent{\em \underline{Vertex function correction identity:}} 
\begin{eqnarray}\label{eq:vertexcorr}
&&\langle \Delta_{\underline{p}, \omega} \;;  \hat \psi^{-}_{\underline{k}, \omega_{2}}\;; \hat \psi^{+}_{\underline{k} + \underline{p},\omega_{2}} \rangle_{h,N}\\&&\qquad = -\frak{B}^{N}_{\omega}(\underline{p}) \sum_{\tilde\omega} \lambda_{\omega,\tilde\omega} Z_{\tilde\omega} v(\underline{p}) \langle \hat n_{\tilde\omega,\underline{p}}\;;  \hat \psi^{-}_{\underline{k}, \omega_{2}}\;; \hat \psi^{+}_{\underline{k} + \underline{p},\omega_{2}}\rangle_{h,N} + H^{h,N}_{1,2;\omega,\omega_{2}}(\underline{p}, \underline{k})\nonumber
\end{eqnarray}
\noindent{\em \underline{$(1,4)$-point correlation correction identity:}}
\begin{eqnarray}\label{eq:6ptWI}
&&\langle \Delta_{\underline{p},\omega}\;; \hat \psi^{+}_{\underline{k}_{1}, \omega_{2}}\;; \hat \psi^{-}_{\underline{k}_{2}, \omega_{2}}\;; \hat \psi^{+}_{\underline{k}_{3}, \omega_{3}} \;; \hat \psi^{-}_{\underline{k}_{4} - \underline{p}, \omega_{3}} \rangle_{h,N} \\&&\qquad = -\frak{B}^{N}_{\omega}(\underline{p}) \sum_{\tilde\omega} \lambda_{\omega,\tilde\omega} Z_{\tilde\omega} v(\underline{p}) \langle \hat n_{\tilde\omega,\underline{p}}\;; \hat \psi^{+}_{\underline{k}_{1}, \omega_{2}}\;; \hat \psi^{-}_{\underline{k}_{2}, \omega_{2}}\;; \hat \psi^{+}_{\underline{k}_{3}, \omega_{3}} \;; \hat \psi^{-}_{\underline{k}_{4} - \underline{p}, \omega_{3}} \rangle_{h,N}\nonumber\\&&\quad\qquad + H^{h,N}_{1,4;\omega,\omega_{2},\omega_{3}}(\underline{K}, \underline{p})\;,\nonumber
\end{eqnarray}
where $\frak{B}^{N}_{\omega}(\underline{p}) = \lim_{h\to -\infty} \frak{B}^{h,N}_{\omega}(\underline{p})$ with $\frak{B}^{h,N}_{\omega}(\underline{p})$ given by (\ref{eq:buban}), and the momenta are chosen so to satisfy momentum conservation. The error terms satisfy the following bounds, for $N$ large enough, and for external momenta of order $2^{h}$:
\begin{eqnarray}\label{eq:errest}
\Big| \frac{1}{D_{\omega}(\underline{p})} H^{h,N}_{1,2;\omega,\omega_{2}}(\underline{p}, \underline{k})\Big| &\leq& C \frac{\lambda_{\geq h}}{Z_{\omega_{2}, h-1}} 2^{- 2h}\nonumber\\
\Big| \frac{1}{L^{2}} \sum_{\underline{p} \in \mathbb{D}_{L}^{\text{p}}}  F_{\geq h}(\underline{p}) \frac{1}{D_{\omega}(\underline{p})} H^{h,N}_{1,4;\omega,\omega_{2},\omega_{3}}(\underline{K}, \underline{p})\Big| &\leq& C\frac{\lambda_{\geq h}}{Z_{\omega_{2}, h-1} Z_{\omega_{3}, h-1}} 2^{-3h}\;,
\end{eqnarray}
where $F_{\geq h}(\underline{p})$ is an arbitrary smooth function, supported for $\|\underline{p}\| \geq 2^{h}$, such that $\| \|\underline{p}\|^{n} F_{\geq h}(\underline{p}) \|_{1} \leq C_{n}$ for any $n\geq 0$.
\end{prop}
\begin{rem}
The proof of (\ref{eq:vertexcorr}) is analogous to the proof of Lemma 3 of \cite{MaQED}, while the proof of (\ref{eq:6ptWI}) is analogous to the proof of Lemma 1 of \cite{BMchiral}. To be precise, the reference \cite{MaQED} considers the case of removed infrared cutoff, $h=-\infty$; the effect of the infrared cutoff can be studied as in \cite{BMchiral}. The reference \cite{BMchiral} considers the case of interacting fermions with local interactions, which introduces extra difficulties with respect to our case. The strategy of \cite{BMchiral} has been revisited and extended in \cite{BFM3}, see Lemma 4.2 there. Finally, we point out that these references consider relativistic fermions with two chiralities, and opposite velocities. However, inspection of the proofs shows that the same arguments apply unchanged to the case of multiple chiralities, with arbitrary velocities. In particular, the $Z$-dependence of the bounds above is only due to the external fermionic lines.
\end{rem}

\subsection{The flow of the Fermi velocity and of the vertex renormalization}\label{sec:flowFermi}

As a first consequence of the Ward identities of Proposition \ref{prp:WI}, combined with the analysis of the correction terms of Proposition \ref{prp:corrections}, we will control the flow of the Fermi velocity $v_{h-1,\omega}$ and of the vertex renormalization $Z_{0,h-1,\omega}$.
\begin{rem}[Notations.]\label{rem:not} In the following, we shall use the following notations:
\begin{eqnarray}\label{eq:LZdef}
&&(\Lambda_{Z})_{\omega_{1},\omega_{2}} := \frac{1}{Z_{\omega_{1}}} \lambda_{\omega_{1},\omega_{2}} Z_{\omega_{2}}\;,\quad Z := \text{diag}(Z_{\omega})\;,\quad |v| := \text{diag}(|v_{\omega}|)\;,\nonumber\\&&\frak{B}^{N}(\underline{p}) := \text{diag}(\frak{B}^{N}_{\omega}(\underline{p}))\;,\quad D(\underline{p}) := \text{diag}(D_{\omega}(\underline{p}))\;.
\end{eqnarray}
Moreover, we shall use the notation $o(1)$ to denote error terms bounded by $C2^{\theta h}$ or by $C\lambda_{\geq h+1}$, for a universal constant $C>0$.
\end{rem}
Let $S_{1,2}(\underline{p},\underline{k})$, $\Delta_{1,2}(\underline{p},\underline{k})$ be the matrices with entries:
\begin{eqnarray}
(S_{1,2}(\underline{p},\underline{k}))_{\omega_{1},\omega_{2}} &=&  \langle \hat n_{\underline{p},\omega_{1}}\;; \hat \psi^{-}_{\underline{k}, \omega_{2}}\;; \hat \psi^{+}_{\underline{k} + \underline{p}, \omega_{2}}\rangle_{h,N} \nonumber\\
(\Delta_{1,2}(\underline{p},\underline{k}))_{\omega_{1},\omega_{2}} &=& \langle \Delta_{\underline{p}, \omega_{1}} \;; \hat \psi^{-}_{\underline{k}, \omega_{2}}\;; \hat \psi^{+}_{\underline{k} + \underline{p},\omega_{2}}  \rangle_{h,N}\;.
\end{eqnarray}
Then, we can rewrite the correction identity (\ref{eq:vertexcorr}) in a more compact form as, using the notations (\ref{eq:LZdef}):
\begin{equation}\label{eq:corrid12}
\Delta_{1,2}(\underline{p},\underline{k}) = -\frak{B}^{N}(\underline{p}) Z\Lambda_{Z} S_{1,2}(\underline{p},\underline{k}) + H_{1,2}(\underline{p},\underline{k})\;.
\end{equation}
Let $\delta S_{2}(\underline{p},\underline{k})$ be the matrix with entries:
\begin{equation}\label{eq:deltaS2def}
(\delta S_{2}(\underline{p},\underline{k}))_{\omega_{1},\omega_{2}} =  \frac{\delta_{\omega_{1},\omega_{2}}}{Z_{\omega_{1}} D_{\omega_{1}}(\underline{p})} (\langle  \hat\psi^{-}_{\underline{k},\omega_{1}} \hat\psi^{+}_{\underline{k},\omega_{1}} \rangle_{h,N} - \langle  \hat\psi^{-}_{\underline{k} + \underline{p},\omega_{1}} \hat\psi^{+}_{\underline{k} + \underline{p},\omega_{1}} \rangle_{h,N})\;;
\end{equation}
we rewrite the vertex Ward identity (\ref{eq:vertexWI}) as:
\begin{eqnarray}
S_{1,2}(\underline{p},\underline{k}) &=& \delta S_{2}(\underline{p},\underline{k}) + \frac{1}{Z D(\underline{p})} \Delta_{1,2}(\underline{p},\underline{k}) \nonumber\\
&=& \delta S_{2}(\underline{p},\underline{k}) - \frac{\frak{B}^{N}(\underline{p})}{D(\underline{p})} \Lambda_{Z} S_{1,2}(\underline{p},\underline{k}) + \frac{1}{Z D(\underline{p})} H_{1,2}(\underline{p},\underline{k})\;,
\end{eqnarray}
where in the last step we used the correction identity (\ref{eq:corrid12}). Therefore, solving for the vertex function:
\begin{equation}\label{eq:vertcomp}
S_{1,2}(\underline{p},\underline{k}) = \frac{1}{1 +  \frac{\frak{B}^{N}(\underline{p})}{D(\underline{p})} \Lambda_{Z}} \Big[ \delta S_{2}(\underline{p},\underline{k})  + \frac{1}{Z D(\underline{p})} H_{1,2}(\underline{p},\underline{k}) \Big]\;.
\end{equation}
Let $\|\underline{k}\|_{\omega} = 2^{h}$, $\|\underline{k} + \underline{p}\|_{\omega} = 2^{h}$, $\|\underline{p}\|_{\omega} = O(2^{h})$. In order to compute the asymptotic behavior of the vertex function, we use that, from Eq. (\ref{eq:2ptf}):
\begin{equation}\label{eq:vert1}
(\delta S_{2}(\underline{p},\underline{k}))_{\omega,\omega'} = \delta_{\omega,\omega'}\frac{Z_{h-1,\omega'} D_{h-1,\omega'}(\underline{p})}{Z_{\omega'} D_{\omega'}(\underline{p})} \langle \hat\psi^{-}_{\underline{k},\omega'}\hat\psi^{+}_{\underline{k},\omega'}\rangle_{h,N} \langle \hat\psi^{-}_{\underline{k} + \underline{p},\omega'}\hat\psi^{+}_{\underline{k} + \underline{p},\omega'}\rangle_{h,N}(1 + o(1))\;,
\end{equation}
where we used the $o(1)$ notation defined in Remark \ref{rem:not}. Also, we recall that, by the second of (\ref{eq:errest}):
\begin{equation}\label{eq:matrixelE}
\Big| \frac{1}{Z_{\omega} D_{\omega}(\underline{p})} H_{1,2;\omega,\omega'}(\underline{p},\underline{k})\Big| \leq C\lambda_{\geq h} \frac{2^{-2h}}{Z_{\omega} Z_{h-1,\omega'}}\;.
\end{equation}
The first information we can extract from (\ref{eq:vertcomp}) is the boundedness of the Fermi velocity. To see this, recall Eq. (\ref{eq:bdvert}), 
\begin{equation}\label{eq:vert2}
\big(S_{1,2}(\underline{p},\underline{k})\big)_{\omega',\omega'} = Z_{0,h-1,\omega'} Z_{h-1,\omega'} \langle \hat\psi^{-}_{\underline{k},\omega'}\hat\psi^{+}_{\underline{k},\omega'}\rangle_{h,N} \langle \hat\psi^{-}_{\underline{k} + \underline{p},\omega'}\hat\psi^{+}_{\underline{k} + \underline{p},\omega'}\rangle_{h,N} (1 + o(1))\;.
\end{equation}
Plugging (\ref{eq:vert1}), (\ref{eq:vert2}) in (\ref{eq:vertcomp}), and dividing by the product of the two-point functions with chirality $\omega'$:
\begin{eqnarray}\label{eq:uff}
Z_{0,h-1,\omega'} Z_{h-1,\omega'} (1 + o(1)) = \frac{Z_{h-1,\omega'} D_{h-1,\omega'}(\underline{p})}{Z_{\omega'} D_{\omega'}(\underline{p})} (1 + o(1))\;,
\end{eqnarray}
To get this identity, we crucially used that, from (\ref{eq:matrixelE}):
\begin{equation}
\frac{1}{\langle \hat\psi^{-}_{\underline{k},\omega'}\hat\psi^{+}_{\underline{k},\omega'}\rangle_{h,N} } \Big(\frac{1}{1 +  \frac{\frak{B}^{N}(\underline{p})}{D(\underline{p})}} \frac{1}{Z D(\underline{p})} H_{1,2}(\underline{p},\underline{k}) \Big)_{\omega',\omega'} \frac{1}{\langle \hat\psi^{-}_{\underline{k} + \underline{p},\omega'}\hat\psi^{+}_{\underline{k} + \underline{p},\omega'}\rangle_{h,N} } \leq C\lambda_{\geq h} Z_{h-1,\omega'}\;,
\end{equation}
and the fact that, for $|\lambda|$ small enough uniformly in $h$, as a consequence of the inductive assumptions (\ref{eq:controlrcc}) on scales $\geq h+1$:
\begin{equation}
\frac{1}{2} \leq \Big|\frac{D_{h-1,\omega'}(\underline{p})}{Z_{\omega'} D_{\omega'}(\underline{p})}\Big|\leq 2\;.
\end{equation}
In conclusion, we get, setting $p_{0} = 0$ or $p_{1} = 0$ in (\ref{eq:uff}):
\begin{equation}\label{eq:checkFermi}
|Z_{\omega} Z_{0, h-1,\omega} - 1| \leq C|\lambda|\;,\qquad |v_{h-1,\omega} - v_{\omega}|\leq C|\lambda|\;,
\end{equation}
for a universal constant $C>0$. This allows to control the flow of the Fermi velocity and of the vertex renormalization.

\subsection{The flow of the effective interaction}

In this section we will show how to control the flow of the effective couplings $\lambda_{h,\omega,\omega'}$. We will use the analysis of the correction terms of Proposition \ref{prp:corrections} to close the Schwinger-Dyson equation (\ref{eq:SD4pt}). Recall the choices (\ref{eq:chiralchoices})
\begin{equation}
\omega_{1} = \omega_{2}\equiv \omega\;,\quad \omega_{3} = \omega_{4} \equiv \omega'\;,\quad \omega \neq \omega'\;;
\end{equation}
The SD equation for the four-point function reads, in the $L\to \infty$, $\varepsilon\to 0$, $\frak{a}\to 0$ limit, (\ref{eq:SD}):
\begin{eqnarray}\label{eq:SD4ptlimti}
&&\langle \hat \psi^{+}_{\underline{k}_{1}, \omega}\;; \hat \psi^{-}_{\underline{k}_{2}, \omega}\;; \hat \psi^{+}_{\underline{k}_{3}, \omega'}\;; \hat \psi^{-}_{\underline{k}_{4}, \omega'} \rangle_{h,N} \\
&& = -\hat g_{\omega'}(\underline{k}_{4}) \sum_{\tilde \omega} \lambda_{\tilde \omega, \omega'} Z_{\tilde \omega} Z_{\omega'} \int \frac{d\underline{p}}{(2\pi)^{2}}\, \hat v(\underline{p}) \langle  \hat n_{\underline{p},\tilde\omega}\;; \hat \psi^{+}_{\underline{k}_{1}, \omega}\;; \hat \psi^{-}_{\underline{k}_{2}, \omega}\;; \hat \psi^{+}_{\underline{k}_{3},\omega'}\;;  \hat \psi^{-}_{\underline{k}_{4} - \underline{p}, \omega'} \rangle_{h,N}\nonumber\\
&& - \hat g_{\omega'}(\underline{k}_{4}) \sum_{\tilde \omega} \lambda_{\tilde \omega, \omega'} Z_{\tilde \omega} Z_{\omega'} \hat v(\underline{k}_{2} - \underline{k}_{1})\langle \hat \psi^{-}_{\underline{k}_{3}, \omega'} \;; \hat \psi^{+}_{\underline{k}_{3},\omega'}\rangle_{h,N} \langle \hat n_{\underline{k}_{1} - \underline{k}_{2},\tilde\omega}\;; \hat \psi^{-}_{\underline{k}_{2}, \omega}\;; \hat \psi^{+}_{\underline{k}_{1}, \omega}\rangle_{h,N}\nonumber\\
&&\equiv \text{I}_{\omega,\omega'} + \text{II}_{\omega,\omega'}\;.
\end{eqnarray}
Let us denote by $S_{4}$ the matrix with entries:
\begin{equation}
\big(S_{4}(\underline{k}_{1}, \underline{k}_{2}, \underline{k}_{3})\big)_{\omega, \omega'} = \langle \hat \psi^{+}_{\underline{k}_{1}, \omega}\;; \hat \psi^{-}_{\underline{k}_{2}, \omega}\;; \hat \psi^{+}_{\underline{k}_{3}, \omega'}\;; \hat \psi^{-}_{\underline{k}_{4}, \omega'} \rangle_{h,N}\;.
\end{equation}
It is convenient to ``amputate'' the external legs, by considering the matrix $\Gamma_{4}(\underline{k}_{1}, \underline{k}_{2}, \underline{k}_{3})$, with entries:
\begin{equation}\label{eq:G4A}
\big(\Gamma_{4}(\underline{k}_{1}, \underline{k}_{2}, \underline{k}_{3})\big)_{\omega,\omega'} := S_{2;\omega}(\underline{k}_{1})^{-1} S_{2;\omega}(\underline{k}_{2})^{-1} \big(S_{4}(\underline{k}_{1}, \underline{k}_{2}, \underline{k}_{3})\big)_{\omega,\omega'} S_{2;\omega'}(\underline{k}_{3})^{-1} S_{2;\omega'}(\underline{k}_{4})^{-1}\;. 
\end{equation}
where $S_{2;\omega}(\underline{k}) = \langle \psi^{-}_{\underline{k},\omega}\psi^{+}_{\underline{k},\omega} \rangle_{h,N}$. We shall discuss the two terms in (\ref{eq:SD4ptlimti}) separately.
\subsubsection{Analysis of the main term}
We denote by $\text{II}^{\text{A}}$ the term obtained from $\text{II}$ after amputating the external legs, as in (\ref{eq:G4A}). We have:
\begin{equation}\label{eq:IIA0}
\text{II}^{\text{A}}_{\omega,\omega'} = - \frac{\hat g_{\omega'}(\underline{k}_{4})}{ S_{2;\omega'}(\underline{k}_{4}) } \sum_{\tilde \omega} \lambda_{\tilde \omega, \omega'} Z_{\tilde \omega} Z_{\omega'} \hat v(\underline{k}_{2} - \underline{k}_{1}) S_{1,2;\tilde\omega, \omega}(\underline{k}_{1} - \underline{k}_{2}, \underline{k}_{2}) S_{2;\omega}(\underline{k}_{2})^{-1} S_{2;\omega}(\underline{k}_{1})^{-1}\;. 
\end{equation}
In matrix notation, using that $\lambda_{\tilde \omega, \omega'}  = \lambda_{\omega', \tilde \omega}$, recall the definition (\ref{eq:LZdef}) of $\Lambda_{Z}$:
\begin{eqnarray}\label{eq:LS21}
\sum_{\tilde \omega} \lambda_{\tilde \omega, \omega'} Z_{\tilde \omega} S_{1,2;\tilde\omega, \omega}(\underline{k}_{1} - \underline{k}_{2}, \underline{k}_{2}) &\equiv& \sum_{\tilde \omega} \lambda_{\omega', \tilde \omega} Z_{\tilde \omega} S_{1,2;\tilde\omega, \omega}(\underline{k}_{1} - \underline{k}_{2}, \underline{k}_{2}) \nonumber\\
&=& (Z \Lambda_{Z} S_{1,2}(\underline{k}_{1} - \underline{k}_{2}, \underline{k}_{2}))_{\omega',\omega}\;,
\end{eqnarray}
and using (\ref{eq:vertcomp}): 
\begin{eqnarray}\label{eq:Eterm}
&&Z \Lambda_{Z} S_{1,2}(\underline{k}_{1} - \underline{k}_{2}, \underline{k}_{2}) \\
&&\quad = Z \Lambda_{Z}\frac{1}{1 +  \frac{\frak{B}^{N}(\underline{k}_{1} - \underline{k}_{2})}{D(\underline{k}_{1} - \underline{k}_{2})} \Lambda_{Z}} \Big[ \delta S_{2}(\underline{k}_{1} - \underline{k}_{2}, \underline{k}_{2})  + \frac{1}{Z D(\underline{k}_{1} - \underline{k}_{2})} H_{1,2}(\underline{k}_{1} - \underline{k}_{2},\underline{k}_{2}) \Big]\;.\nonumber
\end{eqnarray}
The term $ \delta S_{2}(\underline{k}_{1} - \underline{k}_{2}, \underline{k}_{2})$ is a diagonal matrix, that can be rewritten as, see (\ref{eq:vert1}):
\begin{equation}
 \delta S_{2}(\underline{k}_{1} - \underline{k}_{2}, \underline{k}_{2}) = \frac{Z_{h-1} D_{h-1}(\underline{k}_{1} - \underline{k}_{2})}{Z D(\underline{k}_{1} - \underline{k}_{2})} S_{2}(\underline{k}_{1}) S_{2}(\underline{k}_{2}) (1 + o(1))\;.
\end{equation}
In terms of the matrix entries:
\begin{equation}
(\delta S_{2}(\underline{k}_{1} - \underline{k}_{2}, \underline{k}_{2}))_{\omega,\omega'} = \delta_{\omega,\omega'}\frac{Z_{h-1,\omega} D_{h-1,\omega}(\underline{k}_{1} - \underline{k}_{2})}{Z_{\omega} D_{\omega}(\underline{k}_{1} - \underline{k}_{2})} S_{2,\omega}(\underline{k}_{1}) S_{2,\omega}(\underline{k}_{2}) (1 + o(1))\;.
\end{equation}
Plugging this into (\ref{eq:Eterm}), we get:
\begin{eqnarray}
&&\big(Z \Lambda_{Z} S_{1,2}(\underline{k}_{1} - \underline{k}_{2}, \underline{k}_{2})\big)_{\omega',\omega}\nonumber\\
&& = \Big(Z \Lambda_{Z}\frac{1}{1 +  \frac{\frak{B}^{N}(\underline{k}_{1} - \underline{k}_{2})}{D(\underline{k}_{1} - \underline{k}_{2})} \Lambda_{Z}}\Big)_{\omega',\omega} \frac{Z_{h-1,\omega} D_{h-1,\omega}(\underline{k}_{1} - \underline{k}_{2})}{Z_{\omega} D_{\omega}(\underline{k}_{1} - \underline{k}_{2})} S_{2,\omega}(\underline{k}_{1}) S_{2,\omega}(\underline{k}_{2}) (1 + o(1)) \nonumber\\
&&\quad + \Big(Z \Lambda_{Z}\frac{1}{1 +  \frac{\frak{B}^{N}(\underline{k}_{1} - \underline{k}_{2})}{D(\underline{k}_{1} - \underline{k}_{2})} \Lambda_{Z}} \frac{1}{Z D(\underline{k}_{1} - \underline{k}_{2})} H_{1,2}(\underline{k}_{1} - \underline{k}_{2}, \underline{k}_{2})\Big)_{\omega',\omega}\;.
\end{eqnarray}
Using that:
\begin{equation}
 \Big(Z \Lambda_{Z}\frac{1}{1 +  \frac{\frak{B}^{N}(\underline{k}_{1} - \underline{k}_{2})}{D(\underline{k}_{1} - \underline{k}_{2})} \Lambda_{Z}}\Big)_{\omega',\omega} = \lambda_{\omega',\omega} Z_{\omega} + O(\lambda^{2})\;,
\end{equation}
we can write the contribution to $\text{II}^{\text{A}}$ obtained by only considering contributions proportional to $ \delta S_{2}(\underline{k}_{1} - \underline{k}_{2}, \underline{k}_{2})$ as:
\begin{eqnarray}
&&\text{II}_{\omega,\omega'}^{\text{A};1} = \nonumber\\
&&- Z_{\omega'} Z_{\omega}(\lambda_{\omega,\omega'} + O(\lambda^{2}))\hat v(\underline{k}_{2} - \underline{k}_{1}) \frac{Z_{h-1,\omega} D_{h-1,\omega}(\underline{k}_{1} - \underline{k}_{2})}{Z_{\omega} D_{\omega}(\underline{k}_{2} - \underline{k}_{1})} (1 + o(1)) \frac{\hat g_{\omega'}(\underline{k}_{4})}{ S_{2,\omega'}(\underline{k}_{4}) } \\
&&= -Z_{\omega'} Z_{\omega}(\lambda_{\omega,\omega'} + O(\lambda^{2})) \hat v(\underline{k}_{2} - \underline{k}_{1}) \frac{Z_{h-1,\omega} D_{h-1,\omega}(\underline{k}_{2} - \underline{k}_{2})}{Z_{\omega} D_{\omega}(\underline{k}_{2} - \underline{k}_{1})} \frac{Z_{h-1,\omega'} D_{h-1,\omega'}(\underline{k}_{4})}{Z_{\omega'} D_{\omega'}(\underline{k}_{4})} (1 + o(1))\;.\nonumber
\end{eqnarray}
Choosing $\underline{k}_{i} = (k_{0,i}, 0)$, and recalling that $\hat v(\underline{0}) = 1$, we get:
\begin{equation}\label{eq:IIA1}
\text{II}_{\omega,\omega'}^{\text{A};1} = -Z_{h-1,\omega} Z_{h-1,\omega'} (\lambda_{\omega,\omega'} + O(\lambda^{2})) (1 + o(1))\;.
\end{equation}
On the other hand, by Proposition \ref{prp:propcorr}, Eq. (\ref{eq:4ptf}):
\begin{equation}
\Gamma_{4}(\underline{k}_{1}, \underline{k}_{2}, \underline{k}_{3}) = -Z_{h-1,\omega} Z_{h-1,\omega'} (\lambda_{h,\omega,\omega'} + O(\lambda_{\geq h}^{2}))\;.
\end{equation}
Hence, if we could neglect all the other contributions to the Schwinger-Dyson equation, we would immediately get $|\lambda_{h,\omega,\omega'} - \lambda_{\omega,\omega'}| \leq C|\lambda|^{2}$, for a universal constant $C>0$, as wished. In the remaining part of the section, we shall discuss the control of the error terms.
\subsubsection{Analysis of the error terms}
To begin, consider the contribution due to the last term in (\ref{eq:Eterm}). From the second of (\ref{eq:errest}):
\begin{eqnarray}\label{eq:estIIB0}
\Big|\Big( Z \Lambda_{Z} \frac{1}{1 +  \frac{\frak{B}^{N}(\underline{k}_{1} - \underline{k}_{2})}{D(\underline{k}_{1} - \underline{k}_{2})} \Lambda_{Z}} \frac{1}{Z D(\underline{k}_{1} - \underline{k}_{2})} H_{1,2}(\underline{k}_{1} - \underline{k}_{2},\underline{k}_{2}) \Big)_{\omega',\omega}\Big|\leq C\lambda^{2}_{\geq h} \frac{2^{-2h}}{Z_{h-1,\omega}}\;.
\end{eqnarray}
Let us denote by $\text{II}^{\text{A};2}_{\omega,\omega'}$ the contribution to $\text{II}^{\text{A}}_{\omega,\omega'}$ associated to the term $H_{1,2}$. Using the estimate (\ref{eq:estIIB0}), we easily get, setting to zero the spatial component of the esternal momenta: 
\begin{equation}\label{eq:estIIA2}
|\text{II}^{\text{A};2}_{\omega,\omega'}| \leq C\lambda_{\geq h}^{2} Z_{h-1,\omega} Z_{h-1,\omega'}\;.
\end{equation}
Comparing with (\ref{eq:IIA1}), we see that this term gives a contribution that is subleading by an extra factor $\lambda_{\geq h}$.

We are left with considering the term $\text{I}$ in (\ref{eq:SD4ptlimti}). Our main task is to estimate the integral
\begin{equation}
\int \frac{d\underline{p}}{(2\pi)^{2}}\, \hat v(\underline{p}) \langle  \hat n_{\underline{p},\tilde\omega}\;; \hat \psi^{+}_{\underline{k}_{1}, \omega}\;; \hat \psi^{-}_{\underline{k}_{2}, \omega}\;; \hat \psi^{+}_{\underline{k}_{3},\omega'}\;;  \hat \psi^{-}_{\underline{k}_{4} - \underline{p}, \omega'} \rangle_{h,N}\;.
\end{equation}
We rewrite it as, with $\chi_{\leq h}(\underline{p})$ supported for momenta $\|\underline{p}\| \leq C2^{h}$, for $C>0$ large enough:
\begin{eqnarray}\label{eq:splitint}
&&\int \frac{d\underline{p}}{(2\pi)^{2}}\, \hat v(\underline{p})\chi_{\leq h}(\underline{p}) \langle  \hat n_{\underline{p},\tilde\omega}\;; \hat \psi^{+}_{\underline{k}_{1}, \omega}\;; \hat \psi^{-}_{\underline{k}_{2}, \omega}\;; \hat \psi^{+}_{\underline{k}_{3},\omega'}\;;  \hat \psi^{-}_{\underline{k}_{4} - \underline{p}, \omega'} \rangle_{h,N} \nonumber\\
&& + \int \frac{d\underline{p}}{(2\pi)^{2}}\, \hat v(\underline{p}) \chi_{\geq h}(\underline{p})\langle  \hat n_{\underline{p},\tilde\omega}\;; \hat \psi^{+}_{\underline{k}_{1}, \omega}\;; \hat \psi^{-}_{\underline{k}_{2}, \omega}\;; \hat \psi^{+}_{\underline{k}_{3},\omega'}\;;  \hat \psi^{-}_{\underline{k}_{4} - \underline{p}, \omega'} \rangle_{h,N}\;.
\end{eqnarray}
The first term can be estimated simply using the bound (\ref{eq:41bd}) of Proposition \ref{prp:propcorr}. We have, using that $|Z_{h-1,0,\omega}| \leq C$, recall (\ref{eq:checkFermi}):
\begin{equation}\label{eq:41leqh}
\Big| \int \frac{d\underline{p}}{(2\pi)^{2}}\, \hat v(\underline{p})\chi_{\leq h}(\underline{p}) \langle  \hat n_{\underline{p},\tilde\omega}\;; \hat \psi^{+}_{\underline{k}_{1}, \omega}\;; \hat \psi^{-}_{\underline{k}_{2}, \omega}\;; \hat \psi^{+}_{\underline{k}_{3},\omega'}\;;  \hat \psi^{-}_{\underline{k}_{4} - \underline{p}, \omega'} \rangle_{h,N}  \Big| \leq \frac{C\lambda_{\geq h} 2^{-3h}}{Z_{h-1,\omega'} Z_{h-1,\omega}}\;.
\end{equation}
If plugged in the Schwinger-Dyson equation (\ref{eq:SD4ptlimti}), this term gives a subleading contribution to the four-point function, by an extra factor $\lambda_{\geq h}$. We are left with estimating the second term in (\ref{eq:splitint}). To do this, we shall gain insight on the structure of the $(1,4)$-point function using the Ward identity (\ref{eq:WI6}), together with the correction identity (\ref{eq:6ptWI}). As $L\to \infty$, $\varepsilon \to 0$, $\frak{a}\to 0$, we get:
\begin{eqnarray}\label{eq:WI6ptnew}
&&\langle  \hat n_{\underline{p},\tilde\omega}\;; \hat \psi^{+}_{\underline{k}_{1}, \omega}\;; \hat \psi^{-}_{\underline{k}_{2}, \omega}\;; \hat \psi^{+}_{\underline{k}_{3},\omega'}\;;  \hat \psi^{-}_{\underline{k}_{4} - \underline{p}, \omega'} \rangle_{h,N} =\\
&& =\frac{\delta_{\tilde \omega,\omega'}}{Z_{\tilde \omega} D_{\tilde \omega}(\underline{p})} \big(\langle \hat \psi^{+}_{\underline{k}_{1}, \omega}\;; \hat \psi^{-}_{\underline{k}_{2}, \omega}\;; \hat \psi^{+}_{\underline{k}_{3},\omega'}\;;  \hat \psi^{-}_{\underline{k}_{4}, \omega'} \rangle_{h,N}
- \langle \hat \psi^{+}_{\underline{k}_{1}, \omega}\;; \hat \psi^{-}_{\underline{k}_{2}, \omega}\;; \hat \psi^{+}_{\underline{k}_{3}-\underline{p},\omega'}\;;  \hat \psi^{-}_{\underline{k}_{4} - \underline{p}, \omega'} \rangle_{h,N}\big) \nonumber\\
&&\quad + \frac{\delta_{\tilde \omega,\omega}}{ Z_{\tilde \omega} D_{\tilde \omega}(\underline{p}) } \big(\langle  \hat \psi^{+}_{\underline{k}_{1}, \omega}\;; \hat \psi^{-}_{\underline{k}_{2}+\underline{p}, \omega}\;; \hat \psi^{+}_{\underline{k}_{3},\omega'}\;;  \hat \psi^{-}_{\underline{k}_{4} - \underline{p}, \omega'} \rangle_{h,N} - \langle  \hat \psi^{+}_{\underline{k}_{1} - \underline{p}, \omega}\;; \hat \psi^{-}_{\underline{k}_{2}, \omega}\;; \hat \psi^{+}_{\underline{k}_{3},\omega'}\;;  \hat \psi^{-}_{\underline{k}_{4} - \underline{p}, \omega'} \rangle_{h,N}\big) \nonumber\\
&&\quad+ \frac{1}{Z_{\tilde \omega} D_{\tilde \omega}(\underline{p})}\langle \Delta_{\underline{p},\tilde \omega}\;; \hat \psi^{+}_{\underline{k}_{1}, \omega}\;; \hat \psi^{-}_{\underline{k}_{2}, \omega}\;; \hat \psi^{+}_{\underline{k}_{3},\omega'}\;;  \hat \psi^{-}_{\underline{k}_{4} - \underline{p}, \omega'} \rangle_{h,N}\;.\nonumber
\end{eqnarray}
Let $S_{1,4}(\underline{p}, \underline{K})$, with $\underline{K} = (\underline{k}_{1}, \underline{k}_{2}, \underline{k}_{3})$ be the tensor with entries:
\begin{equation}
(S_{1,4}(\underline{p}, \underline{K}))_{\tilde \omega, \omega,\omega'} = \langle  \hat n_{\underline{p},\tilde\omega}\;; \hat \psi^{+}_{\underline{k}_{1}, \omega}\;; \hat \psi^{-}_{\underline{k}_{2}, \omega}\;; \hat \psi^{+}_{\underline{k}_{3},\omega'}\;;  \hat \psi^{-}_{\underline{k}_{4} - \underline{p}, \omega'} \rangle_{h,N}\;,
\end{equation}
and let $\Delta_{1,4}(\underline{p},\underline{K})$ be the tensor with entries:
\begin{equation}
(\Delta_{1,4}(\underline{p},\underline{K}))_{\tilde \omega, \omega, \omega'} = \langle \Delta_{\underline{p},\tilde \omega}\;; \hat \psi^{+}_{\underline{k}_{1}, \omega}\;; \hat \psi^{-}_{\underline{k}_{2}, \omega}\;; \hat \psi^{+}_{\underline{k}_{3},\omega'}\;;  \hat \psi^{-}_{\underline{k}_{4} - \underline{p}, \omega'} \rangle_{h,N}\;.
\end{equation}
We rewrite the correction identity (\ref{eq:6ptWI}) as:
\begin{equation}
\Delta_{1,4}(\underline{p}, \underline{K}) = - \frak{B}^{N}(\underline{p}) Z \Lambda_{Z} \hat v(\underline{p}) S_{1,4}(\underline{p}, \underline{K}) + H_{1,4}(\underline{p}, \underline{K})\;.
\end{equation}
We are now in the position to solve (\ref{eq:WI6ptnew}) for the $(1,4)$-point function. We get:
\begin{eqnarray}\label{eq:6ptWI3}
&&S_{4,1; \tilde \omega, \omega, \omega'}(\underline{K}, \underline{p}) \\
&&= \frac{T_{\tilde \omega, \omega'}(\underline{p})}{Z_{\omega'} D_{\omega'}(\underline{p})} \big(\langle \hat \psi^{+}_{\underline{k}_{1}, \omega}\;; \hat \psi^{-}_{\underline{k}_{2}, \omega}\;; \hat \psi^{+}_{\underline{k}_{3},\omega'}\;;  \hat \psi^{-}_{\underline{k}_{4}, \omega'} \rangle_{h,N}
- \langle \hat \psi^{+}_{\underline{k}_{1}, \omega}\;; \hat \psi^{-}_{\underline{k}_{2}, \omega}\;; \hat \psi^{+}_{\underline{k}_{3}-\underline{p},\omega'}\;;  \hat \psi^{-}_{\underline{k}_{4} - \underline{p}, \omega'} \rangle_{h,N}\big)\nonumber\\
&& + \frac{T_{\tilde \omega,\omega}(\underline{p})}{Z_{\omega} D_{\omega}(\underline{p})}  \big(\langle  \hat \psi^{+}_{\underline{k}_{1}, \omega}\;; \hat \psi^{-}_{\underline{k}_{2}+\underline{p}, \omega}\;; \hat \psi^{+}_{\underline{k}_{3},\omega'}\;;  \hat \psi^{-}_{\underline{k}_{4} - \underline{p}, \omega'} \rangle_{h,N} - \langle  \hat \psi^{+}_{\underline{k}_{1} - \underline{p}, \omega}\;; \hat \psi^{-}_{\underline{k}_{2}, \omega}\;; \hat \psi^{+}_{\underline{k}_{3},\omega'}\;;  \hat \psi^{-}_{\underline{k}_{4} - \underline{p}, \omega'} \rangle_{h,N}\big)\nonumber\\
&& + \Big(T(\underline{p}) \frac{1}{Z D(\underline{p})} H_{1,4}(\underline{K},\underline{p})\Big)_{\tilde \omega,\omega,\omega'}\;,\nonumber
\end{eqnarray}
where:
\begin{equation}\label{eq:Tdef}
T_{\tilde \omega, \omega'}(\underline{p}) = \Big(\frac{1}{1 + \frak{B}^{N}(\underline{p}) D(\underline{p})^{-1} \Lambda_{Z} \hat v(\underline{p})}\Big)_{\tilde \omega, \omega'}\;.
\end{equation}
We now plug the identity (\ref{eq:6ptWI3}) in the second integral in (\ref{eq:splitint}), and we estimate the various terms. The contribution due to the first term in (\ref{eq:6ptWI3}) is exactly zero by parity, due to the fact that $T(\underline{p}) = T(-\underline{p})$, $D(\underline{p}) = -D(-\underline{p})$ and $\hat v(\underline{p}) = \hat v(-\underline{p})$. Consider the contribution to the Schwinger-Dyson equation due to the second; the third and fourth terms can be estimated in the same way. We have to estimate:
\begin{equation}
\int \frac{d\underline{p}}{(2\pi)^{2}}\, \hat v(\underline{p}) \chi_{\geq h}(\underline{p}) \frac{T_{\tilde \omega, \omega'}(\underline{p})}{Z_{\omega'} D_{\omega'}(\underline{p})}  \langle \hat \psi^{+}_{\underline{k}_{1}, \omega}\;; \hat \psi^{-}_{\underline{k}_{2}, \omega}\;; \hat \psi^{+}_{\underline{k}_{3}-\underline{p},\omega'}\;;  \hat \psi^{-}_{\underline{k}_{4}-\underline{p}, \omega'} \rangle_{h,N}\;.
\end{equation}
By the support properties of $\chi_{\geq h}$, $\| \underline{k_{3}} - \underline{p} \| \geq C2^{h}$, $\| \underline{k_{4}} - \underline{p} \| \geq C2^{h}$. By Proposition \ref{prp:correlations}, Eq. (\ref{eq:4ptdiffscales}), the four-point function at the argument of the integral is estimated as:
\begin{equation}
|\langle \hat \psi^{+}_{\underline{k}_{1}, \omega}\;; \hat \psi^{-}_{\underline{k}_{2}, \omega}\;; \hat \psi^{+}_{\underline{k}_{3}-\underline{p},\omega'}\;;  \hat \psi^{-}_{\underline{k}_{4}-\underline{p}, \omega'} \rangle_{h,N}|\leq C\lambda_{\geq h}\frac{2^{-2h_{\underline{p}}} 2^{-2h}}{Z_{h_{p},\omega} Z_{h,\omega}}\;,
\end{equation}
with $2^{h_{\underline{p}}} = \|\underline{p}\|$. Therefore, the corresponding contribution to the $\underline{p}$ integration is bounded as:
\begin{eqnarray}
&&\Big|\int_{\|\underline{p}\| \geq C2^{h}} \frac{d\underline{p}}{(2\pi)^{2}}\, \hat v(\underline{p}) \frac{T_{\tilde \omega, \omega'}(\underline{p})}{Z_{\omega'} D_{\omega'}(\underline{p})}   \langle \hat \psi^{+}_{\underline{k}_{1}, \omega}\;; \hat \psi^{-}_{\underline{k}_{2}, \omega}\;; \hat \psi^{+}_{\underline{k}_{3}-\underline{p},\omega'}\;;  \hat \psi^{-}_{\underline{k}_{4}-\underline{p}, \omega'} \rangle_{h,N}\Big| \nonumber\\ &&\qquad \leq  \int_{\|\underline{p}\| \geq C2^{h}} d\underline{p}\, |\hat v(\underline{p})|\frac{\lambda_{\geq h}}{\|\underline{p}\|^{3}} \frac{2^{-2h}}{Z_{h_{p},\omega} Z_{h,\omega'}}\nonumber\\ &&\qquad \leq C \lambda_{\geq h} \frac{2^{-3h}}{Z_{h,\omega} Z_{h,\omega'}}\;.
\end{eqnarray}
Hence, the contribution of this term to the first line of the right-hand side of (\ref{eq:SD4ptlimti}) is bounded as:
\begin{eqnarray}\label{eq:41corr1}
&&\Big| \hat g_{\omega'}(\underline{k}_{4}) \sum_{\tilde \omega} \lambda_{\tilde \omega, \omega'} Z_{\tilde \omega} Z_{\omega'} \int \frac{d\underline{p}}{(2\pi)^{2}}\, \hat v(\underline{p}) \chi_{\geq h}(\underline{p})  \frac{T_{\tilde \omega, \omega'}(\underline{p})}{Z_{\omega'} D_{\omega'}(\underline{p})}  \langle \hat \psi^{+}_{\underline{k}_{1}, \omega}\;; \hat \psi^{-}_{\underline{k}_{2}, \omega}\;; \hat \psi^{+}_{\underline{k}_{3}-\underline{p},\omega'}\;;  \hat \psi^{-}_{\underline{k}_{4}-\underline{p}, \omega'} \rangle_{h,N}\Big| \nonumber\\
&&\leq C\lambda_{\geq h}^{2} \frac{2^{-4h}}{Z_{h,\omega} Z_{h,\omega'}}\;,
\end{eqnarray}
which is subleading with respect to the left-hand side of (\ref{eq:SD4ptlimti}) by an extra factor $\lambda_{\geq h+1}$. The same holds for the contributions associated to the third and fourth terms in (\ref{eq:6ptWI3}). We are left with discussing the contribution to the Schwinger-Dyson equation due to the last term in the right-hand side of (\ref{eq:6ptWI3}). We have to control the integral:
\begin{equation}
 \int \frac{d\underline{p}}{(2\pi)^{2}}\, \hat v(\underline{p}) \chi_{\geq h}(\underline{p})  \frac{T_{\tilde \omega, \bar \omega}(\underline{p})}{Z_{\bar \omega} D_{\bar \omega}(\underline{p})} H_{1,4;\bar\omega, \omega, \omega'}(\underline{K},\underline{p})\;.
\end{equation}
The bound for this quantity is provided by the last estimate of (\ref{eq:errest}). We have:
\begin{equation}\label{eq:estfin41}
\Big|  \int \frac{d\underline{p}}{(2\pi)^{2}}\, \hat v(\underline{p}) \chi_{\geq h}(\underline{p})   \frac{T_{\tilde \omega, \bar \omega}(\underline{p})}{Z_{\bar \omega} D_{\bar \omega}(\underline{p})}H_{1,4;\bar\omega, \omega, \omega'}(\underline{K},\underline{p}) \Big| \leq C\frac{\lambda_{\geq h}}{Z_{\omega, h} Z_{\omega', h}} 2^{-3h}\;.
\end{equation}
Plugging this estimate in the Schwinger-Dyson equation (\ref{eq:SD4ptlimti}), we see that the corresponding contribution is subleading by a factor $\lambda_{\geq h}$ with respect to the four-point function.
\subsubsection{Conclusion: control of the RG flow of the reference model}
In a more compact notation, the Schwinger-Dyson equation is:
\begin{equation}
\Gamma_{4}(\underline{k}_{1}, \underline{k}_{2}, \underline{k}_{3}) = \text{I}^{\text{A}} + \text{II}^{\text{A}}\;, 
\end{equation}
where $\text{I}^{\text{A}}$ and $\text{II}^{\text{A}}$ are the terms $\text{I}$ and $\text{II}$ in Eq. (\ref{eq:SD4ptlimti}) after amputating the external legs, as in (\ref{eq:IIA0}). By Proposition \ref{prp:correlations}, Eq. (\ref{eq:4ptf}), we know that:
\begin{equation}
\Gamma_{4}(\underline{k}_{1}, \underline{k}_{2}, \underline{k}_{3}) = -Z_{h-1,\omega} Z_{h-1, \omega'} \big(\lambda_{h,\omega,\omega'}  + O(\lambda_{\geq h}^{2})\big)\;.
\end{equation}
By (\ref{eq:IIA1}), (\ref{eq:estIIA2}), we know that:
\begin{equation}
\text{II}^{\text{A}}_{\omega,\omega'} = \text{II}^{\text{A}; 1}_{\omega,\omega'} + \text{II}^{\text{A}; 2}_{\omega,\omega'}\;,
\end{equation}
with:
\begin{eqnarray}
 \text{II}^{\text{A}; 1}_{\omega,\omega'} &=& -Z_{h-1, \omega} Z_{h-1,\omega'} \lambda_{\omega,\omega'} (1 + o(1)) \nonumber\\
 |\text{II}^{\text{A}; 2}_{\omega,\omega'}| &\leq& C\lambda_{\geq h}^{2} Z_{h-1,\omega} Z_{h-1,\omega'}\;.
\end{eqnarray}
Finally, by (\ref{eq:41leqh}), (\ref{eq:41corr1}), (\ref{eq:estfin41}) we have that:
\begin{equation}
|\text{I}^{\text{A}}_{\omega,\omega'}| \leq C\lambda_{\geq h}^{2} Z_{h-1,\omega} Z_{h-1,\omega'}\;.
\end{equation}
All in all, the above estimates imply that:
\begin{equation}\label{eq:compare}
| \lambda_{h,\omega,\omega'} - \lambda_{\omega,\omega'} | \leq C\lambda_{\geq h} (2^{\theta h} + \lambda_{\geq h})\;,
\end{equation}
provided $\lambda_{\geq h}$ is small enough. As discussed at the end of Section \ref{sec:RGIR}, the inequality $|\lambda_{\geq h}| \leq K|\lambda|$ holds for $|\lambda|$ small enough and for all $h$ such that $|h| \lambda^{\frac{1}{2} - \epsilon} \leq 1$. The identity (\ref{eq:compare}) implies the validity of the bound $|\lambda_{\geq h}| \leq K\lambda$, with the same constant $K>0$, for the remaining scales. This allows to prove the validity of the bounds (\ref{eq:controlrcc}) for the running coupling constants on all scales.

We conclude by discussing the vanishing of the beta function. The bound (\ref{eq:compare}) immediately implies that:
\begin{equation}\label{eq:controllambda}
| \lambda_{h,\omega,\omega'} - \lambda_{\omega,\omega'} | \leq |\lambda|^{\frac{3}{2}}\;,
\end{equation}
That is, the function $h\mapsto \lambda_{h,\omega,\omega'} $ does not leave the complex ball of radius $|\lambda|^{\frac{3}{2}}$ centered at $\lambda_{\omega,\omega'}$. Recall that $\lambda_{h,\omega,\omega'} = \lambda_{h+1,\omega,\omega'} + \beta^{\lambda}_{h+1,\omega,\omega'}$, where the beta function $\beta^{\lambda}_{h+1,\omega,\omega'}$ is an analytic function of all effective couplings:
\begin{equation}
\beta^{\lambda}_{h+1,\omega,\omega'} \equiv \beta^{\lambda}_{h+1,\omega,\omega'}(\lambda_{h+1}, \lambda_{h+2},\ldots, \lambda_{0})\;,
\end{equation}
with $\lambda_{k} = \{ \lambda_{k,\omega,\omega'} \}$. As proven in Section 15 of \cite{BG}, see also Theorem 3.1 of \cite{BMWI}, the inequality (\ref{eq:controllambda}) implies that:
\begin{equation}\label{eq:vanishlambda}
|\beta^{\lambda}_{h+1,\omega,\omega'}(s, \ldots, s)| \leq C|s|^{2} 2^{\theta h}\;,
\end{equation}
for $s = \{s_{\omega,\omega'}\}$ and $|s|$ small enough, and for some $0<\theta < 1$. Also, the analysis of Section \ref{sec:flowFermi} implies the vanishing of the beta function for the Fermi velocity as well. We proved that, uniformly in $h$:
\begin{equation}
| v_{h,\omega} - v_{\omega} | \leq C|\lambda|\;.
\end{equation} 
Writing $v_{h,\omega} = v_{h+1,\omega} + \beta_{h+1,\omega}^{v}$, the same argument used to prove (\ref{eq:vanishlambda}) can be used to show that:
\begin{equation}
|\beta^{v}_{h+1,\omega}(s, \ldots, s)| \leq C|s| 2^{\theta h}\;.
\end{equation}
This concludes the proof of Theorem \ref{thm:vanishing}.\qed

\subsection{Correlation functions after cutoff removal}

To conclude the discussion about the reference model, we show how to compute the density-density and vertex function of the reference model, after removing the infrared and ultraviolet cutoff. The analysis of this subsection will be important to compute the edge correlation functions of the lattice model. 
\begin{rem}[Notations.]
In what follows, we shall set:
\begin{equation}
\langle \cdot \rangle = \lim_{h\to -\infty} \lim_{N\to \infty} \langle \cdot \rangle_{h,N}\;.
\end{equation} 
Also, we define $\frak{B}(\underline{p}) = \lim_{N\to \infty} \frak{B}^{N}(\underline{p})$. From (\ref{eq:bubble}):
\begin{equation}
\frak{B}_{\omega}(\underline{p}) = -\frac{D_{\omega}(\tilde{\underline{p}})}{4\pi |v_{\omega}|}\;,\qquad \tilde{\underline{p}} = (p_{0}, -p_{1})\;.
\end{equation}
\end{rem}
The next proposition is the counterpart of Proposition \ref{prp:corrections}, after the removal of the cutoffs.
\begin{prop}[Correlations after cutoffs removal.]\label{prp:corrsref} For $|\lambda|$ small enough, the following holds true.
\vspace{.2cm}

\noindent{\underline{\emph{Density-density correlation correction identity.}}} We have:
\begin{equation}\label{eq:ddcorrinfty}
\langle \Delta_{\underline{p}, \omega}\;; \hat n_{-\underline{p},\omega_{2}} \rangle = \delta_{\omega,\omega_{2}} \frac{1}{Z_{\omega}}\frak{B}_{\omega}(\underline{p}) - \frak{B}_{\omega}(\underline{p}) \sum_{\tilde \omega} \lambda_{\omega,\tilde \omega} Z_{\tilde\omega} v(\underline{p}) \langle \hat n_{\tilde\omega,\underline{p}}\;; \hat n_{-\underline{p},\omega_{2}} \rangle\;.
\end{equation}
\vspace{.2cm}

\noindent{\underline{\emph{Vertex function correction identity.}}} We have:
\begin{equation}\label{eq:vertexcorrinfty}
\langle \Delta_{\underline{p}, \omega} \;;  \hat \psi^{-}_{\underline{k}, \omega_{2}}\;; \hat \psi^{+}_{\underline{k} + \underline{p},\omega_{2}} \rangle = -\frak{B}_{\omega}(\underline{p}) \sum_{\tilde\omega} \lambda_{\omega,\tilde\omega} Z_{\tilde\omega} v(\underline{p}) \langle \hat n_{\tilde\omega,\underline{p}}\;;  \hat \psi^{-}_{\underline{k}, \omega_{2}}\;; \hat \psi^{+}_{\underline{k} + \underline{p},\omega_{2}}\rangle
\end{equation}
\end{prop}
That is, the $H$-error terms in Proposition \ref{prp:corrections} vanish in the limit $h\to -\infty$ and $N\to \infty$, for fixed external momenta. We refer the reader to {\it e.g.} Lemma 3 of \cite{MaQED}, or Theorem 3.2 of \cite{BFM3} for a proof of these results in the case of fermions with two chiralities. The arguments immediately apply to a general number of chiralities. Proposition \ref{prp:corrsref}, combined with the anomalous Ward identities of Proposition \ref{prp:WI}, allow to {\it compute} the density-density and vertex correlation functions, in terms of the two-point function.
\begin{prop}[Density-density correlation and vertex function.] The following identities hold true:
\begin{eqnarray}\label{eq:ddexp}
\langle \hat n_{\underline{p},\omega}\;; \hat n_{-\underline{p},\omega'} \rangle &=& T_{\omega,\omega'}(\underline{p}) \frac{1}{Z_{\omega'}^{2}}\frac{\frak{B}_{\omega'}(\underline{p})}{D_{\omega'}(\underline{p})}\;,\nonumber\\
\langle \hat n_{\underline{p},\omega}\;; \hat \psi^{-}_{\underline{k}, \omega'}\;; \hat \psi^{+}_{\underline{k} + \underline{p},\omega'} \rangle &=& T_{\omega,\omega'}(\underline{p}) \frac{1}{Z_{\omega'} D_{\omega'}(\underline{p})} \big(\langle \hat\psi^{-}_{\underline{k},\omega'} \hat\psi^{+}_{\underline{k},\omega'}\rangle - \langle \hat\psi^{-}_{\underline{k} +\underline{p},\omega'} \hat\psi^{+}_{\underline{k} + \underline{p},\omega'}  \rangle \big)\;,
\end{eqnarray}
where $T(\underline{p})$ is given by the $N\to \infty$ limit of (\ref{eq:Tdef}).
\end{prop}
\begin{proof} Plugging the identities (\ref{eq:ddcorrinfty}), (\ref{eq:vertexcorrinfty}) into the Ward identities (\ref{eq:density2anom}), (\ref{eq:vertexWI}) we get closed equations for the density-density correlation, and for the vertex function in terms of the two-point function. The equations can be easily solved, and the claim follows.
\end{proof}

\section{Connection with the lattice model}\label{sec:compare}

\begin{rem} From this section until the end of the manuscript, we shall decorate by the supscript `ref' all the quantities related to the reference model.
\end{rem}

In this section we will discuss the connection between the construction of the reference model, and the control of the RG flow for the original lattice model. In particular, we will sketch the proof of Proposition \ref{prp:flow}, about the flow of the running coupling constants of the original lattice model. We conclude the section by stating the connection between the correlation functions of the lattice model, and those of the reference model. 

As discussed after Eq. (\ref{eq:betaflow}), the last bound follows from a suitable choice of the counterterms $\mu_{\omega}$. Concerning the first two estimates in (\ref{eq:betaflow}), we proceed as in the proof of Propositon 9.3 of \cite{AMP}. We use the representation of the beta function of the lattice model in terms of Gallavotti-Nicol\`o trees, Section 9.3.1 of \cite{AMP}. Every tree contributing to the beta function of the lattice model can be written the corresponding tree for the reference model, up to an error term that satisfies the same bounds of the tree times an extra dimensional gain $2^{\theta k}$. This relies on the short-memory property and of the continuity property of GN trees, Remarks 9.1 and 9.2 of \cite{AMP}. The final estimates follow from the corresponding estimates for the beta function of the reference model, Theorem \ref{thm:vanishing}.

Recall the notations of Section \ref{sec:schw} for the lattice Schwinger functions. Concerning the Schwinger functions of the reference model, we shall set:
\begin{eqnarray}
S^{\text{ref}}_{2;\omega}(\underline{k}) &:=& \langle \hat \psi^{-}_{\underline{k},\omega} \hat \psi^{+}_{\underline{k},\omega}\rangle \nonumber\\
S^{\text{ref}}_{1,2;\omega,\omega'}(\underline{p}, \underline{k}) &:=& \langle \hat n_{\underline{p},\omega}\;; \hat \psi^{-}_{\underline{k}, \omega'}\;; \hat \psi^{+}_{\underline{k} + \underline{p},\omega'} \rangle\nonumber\\
S^{\text{ref}}_{0,2;\omega,\omega'}(\underline{p}) &:=& \langle \hat n_{\underline{p},\omega}\;; \hat n_{-\underline{p},\omega'} \rangle\;.
\end{eqnarray}
The next proposition allows to compare the correlations of lattice and reference model.
\begin{prop}[Comparison of correlations.]\label{prp:rel} For $\lambda$ small enough, and for $\beta, L$ large enough, the following is true. There exists a choice of the bare parameters $Z^{\text{ref}}_{\omega}, v^{\text{ref}}_{\omega}, \lambda_{\omega,\omega'}^{\text{ref}}$ of the reference model, satisfying
\begin{eqnarray}
&&|Z^{\text{ref}}_{\omega} - 1|\leq C|\lambda|\;,\qquad |v^{\text{ref}}_{\omega} - v_{\omega}|\leq C|\lambda|\;,\qquad |\lambda^{\text{ref}}_{\omega,\omega'} - A_{\omega,\omega'} \lambda|\leq C|\lambda|^{2}
\end{eqnarray}
with $A_{\omega,\omega'} = \sum_{x_{2}, y_{2}} \sum_{\rho,\rho'} \hat w_{\rho\rho'}(\underline{0}; x_{2}, y_{2}) |\xi^{\omega}_{\rho}(k_{F}^{\omega}; x_{2})|^{2} |\xi^{\omega'}_{\rho'}(k_{F}^{\omega'}; y_{2})|^{2}$, such that the following holds.
\vspace{.2cm}

\noindent{\underline{\emph{Lattice $2$-point function.}}} There exists $Q^{\text{ref}}_{\omega}(x_{2})$ satisfying, for all $n\in \mathbb{N}$,
\begin{equation}\label{eq:2pterr0}
| Q^{\text{ref}}_{\omega,\rho}(x_{2}) |\leq \frac{C_{n}}{1 + x_{2}^{n}}\;,\qquad \| Q^{\text{ref}}_{\omega} \|_{2} = 1 + O(\lambda)
\end{equation}
such that, for $\|\underline{k}'\|$ small enough:
\begin{equation}\label{eq:2pt2ref}
S^{\beta, L}_{2; \rho, \rho'}(\underline{k}' + \underline{k}_{F}^{\omega}; x_{2}, y_{2}) = S_{2; \omega}^{\text{ref}} (\underline{k}') Q^{\text{ref}}_{\omega,\rho}(x_{2}) \overline{Q^{\text{ref}}_{\omega,\rho'}(y_{2})} + R^{\beta, L}_{2;\omega, \rho, \rho'}(\underline{k}'; x_{2}, y_{2})\;.
\end{equation}
The error term satisfies the bound, for $0<\theta<1$:
\begin{equation}\label{eq:2pterr}
\Big| R^{\beta, L}_{2;\omega, \rho, \rho'}(\underline{k}'; x_{2}, y_{2}) \Big| \leq \frac{C_{n} \|\underline{k}'\|_{\omega}^{\theta-1}}{1 + |x_{2} - y_{2}|^{n}}\;.
\end{equation}
\vspace{.2cm}

\noindent{\underline{\emph{Lattice vertex function.}}} There exists $Z^{\text{ref}}_{\mu}(z_{2})$ satisfying, for all $n\in \mathbb{N}$, 
\begin{equation}\label{eq:Zdec0}
| Z^{\text{ref}}_{\mu}(z_{2}) | \leq \frac{C_{n}}{1 + z_{2}^{n}}\;,
\end{equation}
such that, for $\|\underline{k}'\|$ small enough, and $\|\underline{k}' + \underline{p}\| = \|\underline{k}'\|$:
\begin{eqnarray}\label{eq:vertex2ref}
&&S^{\beta, L}_{1,2;\mu,\omega,\rho,\rho'}(\underline{p}, \underline{k}'; x_{2}, y_{2}, z_{2}) \nonumber\\
&&= \sum_{\omega'}^{*} Z^{\text{ref}}_{\mu,\omega'}(z_{2}) S^{\text{ref}}_{1,2;\omega',\omega}(\underline{p}, \underline{k}') Q^{\text{ref}}_{\omega,\rho}(x_{2}) \overline{Q^{\text{ref}}_{\omega,\rho'}(y_{2})} + R^{\beta, L}_{\mu,\omega,\rho,\rho'}(\underline{p}, \underline{k}'; x_{2}, y_{2}, z_{2})\;.
\end{eqnarray}
We recall that the asterisk denotes summation over the edge modes localized around the $x_{2} = 0$ edge. The error term satisfies the bound, for $0<\theta<1$:
\begin{equation}\label{eq:Zdec}
\Big| R^{\beta, L}_{\mu,\omega,\rho,\rho'}(\underline{p}, \underline{k}'; x_{2}, y_{2}, z_{2}) \Big| \leq \frac{C_{n} \|\underline{k}'\|^{\theta-2}}{1 + d(x_{2}, y_{2}, z_{2})^{n}}
\end{equation}
with $d(x_{2}, y_{2}, z_{2})$ the tree distance for $(x_{2}, y_{2}, z_{2})$.
\vspace{.2cm}

\noindent{\underline{\emph{Lattice current-current correlation.}}} For $\|\underline{p}\|$ small enough:
\begin{equation}\label{eq:JJ2ref}
S^{\beta, L}_{0,2;\mu,\nu}(\underline{p}; x_{2}, y_{2}) = \sum^{*}_{\omega, \omega'} Z^{\text{ref}}_{\mu,\omega}(x_{2}) Z^{\text{ref}}_{\nu,\omega'}(y_{2}) S_{0,2;\omega,\omega'}^{\text{ref}}(\underline{p}) + R^{\beta,L}_{0,2;\mu,\nu}(\underline{p}; x_{2}, y_{2})\;,
\end{equation}
where:
\begin{equation}\label{eq:contR}
| R^{\beta,L}_{0,2;\mu,\nu}(\underline{p}; x_{2}, y_{2}) |\leq \frac{C_{n}}{1 + |x_{2} - y_{2}|^{n}}\;,\quad \sum_{y_{2}} | R^{\beta,L}_{0,2;\mu,\nu}(\underline{p}; x_{2}, y_{2}) - R^{\beta,L}_{0,2;\mu,\nu}(\underline{0}; x_{2}, y_{2})  | \leq C\|\underline{p}\|^{\theta}\;.
\end{equation}
All bounds are uniform in $\beta, L$, and hold in the $\beta, L\to \infty$ limit.
\end{prop}
Proposition \ref{prp:rel} has been proved in \cite{AMP}, see Proposition 6.1 there, in the case of lattice models with one chiral edge state (up to spin degeneracy). The prerequisite of the proof is the construction of the reference model, which has been discussed in the previous sections. With this construction at hand, the proof of \cite{AMP} applies to the present setting.

\section{Proof of Theorem \ref{thm:main}}\label{sec:proof}

In this section we will prove the universality of edge charge transport in interacting topological insulators, Theorem \ref{thm:main}. We start by deriving a convenient expression for the lattice current-current correlation function.
\paragraph{Computing the lattice current-current correlation.} As a consequence of Proposition \ref{prp:rel} and of the lattice Ward identities, we can express the zero-momentum limit of the edge current-current correlation functions in terms of the parameters 
\begin{equation}
Z^{\text{ref}}_{\mu,\omega} = \sum_{x=0}^{\infty} Z^{\text{ref}}_{\mu,\omega}(x_{2})\;.
\end{equation}
Recall that, after Wick rotation, Eq. (\ref{eq:G2S}):
\begin{equation}
G^{\underline{a}}(\underline{p}) := \sum_{x_{2} = 0}^{a} \sum_{y_{2} = 0}^{a'} S_{0,2;0,1}(\underline{p}; x_{2}, y_{2})\;.
\end{equation}
It is convenient to write:
\begin{equation}
\sum_{x_{2} = 0}^{a} S_{0,2;0,1}(\underline{p}; x_{2}, y_{2}) = \sum_{x_{2} = 0}^{a} S^{\beta, L}_{0,2;0,1}(\underline{p}; x_{2}, y_{2}) + \frak{e}_{1}(\beta, L)\;,
\end{equation}
where $\frak{e}_{1}(\beta, L)\to 0$ as $\beta, L \to \infty$. Then, using the fast decay of correlations in the $y_{2}$ direction, Proposition \ref{prp:rel}, we write:
\begin{equation}
\sum_{x_{2} = 0}^{a} S^{\beta, L}_{0,2;0,1}(\underline{p}; x_{2}, y_{2}) = \sum_{x_{2} = 0}^{L-1} S^{\beta, L}_{0,2;0,1}(\underline{p}; x_{2}, y_{2}) + \frak{e}_{2}(y_{2}, a)\;,
\end{equation}
where, uniformly in $\underline{p}$, recalling that $y_{2} \leq a' \ll a$:
\begin{equation}
| \frak{e}_{2}(y_{2}, a) | \leq \frac{C_{n}}{1 + | y_{2} - a |^{n}}\;.
\end{equation}
Hence,
\begin{eqnarray}\label{eq:aa'2a}
G^{\underline{a}}(\underline{p}) &=& G^{a'}(\underline{p}) + O(|a - a'|^{-n}) \nonumber\\
G^{a'}(\underline{p}) &:=& \sum_{y_{2} = 0}^{a'} \lim_{\beta\to \infty} \lim_{L\to \infty} \sum_{x_{2} = 0}^{L-1} S^{\beta, L}_{0,2;0,1}(\underline{p}; x_{2}, y_{2})\;.
\end{eqnarray}
Next, using the representation (\ref{eq:JJ2ref}), we rewrite:
\begin{eqnarray}\label{eq:dec1}
G^{a'}(\underline{p}) &=& \sum^{*}_{\omega, \omega'} Z^{\text{ref}}_{0,\omega} Z^{\text{ref}}_{1,\omega'} S_{0,2;\omega,\omega'}^{\text{ref}}(\underline{p}) \nonumber\\
&& - \sum_{\omega, \omega'}^{*} \sum_{y_{2} \geq a' + 1}^{\infty} Z^{\text{ref}}_{1,\omega'}(y_{2}) Z^{\text{ref}}_{0,\omega} S_{0,2;\omega,\omega'}^{\text{ref}}(\underline{p}) \nonumber\\
&& + \sum_{y_{2} = 0}^{a'} \sum_{x_{2}= 0}^{\infty} R_{0,2;0,1}(\underline{p}; x_{2}, y_{2})\;.
\end{eqnarray}
Thanks to (\ref{eq:Zdec0}), the second term in (\ref{eq:dec1}) is smaller than any power in $a'$:
\begin{eqnarray}
E_{1}^{a'}(\underline{p}) &:=& -\sum_{\omega, \omega'}^{*} \sum_{y_{2} \geq a' + 1}^{\infty} Z^{\text{ref}}_{1,\omega'}(y_{2}) Z^{\text{ref}}_{0,\omega} S_{0,2;\omega,\omega'}^{\text{ref}}(\underline{p})\nonumber\\
| E_{1}^{a'}(\underline{p}) | &\leq& \frac{C_{n}}{1 + a'^{n}}\;.
\end{eqnarray}
Now, thanks to the lattice Ward identity \ref{eq:S020}, we know that:
\begin{equation}
G^{a'}((p_{0}, 0)) = 0\;.
\end{equation}
This identity, together with the representation (\ref{eq:dec1}), gives:
\begin{equation}\label{eq:dec2}
0 = \sum^{*}_{\omega, \omega'} Z^{\text{ref}}_{0,\omega} Z^{\text{ref}}_{1,\omega'} S_{0,2;\omega,\omega'}^{\text{ref}}((p_{0}, 0)) + E_{1}^{a'}((p_{0}, 0)) + \sum_{y_{2} = 0}^{a'} \sum_{x_{2} = 0}^{\infty}R_{0,2;0,1}((p_{0}, 0); x_{2}, y_{2})\;.
\end{equation}
Thus, by the continuity of the last term in (\ref{eq:dec2}), see the last bound in (\ref{eq:contR}):
\begin{eqnarray}\label{eq:Rdec}
\sum_{y_{2} = 0}^{a'} \sum_{x_{2}= 0}^{\infty} R_{0,2;0,1}(\underline{p}; x_{2}, y_{2}) &=& \sum_{y_{2} = 0}^{a'} \sum_{x_{2} = 0}^{\infty}R_{0,2;0,1}((p_{0}, 0); x_{2}, y_{2}) + E^{a'}_{2}(\underline{p}) \\
&=& - \sum_{\omega, \omega'}^{*} Z^{\text{ref}}_{0,\omega} Z^{\text{ref}}_{1,\omega'} S_{0,2;\omega,\omega'}^{\text{ref}}((p_{0}, 0)) - E_{1}^{a'}((p_{0}, 0)) + E^{a'}_{2}(\underline{p})\nonumber
\end{eqnarray}
where $|E^{a'}_{2}(\underline{p})| \leq a' \|\underline{p}\|^{\theta}$ and where in the last identity we used (\ref{eq:dec2}). Defining $E_{3}^{a'}(\underline{p}) := - E_{1}^{a'}((p_{0}, 0)) + E^{a'}_{2}(\underline{p})$, and plugging (\ref{eq:Rdec}) into (\ref{eq:dec1}), we get:
\begin{equation}\label{eq:a'exp}
G^{a'}(\underline{p}) = \sum^{*}_{\omega, \omega'} Z^{\text{ref}}_{0,\omega} Z^{\text{ref}}_{1,\omega'} \Big(S_{0,2;\omega,\omega'}^{\text{ref}}(\underline{p}) - S_{0,2;\omega,\omega'}^{\text{ref}}((p_{0},0))\Big) + E^{a'}_{3}(\underline{p})\;.
\end{equation}
Recall the definition of edge charge conductance:
\begin{equation}
G = \lim_{a' \to \infty} \lim_{a\to \infty} \lim_{p_{1}\to 0} \lim_{p_{0}\to 0} G^{\underline{a}}(\underline{p})\;.
\end{equation}
From (\ref{eq:aa'2a}), (\ref{eq:a'exp}) we obtain the following remarkable expression:
\begin{equation}
G = (\vec Z^{\text{ref}}_{0}, \mathcal{A} \vec Z^{\text{ref}}_{1})\;,
\end{equation}
where $(\vec f,\vec g) = \sum^{*}_{\omega} f_{\omega} g_{\omega}$, and the matrix $\mathcal{A}$ detects the discontinuity at zero of the density-density correlation of the reference model:
\begin{equation}
\mathcal{A}_{\omega,\omega'} =  \lim_{p_{1}\to 0} \lim_{p_{0}\to 0} S_{0,2;\omega,\omega'}^{\text{ref}}(\underline{p}) - \lim_{p_{0}\to 0} \lim_{p_{1}\to 0} S_{0,2;\omega,\omega'}^{\text{ref}}(\underline{p})\;.
\end{equation}
At this point, notice that if we would have interchanged the $\eta \to 0^{+}$, $p\to 0$ limits in the definition of edge conductance, we would have obtained a trivial result. The matrix $A$ can be computed explicitly, starting from the expression in (\ref{eq:ddexp}) for the density-density correlation of the reference model. We have:
\begin{eqnarray}
\lim_{p_{0}\to 0} \lim_{p_{1}\to 0} S_{0,2}^{\text{ref}}(\underline{p}) &=& \lim_{p_{0}\to 0} \lim_{p_{1}\to 0}\Big(T(\underline{p}) \frac{1}{(Z^{\text{ref}})^{2}} \frac{\frak{B}(\underline{p})}{D(\underline{p})}\Big) \nonumber\\
&=& \frac{-1}{1 - (4\pi |v^{\text{ref}}|)^{-1} \Lambda_{Z}} \frac{1}{4\pi |v^{\text{ref}}|} \frac{1}{(Z^{\text{ref}})^2}\;,
\end{eqnarray}
where we used that $\lim_{p_{0}\to 0} \lim_{p_{1}\to 0} D(\underline{p})^{-1} \frak{B}(\underline{p}) = -(4\pi |v^{\text{ref}}|)^{-1}$, with $|v^{\text{ref}}|$ the diagonal matrix with entries $|v^{\text{ref}}_{\omega}|$ (recall the expression (\ref{eq:bubble}) for the bubble diagram $\frak{B}_{\omega}(\underline{p})$). Similarly, using that $\lim_{p_{1}\to 0} \lim_{p_{0}\to 0} D(\underline{p})^{-1} \frak{B}(\underline{p}) = (4\pi |v^{\text{ref}}|)^{-1}$:
\begin{eqnarray}
\lim_{p_{1}\to 0} \lim_{p_{0}\to 0} S_{0,2}^{\text{ref}}(\underline{p}) = \frac{1}{1 + (4\pi |v^{\text{ref}}|)^{-1} \Lambda_{Z}} \frac{1}{4\pi |v^{\text{ref}}|} \frac{1}{(Z^{\text{ref}})^{2}}\;.
\end{eqnarray}
In conclusion, we get:
\begin{eqnarray}
\mathcal{A} &=& \frac{1}{1 + (4\pi |v^{\text{ref}}|)^{-1} \Lambda_{Z}} \frac{1}{4\pi |v^{\text{ref}}|} \frac{1}{(Z^{\text{ref}})^{2}} + \frac{1}{1 - (4\pi |v^{\text{ref}}|)^{-1} \Lambda_{Z}} \frac{1}{4\pi |v^{\text{ref}}|} \frac{1}{(Z^{\text{ref}})^{2}} \nonumber\\
&=& \frac{1}{1 + (4\pi |v^{\text{ref}}|)^{-1} \Lambda_{Z}} \frac{1}{1 - (4\pi |v^{\text{ref}}|)^{-1} \Lambda_{Z}}  \frac{1}{2\pi |v^{\text{ref}}|} \frac{1}{(Z^{\text{ref}})^{2}}\;.
\end{eqnarray}
This allows to rewrite the edge conductance as:
\begin{equation}\label{eq:Graw}
G = \Big(\vec Z^{\text{ref}}_{0}, \frac{1}{1 + (4\pi |v^{\text{ref}}|)^{-1} \Lambda_{Z}} \frac{1}{1 - (4\pi |v^{\text{ref}}|)^{-1} \Lambda_{Z}}  \frac{1}{2\pi |v^{\text{ref}}|} \frac{1}{(Z^{\text{ref}})^{2}} \vec Z^{\text{ref}}_{1} \Big)\;.
\end{equation}
\paragraph{Computing the vertex renormalizations.} To gain insight on the vertex renormalizations $Z_{\mu,\omega}$, we start from the lattice vertex Ward identity, Eq. (\ref{eq:vertlat0}):
\begin{equation}\label{eq:latvertWI}
\sum_{\mu = 0,1} p_{\mu} \eta_{\mu}(\underline{p}) \sum_{x_{2} = 0}^{\infty} S_{1,2; \mu,\rho,\rho'}(\underline{k}, \underline{p}; x_{2}, y_{2}, z_{2}) = S_{2;\rho,\rho'}(\underline{k}; y_{2}, z_{2}) - S_{2;\rho,\rho'}(\underline{k} + \underline{p}; y_{2}, z_{2})\;,
\end{equation}
with $\eta_{0}(\underline{p}) = -i$ and $\eta_{1}(\underline{p}) = p_{1} + O(p_{1}^{2})$. By (\ref{eq:vertex2ref}), for $\underline{k} = \underline{k}_{F}^{\omega} + \underline{k}'$:
\begin{eqnarray}
&&\sum_{x_{2} = 0}^{\infty} S_{1,2; \mu,\rho,\rho'}(\underline{k}, \underline{p}; x_{2}, y_{2}, z_{2})  = \nonumber\\&&\quad \sum_{\omega'} Z^{\text{ref}}_{\mu,\omega'} S^{\text{ref}}_{1,2;\mu,\omega',\omega}(\underline{k}', \underline{p}) Q^{\text{ref}}_{\omega,\rho}(y_{2}) \overline{Q^{\text{ref}}_{\omega,\rho'}(z_{2})} + R_{1,2;\mu,\omega,\rho,\rho'}(\underline{k}',\underline{p}; y_{2}, z_{2})\;,
\end{eqnarray}
with:
\begin{equation}
R_{1,2;\mu,\omega,\rho,\rho'}(\underline{k}',\underline{p}; y_{2}, z_{2}) = \sum_{x_{2} = 0}^{\infty} R_{1,2;\mu,\omega,\rho,\rho'}(\underline{k}',\underline{p}; x_{2}, y_{2}, z_{2})\;.
\end{equation}
Also, recall (\ref{eq:2pt2ref}):
\begin{equation}
S_{2; \rho, \rho'}(\underline{k}' + \underline{k}_{F}^{\omega}; y_{2}, z_{2}) = S_{2; \omega}^{\text{ref}} (\underline{k}') Q^{\text{ref}}_{\omega,\rho}(y_{2}) \overline{Q^{\text{ref}}_{\omega,\rho'}(z_{2})} + R_{2;\omega, \rho, \rho'}(\underline{k}'; y_{2}, z_{2})\;.
\end{equation}
We then multiply left-hand side and right-hand side of (\ref{eq:latvertWI}) times $\overline{Q^{\text{ref}}_{\omega,\rho}(y_{2})} Q^{\text{ref}}_{\omega,\rho'}(z_{2})$, and sum over all $\rho,\rho'$ and $y_{2}, z_{2}$. We get:
\begin{eqnarray}\label{eq:WIlat2}
\sum_{\mu = 0,1} p_{\mu} \eta_{\mu}(\underline{p}) \Big[ \sum^{*}_{\omega'} Z_{\mu,\omega'} S^{\text{ref}}_{1,2; \omega',\omega}(\underline{p}, \underline{k}') + R_{1,2;\mu,\omega}(\underline{p}, \underline{k}') \Big] &=& S_{2; \omega}^{\text{ref}} (\underline{k}') - S_{2; \omega}^{\text{ref}} (\underline{k}' + \underline{p}) \\
&& + R_{2;\omega}(\underline{k}') - R_{2;\omega}(\underline{k}' + \underline{p})\;,\nonumber
\end{eqnarray}
where we introduced the notations:
\begin{eqnarray}
R_{1,2;\mu,\omega}(\underline{p}, \underline{k}') &:=& \frac{1}{\|Q^{\text{ref}}_{\omega}\|_{2}^{2}}\sum_{\rho,\rho'} \sum_{y_{2}, z_{2}} R_{1,2;\mu,\omega,\rho,\rho'}(\underline{p}, \underline{k}'; y_{2}, z_{2}) \overline{Q^{\text{ref}}_{\omega,\rho}(y_{2})} Q^{\text{ref}}_{\omega,\rho'}(z_{2}) \nonumber\\
R_{2;\omega}(\underline{k}') &:=& \frac{1}{\|Q^{\text{ref}}_{\omega}\|_{2}^{2}}\sum_{\rho,\rho'} \sum_{y_{2}, z_{2}} R_{2;\omega, \rho, \rho'}(\underline{k}'; y_{2}, z_{2}) \overline{Q^{\text{ref}}_{\omega,\rho}(y_{2})} Q^{\text{ref}}_{\omega,\rho'}(z_{2})\;,
\end{eqnarray}
recall that $\|Q^{\text{ref}}_{\omega}\|_{2} = 1 + O(\lambda)$, see (\ref{eq:2pterr}). Next, recall the explicit expression for the vertex function of the reference model, Eq. (\ref{eq:ddexp}):
\begin{equation}
S^{\text{ref}}_{1,2; \omega',\omega}(\underline{p}, \underline{k}') = T_{\omega',\omega}(\underline{p}) \frac{1}{Z^{\text{ref}}_{\omega} D_{\omega}(\underline{p})} ( S_{2; \omega}^{\text{ref}} (\underline{k}') - S_{2; \omega}^{\text{ref}} (\underline{k}'  + \underline{p} )\big)\;.
\end{equation}
Recall the estimates (\ref{eq:2pterr}), (\ref{eq:Zdec}) for the error terms. Choosing $\|\underline{k}'\| = \kappa$, $\|\underline{k}' + \underline{p}\| = \kappa$, $\|\underline{p}\| = O(\kappa)$, we see that in (\ref{eq:WIlat2}) the $R$-terms give contributions that are suppressed by a factor $\kappa^{\theta}$ with respect to the other quantities. Hence, we get:
\begin{equation}\label{eq:divide}
\sum_{\omega'}^{*} D^{Z}_{\omega'}(\underline{p}) T_{\omega',\omega}(\underline{p}) \frac{1}{Z^{\text{ref}}_{\omega} D_{\omega}(\underline{p})}( S_{2; \omega}^{\text{ref}} (\underline{k}') - S_{2; \omega}^{\text{ref}} (\underline{k}'  + \underline{p} )\big) = S_{2; \omega}^{\text{ref}} (\underline{k}') - S_{2; \omega}^{\text{ref}} (\underline{k}' + \underline{p}) + O(\kappa^{\theta - 1})
\end{equation}
with 
\begin{equation}
D^{Z}_{\omega'}(\underline{p}) := -ip_{0}Z^{\text{ref}}_{0,\omega'} + p_{1} Z^{\text{ref}}_{1,\omega'}\;.
\end{equation}
Dividing both sides of (\ref{eq:divide}) by $( S_{2; \omega}^{\text{ref}} (\underline{k}') - S_{2; \omega}^{\text{ref}} (\underline{k}'  + \underline{p} )\big) (Z^{\text{ref}}_{\omega} D_{\omega}(\underline{p}))^{-1} = O(\kappa^{-2})$, we find:
\begin{equation}\label{eq:DZD}
\sum_{\omega'}^{*}  D^{Z}_{\omega'}(\underline{p}) T_{\omega',\omega}(\underline{p}) = Z^{\text{ref}}_{\omega} D_{\omega}(\underline{p}) + O(\kappa^{\theta+1})\;.
\end{equation}
Recall the expression (\ref{eq:Tdef}) for $T(\underline{p})$. We have:
\begin{equation}
\lim_{p_{0}\to 0} \lim_{p_{1}\to 0} T(\underline{p}) = \frac{1}{1-(4\pi|v^{\text{ref}}|)^{-1}\Lambda_{Z}}\;,\qquad \lim_{p_{1}\to 0} \lim_{p_{0}\to 0} T(\underline{p})  = \frac{1}{1+(4\pi|v^{\text{ref}}|)^{-1}\Lambda_{Z}}\;.
\end{equation}
In conclusion, Eq. (\ref{eq:DZD}) implies:
\begin{equation}
\sum_{\omega'}^{*} Z^{\text{ref}}_{0,\omega'} \Big(\frac{1}{1-(4\pi|v^{\text{ref}}|)^{-1}\Lambda_{Z}}\Big)_{\omega',\omega} = Z^{\text{ref}}_{\omega}\;,\qquad \sum^{*}_{\omega'} Z^{\text{ref}}_{1,\omega'} \Big(\frac{1}{1+(4\pi|v^{\text{ref}}|)^{-1}\Lambda_{Z}}\Big)_{\omega',\omega} = Z^{\text{ref}}_{\omega} v^{\text{ref}}_{\omega}\;,
\end{equation}
which we can also rewrite as:
\begin{equation}\label{eq:Z0Z1}
\vec Z^{\text{ref}}_{0} = \big( 1 - \Lambda_{Z}^{T} (4\pi|v^{\text{ref}}|)^{-1} \big) \vec Z^{\text{ref}}\;,\qquad \vec Z^{\text{ref}}_{1} = \big( 1 + \Lambda_{Z}^{T} (4\pi|v^{\text{ref}}|)^{-1} \big) v^{\text{ref}}\vec Z^{\text{ref}}\;.
\end{equation}
\paragraph{Conclusion: universality of edge charge conductance.} Finally, let us plug the identities (\ref{eq:Z0Z1}) into (\ref{eq:Graw}). We get, omitting temporarily all `ref' labels:
\begin{eqnarray}
G &=& \Big( \vec Z, ( 1 - (4\pi|v|)^{-1} \Lambda_{Z}  ) \frac{1}{1 + (4\pi |v|)^{-1} \Lambda_{Z}} \frac{1}{1 - (4\pi |v|)^{-1} \Lambda_{Z}}  \frac{1}{2\pi |v|} \frac{1}{Z^{2}} ( 1 + \Lambda_{Z}^{T} (4\pi|v|)^{-1} ) v \vec Z\Big)\nonumber\\
&\equiv& \Big( \vec Z, \frac{1}{1 + (4\pi |v|)^{-1} \Lambda_{Z}}  \frac{1}{2\pi |v|} \frac{1}{Z^{2}} ( 1 + \Lambda_{Z}^{T} (4\pi|v|)^{-1} ) v \vec Z\Big)\;.
\end{eqnarray}
Now, using that, recalling the definition (\ref{eq:LZdef}) of $\Lambda_{Z}$:
\begin{eqnarray}
\frac{1}{2 \pi |v|} \frac{1}{Z^{2}} \Big(1 + \Lambda_{Z}^{T}\frac{1}{4\pi |v|} \Big) &=& \frac{1}{2 \pi |v|} \frac{1}{Z^{2}} \Big(1 + Z \Lambda \frac{1}{Z} \frac{1}{4\pi |v|} \Big) \nonumber\\
&=&  \Big(1 +  \frac{1}{4\pi |v|}\frac{1}{Z} {\Lambda} {Z} \Big) \frac{1}{2\pi |v|} \frac{1}{{Z}^{2}}\nonumber\\
&\equiv& \Big(1 +  \frac{1}{4\pi |v|}  {\Lambda}_{Z} \Big) \frac{1}{2\pi |v|} \frac{1}{{Z}^{2}}\;,
\end{eqnarray}
we immediately get, reintroducing the `ref' labels:
\begin{eqnarray}
G &=& \Big( \vec Z^{\text{ref}}, \frac{1}{2\pi |v^{\text{ref}}|} \frac{1}{(Z^{\text{ref}})^{2}} v^{\text{ref}} \vec Z^{\text{ref}}\Big) \nonumber\\
&=& \sum^{*}_{\omega} \frac{\text{sgn}(v^{\text{ref}}_{\omega})}{2\pi} \nonumber\\
&=& \sum^{*}_{\omega} \frac{\text{sgn}(v_{\omega})}{2\pi}
\end{eqnarray}
where we used that $|v^{\text{ref}}_{\omega} - v_{\omega}|\leq C|\lambda|$. This concludes the proof of Theorem \ref{thm:main}.\qed

\paragraph{Acknowledgements.} V.M. acknowledges financial support from MIUR, PRIN 2017 for the project MaQuMA, PRIN201719VMAST01. The work of V.M. has been supported by the European Research Council (ERC)
under the European Union’s Horizon 2020 research and innovation program ERC CoG UniCoSM, grant agreement n.724939. The work of M.P. has been supported by the European Research Council (ERC) under the European Union’s Horizon 2020 research and innovation program ERC StG MaMBoQ, n.802901. This work has been carried out under the auspices of the GNFM of INdAM. We thank Christian Sp\aa nsl\"att  Rugarn for interesting comments.

\end{document}